\documentclass[11pt]{article}

\usepackage{amsthm}

\usepackage{amsmath, amssymb,amsfonts,dsfont,xspace,graphicx,relsize,bm,mathtools,xcolor}
\usepackage{tcolorbox,thmtools}
\usepackage{mathrsfs}
\usepackage[pagebackref]{hyperref}
\usepackage{enumerate}
\usepackage{cleveref}
\usepackage{verbatim}
\usepackage[noend]{algpseudocode}
\usepackage{adjustbox}

\usepackage{float}

\definecolor{OliveGreen}{RGB}{60,76,36}
\def\01{\{0,1\}}

\newcommand{\ket}[1]{|#1\rangle}
\newcommand{\bra}[1]{\langle#1|}
\newcommand{\ketbra}[2]{|#1\rangle\langle#2|}

\newcommand{\diag}{\mbox{\rm diag}}

\newcommand{\normstate}[1]{\Big\lVert#1\Big\rVert}
\newcommand{\norm}[1]{\mbox{$\parallel{#1}\parallel$}}

\newcommand{\Cc}{{\mathcal C}} 
\newcommand{\calA}{{\mathcal A}}
\newcommand{\calB}{{\mathcal B}}

\newcommand{\calL}{{\mathcal L}}
\newcommand{\SWAP}{{\textsf{SWAP}} }

\newcommand{\FF}{{\{0,1\}}}
\newcommand{\Exp}{\mathop{\mathbb{E}}}

\newcommand{\Sh}{\ensuremath{\mathcal{S}}}

\newcommand{\PAC}{\ensuremath{\textsf{PAC}}}
\newcommand{\DNF}{\ensuremath{\textsf{DNF}}}
\newcommand{\OR}{\ensuremath{\textsf{OR}}}

\newcommand{\opt}{\ensuremath{\textsf{opt}}}

\newcommand{\A}{\ensuremath{\mathcal{A}}}
\newcommand{\F}{\ensuremath{\mathbb{F}}}

\newcommand{\R}{\ensuremath{\mathbb{R}}}

\newcommand{\id}{\ensuremath{\mathbb{I}}}

\newcommand{\AND}{\ensuremath{\textsf{AND}}}
\newcommand{\NOT}{\ensuremath{\textsf{NOT}}}
\newcommand{\AC}{\ensuremath{\textsf{AC}}}
\newcommand{\TAC}{\ensuremath{\textsf{TAC}}}
\newcommand{\Jun}{\ensuremath{\textsf{Jun}}}
\newcommand{\DT}{\ensuremath{\textsf{DT}}}
\newcommand{\Had}{\ensuremath{\textsf{Had}}}
\newcommand{\SOTA}{\ensuremath{\textsf{SOTA}}}

\DeclareMathOperator{\poly}{poly}
 
\usepackage{multirow}
\usepackage{multicol}
\newcommand{\calF}{{\cal F }}

\newcommand{\calC}{{\cal C }}
\newcommand{\calS}{{\cal S }}

\def\01{\{0,1\}}

\newcommand{\supp}{\mathsf{supp}}
\newcommand{\Par}{\mathsf{Par}}
\newcommand{\WAL}{\mathsf{WAL}}

\newcommand{\spann}{\mathsf{span}}

\DeclareMathOperator{\sign}{sign}

\usepackage{float}

\newtheorem{theorem}{Theorem}[section]
\newtheorem{definition}[theorem]{Definition}

\newtheorem{fact}[theorem]{Fact}
\newtheorem{lemma}[theorem]{Lemma}

\newtheorem{corollary}[theorem]{Corollary}

\newtheorem{claim}[theorem]{Claim}
\newtheorem{result}[theorem]{Result}

\usepackage[linesnumbered,ruled,vlined]{algorithm2e}
\SetKwInput{KwInput}{Input}
\SetKwInput{KwOutput}{Output}
\SetKwInOut{Promise}{Promise}
\SetKwInput{Goal}{Goal}
\SetKwProg{Fn}{function}{}{}
\SetKwFor{RepTimes}{repeat}{times}{}
\SetKwFunction{Ver}{Verify}
\SetKwFunction{Prep}{PrepareState}
\SetKwComment{Comment}{/* }{ */}

\newtcolorbox{myalgorithm}[1][]{
    colback=gray!10, 
    colframe=black, 
    arc=5pt, 
    boxrule=0.5pt, 
    left=0pt, right=0pt, top=0pt, bottom=0pt 
}
\usepackage{setspace}
\usepackage{makecell}

\SetKwInput{KwInput}{Input}
\SetKwInput{KwOutput}{Output}
\SetKwInOut{Promise}{Promise}
\SetKwInput{Goal}{Goal}
\SetKwProg{Fn}{function}{}{}
\SetKwFor{RepTimes}{repeat}{times}{}
\SetKwFunction{Ver}{Verify}
\SetKwFunction{Prep}{PrepareState}

\usepackage{thm-restate} 
\newcommand{\customlabel}[2]{%
\protected@write \@auxout {}{\string \newlabel {#1}{{#2}{}}}}
\makeatother
\global\long\def\argmin{\operatornamewithlimits{argmin}}

\usepackage{kpfonts}

\usepackage[margin=1in]{geometry}
\hypersetup{
    colorlinks,
    linkcolor={blue!100!black},
    citecolor={blue!100!black},
}

\begin{document}

\sloppy 
	
\title{Learning depth-3 circuits\\ via quantum agnostic boosting}

\author{ 
Srinivasan \\Arunachalam\thanks{IBM Quantum, Almaden Research Center. \href{mailto:Srinivasan.Arunachalam@ibm.com}{Srinivasan.Arunachalam@ibm.com}}
\and
Arkopal\\ Dutt \thanks{IBM Quantum, Cambridge, MA. \href{arkopal@ibm.com}{arkopal@ibm.com} } 
\and 
Alexandru \\ Gheorghiu
\thanks{IBM Quantum, Cambridge, MA. \href{agheorghiu@ibm.com}{agheorghiu@ibm.com} } 
\and
Michael de\\  Oliveira
\thanks{
International Iberian Nanotechnology Laboratory; LIP6, Sorbonne Universite; INESC TEC. \href{michael.oliveira@inl.int}{michael.oliveira@inl.int}} }

\maketitle

\begin{abstract}
We initiate the study of \emph{quantum agnostic learning} of phase states with respect to a function class $\mathcal{C}\subseteq \{c:\{0,1\}^n\rightarrow \{0,1\}\}$: given copies of an unknown $n$-qubit state $|\psi\rangle$ which has fidelity $\textsf{opt}$ with a phase state $|\phi_c\rangle=\frac{1}{\sqrt{2^n}}\sum_{x\in \{0,1\}^n}(-1)^{c(x)}|x\rangle$  for some $c\in \mathcal{C}$, output  $|\phi\rangle$ which has fidelity $|\langle \phi | \psi \rangle|^2 \geq \textsf{opt}-\varepsilon$. To this end, we give agnostic learning protocols for the following classes:
\begin{enumerate}
    \item Size-$t$ decision trees which runs in time $\textsf{poly}(n,t,1/\varepsilon)$. This also implies $k$-juntas can be agnostically learned in time $\textsf{poly}(n,2^k,1/\varepsilon)$.
    \item $s$-term DNF formulas in time  $\textsf{poly}(n,(s/\varepsilon)^{\log \log (s/\varepsilon) \cdot \log(1/\varepsilon)})$.
\end{enumerate}
    
Our main technical contribution is a \emph{quantum agnostic boosting} protocol which converts a ``weak'' agnostic learner, which outputs a parity state $|\phi\rangle$ such that $|\langle \phi|\psi\rangle|^2\geq \textsf{opt}/\textsf{poly}(n)$, into a ``strong'' learner which outputs a superposition of parity states $|\phi'\rangle$ such that $|\langle \phi'|\psi\rangle|^2\geq \textsf{opt} - \varepsilon$.  

\vspace{2mm}

Using quantum agnostic boosting, we give a $n^{O(\log(n/\varepsilon)\cdot \log \log n)}$-time algorithm for  $\varepsilon$-learning $\textsf{poly}(n)$-sized depth-$3$ circuits (consisting of $\textsf{AND}$, $\textsf{OR}$, $\textsf{NOT}$ gates) in the uniform $\textsf{PAC}$ model given quantum examples. Classically, obtaining an algorithm with a similar complexity has been an open question in the $\textsf{PAC}$ model and our work answers this given quantum~examples.  
\end{abstract}
\newpage 

\setcounter{tocdepth}{3}
{\small \tableofcontents}

\newpage

\section{Introduction}\label{sec:intro}

\paragraph{{Learning classical circuits.}} A central goal in computational learning theory is to determine which natural classes of Boolean functions can be efficiently learned, both in the \emph{Probably Approximately Correct} ($\PAC$) model and in the more challenging \emph{agnostic model}. Circuit classes of small depth provide a canonical test case since they are expressive enough to capture rich computational phenomena, yet structured enough that one might hope for efficient  algorithms. The seminal work of Linial, Mansour and Nisan~\cite{linial1993constant} first showed that depth-$d$ circuits on $n$-bit inputs, consisting of $\AND, \OR, \NOT$ gates (the class of $\AC^0$ circuits) are learnable in time $n^{O(\log^{d-1}n)}$ in the uniform $\PAC$ model with only  examples. While this provides a general guarantee, it leaves open the question of whether specialized \emph{efficient} algorithms exist for concrete depths of $d=2,3,4,5$? By the influential work of Naor and Reingold~\cite{naor2004number}, it is believed that  depth-$5$ circuits are hard to classically learn (\emph{assuming} factoring is hard).  
This naturally shifts the attention to depths $2,3,4$. 
In particular, the status of learning depth-$2$ circuits has been a longstanding 30-year old open question, with the best-known algorithm in the $\PAC$ model running in time $n^{O(\log n)}$~\cite{verbeurgt1998learning}. 
Analogous to classical examples in the $\PAC$ model, Bshouty and Jackson~\cite{bshouty1995learning} introduced \emph{quantum examples} and the quantum $\PAC$ model and surprisingly showed that depth-$2$ circuits are learnable in \emph{quantum polynomial time}, thus giving a separation between quantum $\PAC$ and the state-of-the-art $(\SOTA)$ classical $\PAC$ learning. 
The natural next frontier in quantum learning theory is then
\begin{quote}
  \centering{ \emph{Can we learn depth-$3$ circuits in the quantum $\PAC$ model efficiently?}}
\end{quote}
Classically, a well-known idea to $\PAC$ learn depth-$3$ circuits is to consider the formulation of \emph{agnostic learning}. In particular, works of~\cite{feldman2009distribution,kalai2008agnostic} showcased how agnostic learning depth-$2$ circuits, in particular $\DNF$ formulas, could yield learning algorithms for depth-$3$ circuits. This naturally motivates the need to understand learning $\DNF$s in the  quantum agnostic learning  model.
\vspace{-9mm}
\paragraph{{Quantum agnostic learning.}} Tomography of quantum states (i.e., learning \emph{quantum states} in the noise-free model) has been well-studied in quantum computing. Quantum agnostic learning got traction only recently with the works of
Grewal et al.~\cite{grewal2024improved} and Chen et al.~\cite{chen2024stabilizer} that considered learning  stabilizer states.
Subsequently, a few works considered stabilizer product states~\cite{grewal2024agnostic}, high stabilizer-dimension states~\cite{chen2024stabilizer} and product states~\cite{bakshi2024learning}.
In the pursuit of agnostic learning $\DNF$s, 
we initiate the problem of agnostic learning \emph{phase states} (which is incomparable to the works mentioned and which we will discuss in more detail below). We now define the model of agnostic learning specialized to phase states. For the concept class $\Cc\subseteq \{c:\FF^n\rightarrow \FF\}$, we denote the phase state $\ket{\phi_c}$ corresponding to a function $c \in \Cc$ as
\begin{equation}
    \ket{\phi_c} = \frac{1}{\sqrt{2^n}} \sum_{x \in \FF^n} (-1)^{c(x)} \ket{x}.
\end{equation}
In quantum agnostic learning, an algorithm is given copies of an unknown $\ket{\psi}$ and the goal of an improper \emph{strong agnostic learner} is to output a $\ket{\phi'}$ (not necessarily a phase state) such that
$$
|\langle \psi|\phi'\rangle|^2\geq \opt-\varepsilon,
$$
where $\opt=\max_{c\in \Cc}|\langle \psi|\phi_c\rangle|^2$. The natural question at this point is: what classes of functions are agnostically learnable in this model? As far as we know, the agnostic learnability of classes such as parities, decision trees, $\DNF$s have not yet been considered, motivating the question 
\begin{quote}
  \centering{ \emph{Can we agnostically learn phase states corresponding to $\DNF$s efficiently?}}
\end{quote}

\subsection{Main results}
In this work, we make progress on both of the questions highlighted above. For the tasks of agnostically learning $\poly(n)$-sized $\DNF$s, we give a $n^{O(\log \log n \cdot \log (1/\varepsilon))}$-time quantum algorithm and for $\PAC$ learning $\poly(n)$-sized depth-$3$ circuits, we give $n^{O(\log (n/\varepsilon)\cdot \log \log n  )}$ time quantum algorithm. Below, we will discuss these results, starting with our main technical contribution, \emph{quantum agnostic boosting}. Classically obtaining  similar complexities given only classical examples is an open question and we compare both classical and quantum complexities below. 

\paragraph{{Quantum agnostic boosting.}}
We denote the fidelity with respect to the concept class as
$$
\calF_{\Cc}(\ket{\psi})=\max_{c\in \Cc} |\langle \psi|\phi_c\rangle|^2.
$$
Similar to the definition of a strong agnostic learner above, we define a \emph{weak agnostic learner} as one that outputs a state $\ket{\phi''}$ such that
$$
|\langle \psi|\phi''\rangle|^2\geq P(\opt/n),
$$
for some function $P:\R^+\rightarrow \R^+$.
We say a learner is \emph{efficient} if the running time scales polynomially in the description size of the concept class  in the (inverse of the) error parameter $1/\varepsilon$. Finally, we remark that the $\PAC$ model is defined just as above, except that the unknown $\ket{\psi}$ is promised to be $\ket{\phi_c}$ for an (unknown) $c\in \Cc$, in which case $\opt=1$. With this we are ready to state our main~contributions.

We give a quantum boosting algorithm that, given access to a weak agnostic learner with a polynomial overhead, produces a strong agnostic learner for a concept class $\Cc$.  To this end, we first define the class $\Cc$ of \emph{parity functions}, i.e. $\chi_S(x) = \langle S , x \rangle$ for some $S \in \{0,1\}^n$. We refer to $\{\ket{\chi_S}=2^{-n/2}\sum_x \chi_S(x)\ket{x}:S\in \FF^n\}$ as \emph{parity states}.
We summarize this in the theorem below.
\begin{result}
Let $\Cc$ be a concept class and  $\ket{\psi}$ be an unknown $n$-qubit state such that $\calF_{\Cc}(\ket{\psi})=\opt$. For every $\tau\geq 0$, let $\calA_{\WAL}$ be a weak agnostic learner for $\Cc$, i.e., given copies of  $\ket{\varphi}$ with  $\calF_{\Cc}(\ket{\varphi}) \geq \tau$, outputs a parity $\ket{\chi_S}$ such that $|\langle \varphi |\chi_S\rangle|^2\geq P(\tau/n)$ in time $T_{\WAL}$.
Then, there is a strong agnostic learner for $\Cc$, i.e., given copies of $\ket{\psi}$, runs in time $\poly(n, T_{\WAL},1/\varepsilon,1/P(\varepsilon/n))$ and outputs $\ket{\widehat{\phi}}$ with
$
|\langle \psi|\widehat{\phi}\rangle|^2\geq \opt-\varepsilon.
$
\end{result}
We formally state this theorem in Section~\ref{sec:agnostic_boosting}.  To provide some context, classically, Freund and Schapire~\cite{schapire:boostfirst,freund:boostfirst,freund:boosting} proposed the boosting algorithm called \emph{AdaBoost} that \emph{efficiently} uses a \emph{weak learner} as a black-box to construct a \emph{strong learner} in the usual $\PAC$ model. The AdaBoost algorithm by Freund and Schapire was one of the first few theoretical boosting algorithms that was simple enough to be extremely useful and successful in practice~\cite{schapire:foundations}. Similarly, agnostic boosting has been considered by works of~\cite{ben2001agnostic,kalai2008agnostic,feldman2009distribution} wherein they show similar results to boost weak agnostic learners to strong ones.

In the quantum setting, there are only a handful of works that have used boosting in the $\PAC$ model~\cite{bshouty1995learning,arunachalam2020quantum,izdebski2020improved}, and ours is the first work that demonstrates how to perform boosting in the harder model of \emph{agnostic  learning}.\footnote{We remark that~\cite{chatterjee2024efficient} also discusses an agnostic quantum booster but their model assumes that the unknown quantum state is a \emph{function state} whereas we make no assumption on our input state.} 
Apart from the context of learning Boolean functions, we believe that our quantum boosting algorithm could have utility for learning more general quantum states. Recently,~\cite{ad2025structure} utilized boosting on top of a weak agnostic learner for stabilizer states to give tomography protocols of states promised to have structured stabilizer~decompositions.

\paragraph{{Agnostic learning decision trees, juntas and DNFs.}} Our second contribution involves applying the quantum boosting algorithm on top of quantum weak agnostic learners for interesting concept classes. In particular, we first observe that (strong) agnostic learning of parity states can be done efficiently. We then exploit properties of different classes -- size-$t$ decision trees, $k$-juntas and $s$-term $\DNF$ formulas\footnote{We refer the reader to Section~\ref{subsect:conceptclasses} for a definition of these classes.}-- in terms of their Fourier spectrum, which allow us to obtain agnostic learning algorithms for these concept classes as well, all based on quantum agnostic boosting. These quantum agnostic learners in particular use weak agnostic learners that output parity states. This result is summarized below.

\begin{result}
Size-$t$ decision trees and $k$-juntas are learnable in time $\poly(n,t,1/\varepsilon)$~and $\poly(n,2^k,1/\varepsilon)$ respectively.    
Similarly, $s$-term $\DNF$ formulas are learnable in time~$\poly(n,(s/\varepsilon)^{\log \log s/\varepsilon \cdot \log (1/\varepsilon)})$. 
\end{result}
We formally state this theorem for each class in Section~\ref{sec:learning_algos}. In order to compare with classical results and for ease of exposition, let us consider the parameters above to take values of $2^k,t,s,1/\varepsilon=\poly(n)$ (which in general are considered the most interesting settings of these parameters for learning). Classically, the $\SOTA$ algorithm for learning decision trees in the agnostic model (without membership queries) runs in $n^{O(\log n)}$~\cite{kalai2008agnostically} time and with membership queries scales as $\poly(n)$~\cite{gopalan2008agnostically}.\footnote{By membership queries we mean that an algorithm can query an unknown $f:\FF^n\rightarrow \FF$ on an $x$ of its choice. In the (uniform) agnostic model, by membership query we mean the following: for every distribution $D:\FF^{n+1}\rightarrow [0,1]$ whose marginal on the first $n$ bits is uniform, the learner can query $x$ and obtain a sample from $b\sim D(x,\cdot)$.}
The same results hold for juntas since $k$-juntas have size-$2^k$ decision trees. In contrast, for $\DNF$ formulas, the $\SOTA$ algorithm in the agnostic model is $n^{O(\log^2 n)}$~\cite{kalai2008agnostically,gopalan2008query}. Our work shows that one can achieve an $n^{O(\log n)}$-time result with only quantum examples.
 
\paragraph{{$\PAC$-learning depth-$3$ circuits.}} Finally, using our quantum boosting algorithm, we also obtain a new quantum $\PAC$ learning algorithm for depth-$3$ circuits.
\begin{result}
The class of size-$s$ depth-$3$ circuits is quantum $\PAC$ learnable upto error $\varepsilon$ in time $\poly(n,(s/\varepsilon)^{\log (s/\varepsilon) \cdot \log \log (s/\varepsilon) })$.
\end{result}
We formally state this theorem in Section~\ref{sec:pac_learn_depth3}. Classically, it is a long-standing open problem to efficiently $\PAC$ learn \emph{depth-$2$} circuits of size $s$, wherein the state-of-the-art algorithm runs in time $n^{O(\log s)}$. In contrast, Bshouty and Jackson~\cite{bshouty1995learning} showed depth-$2$ circuits are quantumly learnable in time $\poly(n,s)$ and we further show that even depth-$3$ circuits  exhibit a separation between quantum $\PAC$ and classical $\PAC$~learning. Classically the $\SOTA$ algorithm  for learning $\poly(n)$-sized depth-$3$ circuits (without queries) scales as $n^{O(\log^2 n)}$  that comes via Fourier concentration bounds of Tal~\cite{tal2017tight}.  
We summarize our main contributions in the table~below.

{\renewcommand{\arraystretch}{2.5}\begin{table}[h]
\centering
\begin{tabular}{c | c  | c  |c|} 
\cline{2-4}
\multirow{2}{*}{}& \multicolumn{2}{c|}{\textbf{Classical}} & \multicolumn{1}{c|}{\textbf{Quantum}} \\
\cline{2-4}
& \emph{Membership queries} &\emph{Random examples} & \emph{Quantum examples} \\ [0.3ex] 
\hline
\multicolumn{1}{|c|}{Agnostic  $\DNF$} &   \makecell{$n^{\log (1/\varepsilon)\log \log n}$\\[0.3mm]\cite{gopalan2008query}}
  & \makecell{$n^{\log (1/\varepsilon) \log n}$\\[0.3mm]\cite{kalai2008agnostically}} &\makecell{$n^{\log (1/\varepsilon) \log \log n}$\\[0.3mm] \textbf{\emph{This work}}} \\ 
\hline
\multicolumn{1}{|c|}{$\PAC$ depth-$3$}   &\makecell{$n^{\log (n/\varepsilon) \log \log n}$\\[0.5mm]\cite{feldman2009distribution,kanade2009potential}} &\makecell{$n^{\log(n/\varepsilon)\log n}$\\[0.5mm]\cite{tal2017tight}} & \makecell{$n^{\log (n/\varepsilon)  \log \log n}$\\[0.5mm] \textbf{\emph{This work}}} \\ \hline
\end{tabular}
\caption{This gives a summary of the state-of-the-art results for agnostic learning $\poly(n)$-sized $\AC^0$ circuits of depth-$2$ (i.e., $\DNF$s) and $\PAC$ learning depth-$3$ circuits acting on $n$ bits.}
\end{table}
}
Apart from depth-$3$ circuits, we are able to learn $\TAC_2^0$, i.e., \emph{threshold of depth-$2$} circuits (consisting of $\AND, \OR, \NOT$ gates). This class of circuits is compelling for two reasons. First, even quantum-efficiently learning thresholds of \emph{threshold} gates would imply breakthrough classical circuit lower bounds (see~\cite{arunachalam2022quantum} for more details). Second, as observed in \cite{arunachalam2021quantum}, learning threshold circuits is equivalent to learning weights of feed-forward neural networks. Thus, quantum efficient learnability of threshold circuits would translate into a dramatic advantage over classical computers for training neural networks, echoing Aaronson's ``Ten Semi-Grand Challenges for Quantum Computing Theory"~\cite{Aaronson-blog}. Prior to our work, no results were known for $\PAC$ learning circuits consisting of any threshold gates; ours is the first to handle a single threshold on the~top.

\subsection{Prior works and concept challenges}
We first discuss a few potential approaches that do not yield efficient algorithms. These ideas will, in turn, motivate the need for new learning algorithms.

\subsubsection{In the quantum setting}
Often in quantum learning, the first class that one wants to learn is \emph{parities}. If one can learn these, the next step is to learn depth-$1$ circuits (i.e., disjunctions/conjunctions), juntas and then depth-$2$ circuits (i.e., $\DNF$ formulas) -- the key milestone en route to learning depth-$3$ circuits. Let us first discuss the applicability of recent algorithms to the problem of agnostic learning phase states corresponding to these concept classes. 

\vspace{1mm}

\textbf{\emph{1. Fourier/Bell sampling:}} The starting point of almost all quantum learning algorithms is 
Fourier sampling. Indeed, given copies of $\ket{\phi_f}=\frac{1}{\sqrt{2^n}}\sum_x f(x)\ket{x}$, one can simply apply the Hadamard transform and measure in order to sample from the Fourier  distribution $\{\widehat{f}(S)^2\}_S$. In the agnostic model, when we do not have copies of $\ket{\phi_f}$, but rather copies of $\ket{\psi}$ that are promised to be $\tau$-close\footnote{Throughout the paper, by ``$\tau$-close'', we mean that the two states have squared overlap at least $\tau$, i.e., $| \langle \psi | \phi_f \rangle |^2 \geq \tau.$} to $\ket{\phi_f}$, it is unclear whether Fourier sampling $\ket{\psi}$ would yield anything meaningful. It could very well be that $\ket{\phi_f}$ has one large Fourier coefficient $S$, but $\ket{\psi}$ puts $0$ amplitude on $\ket{S}$ and agrees with $\ket{\phi_f}$ elsewhere (so that it is $\tau$-close to $\ket{\phi_f}$). It is therefore unclear if one can strongly learn the unknown $f$ (i.e., learn $f$ up to $\opt-\varepsilon$) given copies of $\ket{\psi}$. 

Observe that Fourier sampling does work for parity states, denoted as $\Cc_\Par=\{\chi_S(x)=\langle S,x\rangle \text{ for all } S\in \FF^n\}$, since the Fourier spectrum of parity function $\chi_S$ is a point-mass on $\ket{S}$.  In particular, this implies that if  $\ket{\psi}$ is $\tau$-close to $\ket{\chi_S}$, then must have a ``large" amplitude on $\ket{S}$. Hence, strong learning via Fourier sampling can be achieved using $O(1/\tau^2)$ samples. But this no longer holds for other classes of states and, in particular, $\DNF$s for which the Fourier spectrum does not enjoy such a ``point-concentration'' property. 
One can instead consider \emph{Bell sampling}, which has been used in~\cite{grewal2024agnostic} to agnostically learn stabilizer product states (of which parity states are a subclass).
Unfortunately, the results in~\cite{grewal2024agnostic} solve a more general problem than the one we are considering and the resulting time and sample complexity of their results is $\poly(n,2^{1/\tau})$. As we are aiming for a polynomial scaling in $1/\tau,$ this is unsatisfactory. 
\vspace{1mm}

\textbf{\emph{2. Stabilizer Bootstrapping:}} Another approach is to look at the recent work of Chen et al.~\cite{chen2024stabilizer} for agnostic learning of stabilizer states. As with~\cite{grewal2024agnostic}, however, they solve the more general task of agnostically learning states with \emph{stabilizer dimension-$t$}, achieving sample and time complexity scaling $\poly(n,(1/\tau)^{\log 1/\tau},2^t)$. It is not hard to see that learning $s$-juntas or $s$-term $\DNF$s reduces to agnostic learning states with stabilizer dimension $O(s)$, so their agnostic learning algorithm would scale as $\poly(n,(1/\tau)^{\log 1/\tau},2^s)$.

\vspace{1mm}
\textbf{\emph{3. Product state learning:}} A third approach is to look at a recent work of Bakshi et al.~\cite{bakshi2024learning} who considered agnostic learning matrix product states ($\textsf{MPS}$): in particular, they show that agnostic learning $n$-qubit $\textsf{MPS}$ with bond dimension $r$ has complexity $\poly(n,r,1/\varepsilon)$. The question then is: what is the bond dimension of junta states or $\DNF$ states?\footnote{In Appendix~\ref{sect:bonddim} we show the bond dimension bounds that we claim next.} 
We observe that $k$-junta states have bond dimension that scales as $2^{ \lfloor k/2 \rfloor}$, so in particular this yields a quantum agnostic learning algorithm with time complexity $\poly(n,2^{ k },1/\varepsilon)$, which is already better than the previous two approaches that we mentioned. In fact, for disjunctions (i.e., depth-$1$ circuits) their bond dimension equals $2$, so it yields a $\poly(n,1/\varepsilon)$ algorithm. However, for $s$-term $\DNF$s, the corresponding function states can be shown to have bond dimension that is at least $2^s$. This in turn yields a running time of $\poly(n,2^s,1/\varepsilon)$, which is again too high.

\subsubsection{In the classical setting} In classical learning theory, agnostic learning decision trees ($\DT$) and $\DNF$s has received a lot of attention. Classically, the model of agnostic learning is defined as follows: a learner is given uniformly random $x\in \FF^n$ and $b\sim (1+\phi(x)/2, 1-\phi(x)/2)$ and the goal is to find a function $f$ from the concept class which agrees with $\phi$ as well as possible. The $\SOTA$ algorithms for size-$s$ $\DT$s  run in time polynomial in $n,s$ but requires membership queries~\cite{feldman2009distribution,gopalan2008agnostically}. They crucially use that the $\ell_1$ norm of the Fourier coefficients of $f\in \DT(s)$ is at most $s$. However, $\DNF$s do not enjoy the property of small $\ell_1$ norm, but instead are only known to have sparse Fourier spectrum and at this point, it is even unclear classically how one can obtain a weak learner for $\DNF$s (with examples) if $\phi$ is only promised to satisfy $\max_x |\phi(x)|\leq 1$. We show that this can be circumvented when given access to quantum examples and a weak agnostic learner can be naturally proposed with just the Fourier concentration promise.


\paragraph{Our contributions.} Although the three quantum approaches give new agnostic learning algorithms for parities, disjunctions and juntas, as we've mentioned these approaches eventually lead to agnostic algorithms for $s$-term $\DNF$s with complexity that is exponential in $s$, and our goal is to ideally have a $\poly(n,s)$ algorithm for this class. Conceptually, the goal of these works is to solve a much harder problem, so they do not immediately yield results for the task that we are concerned with in this work (whose motivation comes more from quantum learning theory).  As alluded to earlier, our main contribution is in obtaining an umbrella framework, that achieves two purposes: $(i)$ unifies all the learning algorithms for different classes of phase states, and $(ii)$ is simpler than the algorithms mentioned above (which solve a harder task). We achieve this via \emph{quantum agnostic boosting}, which we describe in the next section, which could be of independent interest. 
\vspace{3mm}

\subsection{Technical overview}
In this section, we give an overview of the proofs of our main results. 

\subsubsection{Quantum agnostic boosting}\label{sec:intro_boosting}
Our quantum agnostic learner is inspired by the gradient-descent based algorithms for classical boosting proposed by the works of Kalai-Kanade~\cite{kanade2009potential} and Feldman~\cite{feldman2009distribution}. Below, we first give a high-level idea of the boosting algorithm before describing the iterations of the procedure.

\paragraph{High-level idea.} Recall that we have a quantum state $\ket{\psi}$ satisfying $\calF_{\Cc}(\ket{\psi}) = \opt$ and let $\varepsilon\in (0,1)$. Consider a weak agnostic learner $\calA_{\WAL}$ that given copies of $\ket{\varphi}$ with $\calF_{\Cc}(\ket{\varphi}) \geq \tau$, outputs a parity state $\ket{\chi_S}$ such that $|\langle \chi_S | \varphi \rangle|^2 \geq P(\tau/n)$ in time $T_{\WAL}$. Our agnostic boosting algorithm then does the following: it first runs $\calA_{\WAL}$ to find a parity state $\ket{\chi_{S_1}}$ that has $P(\opt/n)$-overlap with $\ket{\psi}$. After finding ${S_1}$,  the algorithm ``subtracts" $\ket{\chi_{S_1}}$ from $\ket{\psi}$ by constructing the state $\ket{\psi_2}=\ket{\psi}-\beta_1\ket{\chi_{S_1}}$ (ignoring the normalization for now). It then checks two things: $(i)$ is $\|\ket{\psi_2}\|_2\leq \varepsilon$ and $(ii)$ runs $\calA_{\WAL}$ to check whether $\calF_{\Cc}(\ket{\psi_2}) \leq \varepsilon$ or not: if either of these conditions are met, the algorithm halts and outputs $\ket{\widehat{\phi}}\propto \beta_1\ket{\chi_{S_1}}$. Intuitively the former checks if we have done well on \emph{tomography} (i.e. checking whether $\beta_1 \ket{\chi_{S_1}}$ is close to $\ket{\psi}$), a harder task than agnostic learning, and the latter can be shown to be sufficient for the agnostic learning task (since it is checking whether, by subtracting $\beta_1 \ket{\chi_{S_1}}$, we have moved the state far from $\Cc$). If neither is satisfied, the algorithm repeats the same procedure on $\ket{\psi_2}.$ Again, this means running the weak learner on $\ket{\psi_2}$ to produce a parity state $\ket{\chi_{S_2}}$ which is ``subtracted'' along $\ket{\chi_{S_1}}$ from $\ket{\psi}$ and then the stopping conditions are checked. Eventually, after $\kappa$ many iterations, the algorithm terminates producing a (suitably normalized) state $\ket{\widehat{\phi}}=\sum_i \beta_i \ket{\chi_{S_i}}$, which will be the output of the algorithm. We now discuss the iterations in more~detail and why $\ket{\widehat{\phi}}$ accomplishes agnostic learning.

\paragraph{Iterations in the boosting algorithm.} 
Our agnostic boosting algorithm will build a state $\ket{\widehat{\phi}}$ expressed as a linear combination of parity states
$$
\ket{\widehat{\phi}} = \sum_{i=1}^\kappa {\beta}_i \ket{\chi_{S_i}},
$$
across $\kappa$ iterations and stops when $\ket{\widehat{\phi}}$ achieves  agnostic learning condition i.e., $|\langle \psi | \widehat{\phi}\rangle|^2 \geq \opt - \varepsilon$. Each parity state is learned sequentially in each iteration, as described above. Let us denote $\ket{\widehat{\phi}^{(t)}}$ as a ``running estimate'' state at the end of iteration $t$. We also denote the parity states learned up to (and including) iteration $t \geq 1$ as $\{\ket{\chi_{S_i}}\}_{i \in [t]}$ and the corresponding span as $T(t) = \spann(\{\ket{\chi_{S_i}}\}_{i \in [t]})$. We denote the orthogonal projector onto this span as $\Lambda_{T(t)}$. Lastly, we denote the inner products $\beta_i := \langle \psi | \chi_{S_i} \rangle$ and the norms $\alpha_{t+1} = \| (\id-\Lambda_{T(t)})\ket{\psi} \|_2$.

Initially, in iteration $t=1$, the algorithm first runs the weak learner $\calA_{\WAL}$ on copies of $\ket{\psi}$ to find a parity state $\ket{\chi_{S_1}}$ such that $|\langle \psi |\chi_{S_1} \rangle|^2 \geq P(\opt/n)$ (where $P$ is the promise of $\calA_{\WAL}$). The running estimate is then
$$
\ket{\widehat{\phi}^{(1)}} = \beta_1 \ket{\chi_{S_1}},
$$
where $\beta_1 = \langle \psi |\chi_{S_1} \rangle$.\footnote{In our boosting algorithm, we will not actually compute $\beta_1$ at this stage, and instead only keep $\ket{\chi_{S_1}}$.} Before proceeding to the next iteration, the algorithm checks if we would have accomplished state tomography i.e., $|\langle \psi | \widehat{\phi}^{(1)}\rangle|^2 \geq 1 - \varepsilon$ and stops if this is the case. This is done by checking if $|\alpha_2|^2 = 1 - |\beta_1|^2 < \varepsilon$. If not, the residual state is set to 
$$
\ket{\psi_2} := \big(\id-\Lambda_{T(1)}\big)\ket{\psi} = \big(\id - \ket{\chi_{S_1}}\bra{\chi_{S_1}}\big) \ket{\psi}
$$ 
up to renormalization, where $T(1) = \spann(\{ \ket{\chi_{S_1}} \})$. We now proceed to the next~iteration.

In iteration $t=2$, the boosting algorithm first checks if the running estimate $\ket{\widehat{\phi}^{(1)}}$ accomplishes the task of \emph{agnostic learning} i.e., $|\langle \psi | \widehat{\phi}^{(1)} \rangle|^2 \geq \opt - \varepsilon$. To do this, the learner prepares copies of the residual state $\ket{\psi_2}$ and checks if $\calF_{\Cc}(\ket{\psi_2}) < \varepsilon$ or not. To accomplish this, the learner does the following: recall that if $\calF_{\Cc}(\ket{\psi_2}) \geq \varepsilon$, then running $\calA_{\WAL}$ on $\ket{\psi_2}$, would produce a parity $\ket{\chi_U}$ such that $|\langle \psi_2|\chi_U\rangle|^2\geq P(\varepsilon/n)$ (which can be checked by a $\SWAP$ test). By the contrapositive, if the output of $\calA_{\WAL}$ does not output a parity for which $|\langle \psi_2|\chi_U\rangle|^2\geq P(\varepsilon/n)$, then $\calF_{\Cc}(\ket{\psi_2}) < \varepsilon$ and we stop (hence we have implicitly used $\calA_{\WAL}$ also as a tester for fidelity). It might seem counterintuitive to run a test on the residual state instead of directly checking the overlap of $\ket{\psi}$ with $\ket{\widehat{\phi}^{(1)}}$ via a \SWAP test. However for the latter, we would need to know $\opt$ ahead of time, whereas we are assuming that $\opt$ is not known. Instead, we show that if $\calF_{\Cc}(\ket{\psi_2}) < \varepsilon
$, then $\ket{\widehat{\phi^1}}$ (normalized) solves the task of agnostic learning. 

If $\calF_{\Cc}(\ket{\psi_2}) \geq \varepsilon$, the algorithm runs  $\calA_{\WAL}$ on copies of $\ket{\psi_2}$, to find a parity function $\ket{\chi_{S_2}}$ such that $|\langle \psi_2 | \chi_{S_2} \rangle|^2 \geq P(\varepsilon/n)$. We then observe, by writing out $\ket{\psi_2}$, that
$$
P(\varepsilon/n) \leq |\langle \chi_{S_2} | \psi_2 \rangle|^2 = \frac{1}{|\alpha_2|^2}|\langle \chi_{S_2} | (\ket{\psi} - \beta\ket{\chi_{S_1}})|^2 = \frac{1}{|\alpha_2|^2}|\langle \chi_{S_2} | \psi \rangle |^2 \implies |\langle \chi_{S_2} | \psi \rangle|^2 \geq \varepsilon \cdot P(\varepsilon/n),
$$
where we have used $\langle \chi_{S_2} | \chi_{S_1} \rangle = 0$ in the last step before the implication and used the fact that $|\alpha_2|^2 \geq \varepsilon$ (as determined at the end of iteration $t=1$) to give the implication. Our running estimate at this point is 
$$
\ket{\widehat{\phi}^{(2)}} = \beta_1 \ket{\chi_{S_1}} + \beta_2 \ket{\chi_{S_2}}
$$ with the promise that $|\beta_1|^2 \geq P(\opt/n)$ and $|\beta_2|^2 \geq \varepsilon \cdot P(\varepsilon/n)$. Overall, this implies that
$$
|\langle \widehat{\phi}^{(2)} | \psi \rangle|^2 \geq 2 \varepsilon P(\varepsilon/n),
$$
and thus have made progress towards the task of agnostic learning. As in the previous iteration, the learner now checks if $|\alpha_3|^2 = 1 - |\beta_1|^2 - |\beta_2|^2 < \varepsilon$. If not, the learner sets the residual state to 
$$
\ket{\psi_3} \propto (\id-\Lambda_{T(2)})\ket{\psi} = \big(\id - \ket{\chi_{S_1}} \bra{\chi_{S_1}} - \ket{\chi_{S_2}}\bra{\chi_{S_2}}\big)\ket{\psi} = \ket{\psi} - \ket{\widehat{\phi}^{(2)}}
$$ 
up to normalization, with $T(2) = \spann(\ket{\chi_{S_1}},\ket{\chi_{S_2}})$. The algorithm then moves to iteration $t=3$ and continues until either $|\alpha_{t+1}|^2 < \varepsilon$ or $\calF_{\Cc}(\ket{\psi_{t+1}}) < \varepsilon$ which can again be checked using the $\mathcal{A}_{\WAL}$. In other words, the algorithm stops when state tomography or agnostic learning has been achieved.

Overall our agnostic boosting algorithm can be divided into two stages, \emph{structure learning} and \emph{parameter learning}.\footnote{The choice for these terms comes from the literature on learning \emph{graphical models} where the goal is to learning the interactions and interaction strengths.}  In {structure learning}, the goal is to learn the parities that constitute the elements of $\ket{\widehat{\phi}^{(t)}}$, so each iteration starts with structure learning. At multiple times, we mentioned that $\ket{\widehat{\phi}^{(t)}}$ is the state prepared at the $t$th iteration, but so far we only determined the parities present inside $\ket{\widehat{\phi}^{(t)}}$. Ideally, one could have let $\ket{\widehat{\phi}^{(t)}}$ be the projection $\Lambda_{T(t)} \ket{\psi}$ but that requires learning the coefficients $\beta_i$ (including the phases). Estimating these coefficients $\beta_i$s is referred to as parameter learning. To do so, one could compute $\beta_{t+1} = \langle \chi_{S_{t+1}} | \psi \rangle$ via the Hadamard test using the state preparation unitaries (and controlled versions) of $\ket{\chi_{S_{t+1}}}$ and $\ket{\psi}$. However, we avoid the need for a state preparation unitary and instead show that with just copies of $\ket{\psi}$, we can estimate $\beta_i$ up to a global phase, and a valid proxy state $\ket{\widehat{\phi}^{(t)}}$ that is close to $\Lambda_{T(t)}\ket{\psi}$, hence is good at the task of agnostic learning.

\paragraph{What remains?} In the brief exposition above, there are a number of subtleties that we have swept under the rug: $(i)$ an upper bound on the number of iterations $\kappa$\footnote{As part of our analysis, we show that we stop in $O(1/(\varepsilon P(\varepsilon/n))$ many iterations and the time complexity of the overall algorithm is then dictated by the promise $P$ of $\calA_{\WAL}$.}, $(ii)$ the eventual correctness of the final state $\ket{\widehat{\phi}^{(t)}}$, $(iii)$ the preparation of the residual states $\ket{\psi_t}$, $(iv)$  the normalization factors in $\ket{\psi_t}$, $(v)$ the circuit implementations of various subroutines in the algorithm and their complexity and finally $(vi)$ the errors in the steps involving estimation and how they propagate in the algorithm. Our final boosting algorithm incorporates all these details and making it rigorous is the most technical part of our work. 

\subsubsection{Learning algorithms}
In this section, we state the learning algorithms that are either used by the boosting procedure, or which the boosting procedure implies. Beginning with the former, we describe a weak learner for parity states. The subsequent algorithms are obtained by using the boosting procedure. 

\textbf{Weak learner.} To agnostically learn parity states, we simply observe that, if $\ket{\psi}$ is $\tau$-close to a parity $\ket{\chi_S}$, then we have that
$$
|\langle \psi'| S\rangle|^2=|\langle \psi|\Had \cdot \Had |\chi_S\rangle|^2\geq \tau,
$$
where $\ket{\psi'}=\Had\ket{\psi}$. Thus, if we measure $\ket{\psi'}$, $O(1/\tau^2)$ many times in the computational basis, we will recover $S$. Specifically, we record the measurement outcomes and check, via a \SWAP test, which basis state has the highest fidelity with $\ket{\psi'}.$ This will be the agnostic learner for parities.

\vspace{1mm}

\noindent\textbf{Agnostic learning for decision trees.} As mentioned earlier, unlike parity states whose Fourier spectrum is concentrated on a single point, for  decision trees, $\DNF$s and juntas, we do not have this property. In fact, it is well-known~\cite{kushilevitz1993decision} that for a function $f$, computed by a size-$s$ decision tree, we have that $\sum_T |\widehat{f}(T)|\leq s$. In particular, it is not not too hard to see that if $| \langle \psi | \phi_f \rangle |^2 \geq \tau,$ then
$$
\sqrt{\tau} \leq |\langle \psi|\phi_f\rangle| = \Big| \sum_T \widehat{f}(T) \langle \psi'|T\rangle \Big| \leq \sum_T | \widehat{f}(T) | \; | \langle \psi' | T \rangle| \leq O(s) \max_T |\langle \psi' |T\rangle|,
$$ 
where in the first equality we applied Hadamard on both states and denoted $\ket{\psi'}=\Had\ket{\psi}$. We then used the triangle inequality and the fact that the $\ell_1$ norm of the Fourier coefficients of $f$ is at most $O(s)$. The above implies that there is a basis state in $\ket{\psi'}$ whose amplitude is $\Omega(\sqrt{\tau}/s)$. Finding it can be done by measuring $\ket{\psi'}$ several times and recording the statistics of the measurement outcomes. This serves as a weak learner and we then use our boosting algorithm to obtain a strong learner which outputs $\ket{\phi}$ (as a superposition over $\poly(s/\varepsilon)$ parity states) such that
\begin{align}
\label{eq:decisiontreeagnostic}
|\langle \psi|\phi\rangle|^2\geq\max_{c\in \DT(s)} |\langle \psi|\phi_c\rangle|^2-\varepsilon.
\end{align}
The overall complexity is $\poly(n,s,1/\varepsilon)$. At this point, we use the fact that $k$-juntas are decision trees of size $2^k$, giving  an algorithm with complexity $\poly(n,2^k,1/\varepsilon)$ for agnostic learning juntas.

\vspace{1mm}

\noindent \textbf{Agnostic learning $\DNF$s.} The agnostic learner for $\DNF$s is similar to the one for decision trees, except that we need to use Mansour's result~\cite{mansour1992n,lecomte2022sharper} that shows that for size-$s$ $\DNF$ formulas, the entire Fourier spectrum is concentrated on $s^{O(\log (1/\varepsilon) \cdot \log \log s )}$ coefficients. Using a similar argument as the one for deriving Eq.~\eqref{eq:decisiontreeagnostic}, one can show that if $\ket{\psi}$ is $\tau$-close to a $\DNF$ phase state, then there exists a basis state $\ket{\chi_T}$ such that $|\langle \psi|\chi_T\rangle|^2\geq \tau/s^{O(\log (1/\varepsilon) \cdot \log \log s)}$. Once again using our quantum boosting algorithm, we obtain a $\poly(n,s^{\log \log s\log (1/\varepsilon)})$ algorithm for agnostic learning size-$s$ $\DNF$ formulas. 
\vspace{2mm}

\noindent \textbf{$\PAC$ learning depth-$3$ circuits.} To learn these circuits, we employ our agnostic $\DNF$ learner. The key insight is that when the input state $\ket{\psi_f}$ is promised to correspond to a depth-$3$ circuit then state tomography is accomplished when agnostic learning against $\DNF$s is accomplished, which was also observed classically~\cite{feldman2009distribution}. This follows from using the seminal result of Hajnal et al.~\cite{hajnal1993threshold} that says that if $f$ is a threshold of $m$ many $\DNF$ formulas $\{g_1,\ldots,g_m\}$ (each with size at most $s$), then $|\langle \psi_f|\psi_{g_i}\rangle| \geq 1/m,$ for some $i \in [m]$,  where $\ket{\psi_f}$ (in this section) equals $\frac{1}{\sqrt{2^n}}\sum_x f(x)\ket{x}$ since we are in the $\PAC$ learning setting. We can now use our agnostic $\DNF$ learner outputs a quantum state $\ket{\phi}$ which is at least $\opt/m^2$ close to $f$. This will serve as our weak learner which we will then boost into a strong $\PAC$ learner. The boosting algorithm outputs a (classical description of a) quantum state $\ket{\phi}$ which is close to an unknown $\ket{\psi_f}$; we now \emph{round} the final state $\ket{\phi}$ of the  algorithm and show that it satisfies the requirement of   $\PAC$ learning. Since the runtime of the $\DNF$ learner scales as  $\poly(n,s^{\log (1/\varepsilon) \cdot \log \log s})$, the overall $\delta$-error quantum $\PAC$ learning algorithm scales as~$\poly(n,s^{\log (s/\delta) \cdot \log \log s})$.

\subsubsection{Learning in the distributional model.} Finally, we remark that there are two natural definitions of quantum agnostic learning: the one defined in the introduction of this work, i.e., $\ket{\psi}$ is arbitrary and promised to be close to $\ket{\phi_c}$ (where $c\in \Cc$), or the one that was considered in~\cite{arunachalam2017guest,caro2024interactive} wherein there is an unknown distribution $D:\FF^{n+1}\rightarrow [0,1]$ whose first $n$ bits are uniform and the last bit is described by the marginal $\big((1+\phi(x))/2,(1-\phi(x))/2\big)$ where $\phi:\FF^n\rightarrow [-1,1]$ is an arbitrary function. The quantum algorithm is given copies of
$$
\ket{\psi_D}=\frac{1}{\sqrt{2^n}}\sum_x \ket{x} \otimes \Big(\sqrt{\frac{1+\phi(x)}{2}}\ket{0}+\sqrt{\frac{1-\phi(x)}{2}}\ket{1}\Big),
$$
and the goal is to output a function $h:\FF^n\rightarrow \FF$ such that $\Pr_{(x,b)\sim D}[h(x)=b]\geq \opt-\varepsilon$ where $\opt=\max_{c\in \Cc}[c(x)=b]$. In contrast to the first model, in this distributional model the learning algorithm enjoys the benefit of knowing that the unknown quantum state $\ket{\psi}$ has the form of $\ket{\psi_D}$. However, the algorithm needs to output a \emph{function} $h$, whereas in the first model it can output an arbitrary $\ket{\varphi}$. The two situations are therefore incomparable and interesting for different reasons. In Section~\ref{sec:distributionagnostic} we show that if $\phi$ is ``nice", i.e., $\Exp_x[\phi(x)^2]\leq \opt$, having a learning algorithm in the first model implies a learning algorithm in the second model. In particular, for these distributions, we also obtain quantum learning algorithms in the distributional model.
We leave the other direction for future~work.

\subsection{Outlook}

\textbf{Related works.} We remark that there have been a few recent works on quantum agnostic learning that we briefly describe here.  In~\cite{arunachalam2017guest}, the authors showed that quantum examples are equivalent to classical examples in the distribution-independent framework for agnostic learning function classes. Badescu and O'Donnell~\cite{DBLP:conf/stoc/BadescuO21} considered the setting of agnostic learning \emph{arbitrary} classes of quantum states and gave sample complexity bounds for this task. In another direction, works of~\cite{chatterjee2024efficient,caro2023classical,caro2024interactive} considered quantum agnostic learning for the case where the input state has some specific structure (in the former work they assume it is a function state and in the latter two works they assume it is a ``mixture of superpositions'').

\textbf{Open questions.} Our work opens up a number of interesting research directions.
\begin{enumerate}
    \item \textbf{\emph{Learning quantum objects}}: In this work, We considered the learnability of depth-$2$ and depth-$3$ $\AC^0$ circuits, what about learning depth-$2$ or depth-$3$ $\textsf{QAC}^0$ circuits with or without fanout gates? Recently,~\cite{foxman2025random} proved the hardness of learning $\textsf{QAC}^0$ circuits, but their hard instances require depth that is a ``large constant.'' 

    Similarly, we could consider the problem of agnostic learning \emph{low-degree} {phase states}. Tomography protocols for these class of states are known~\cite{abdy2023phasestates} but agnostic learning algorithms are unknown.
  
    \item \textbf{\emph{Learning in the distributional model}}: We showed how to port learning algorithms from the state agnostic learning model to the  standard distributional quantum agnostic model when the marginal function on the last bit $\phi:\FF^n\rightarrow [-1,1]$ satisfied $\Exp_x[\phi(x)^2]=1/\poly(n)$. In this state distributional model, can we learn even parities for all $\phi$, or can we prove lower bounds that rule this out? 
    \item \textbf{\emph{Learning more expressive circuits?}} Classically, it is believed that~\cite{naor2004number} depth-$5$ circuits are hard to learn (\emph{assuming} factoring is hard). 
    Our work leaves opens the status of learning depth-$4$ circuits, which is the ``only depth setting'' for which we do not know any classical or quantum learning algorithms or hardness results. 

    Similarly, classically the works~\cite{jackson2002learnability,chen2021majority} have considered the class of learning threshold of $\AC^0$  gates,  and~\cite{carmosino2016learning} has looked at learning $\AC^0$ augmented with $\!\!\!\!\mod p$ gates, with membership queries. It would be interesting if one could learn threshold of $\AC^0[p]$ circuits using only quantum examples (removing the need for classical queries).

    \item \textbf{\emph{Proper learning}}: The agnostic learning algorithms presented in this work for decision trees, juntas, and $\DNF$s have been improper. A natural question is then: Could we obtain \emph{proper} agnostic learners for these classes of phase states with similar time complexities? This has also remained open classically~\cite{gopalan2008query} and a quantum approach might lead to new insights.
\end{enumerate}

\paragraph{Acknowledgments.}  SA thanks Matthias Caro, Alex Grilo and Ryan Sweke for an early discussion on agnostic learning. AD thanks Isaac Chuang and Kristan Temme for early discussions on agnostic learning phase states. We thank Igor Carboni, Gautam Chandrasekaran, Varun Kanade, and Adam Klivans for helpful clarifications on the classical $\SOTA$ algorithms. This work was done when MdO was an intern at IBM Quantum. 

\section{Preliminaries}
\subsection{Notation}
Let $[n]=\{1,\ldots,n\}$. We define $\calB_\infty^k$ as the unit complex ball, i.e., $x\in \calB_\infty^k$ if $x_i \in \mathbb{C}$ for all $i\in [k]$ and $|x_i| \in (0,1]$. For a set $S\subseteq [n]$ we denote $z\in \01^S$ to be a bit-string of length $|S|$. For notational convenience, we will denote $\ket{z_S,0_{\overline{S}}}$ to denote the quantum state where the $i$'th qubit is $z_i$ if $i\in S$ and $0$ otherwise. Similarly, by $\ket{+}_S\ket{0}_{\overline{S}}$, we mean qubit $i$ equals $\ket{+}$ if $i\in S$ and $\ket{0}$ otherwise. For $\varepsilon\in (0,1)$, we say $f(\varepsilon)=\poly(\varepsilon)$ if there exist constants $c_1,c_2\geq 1$ such that $f(\varepsilon)=c_1\varepsilon^{c_2}$.\footnote{In this paper, there are several polynomial factors that we have not explicitly optimized, so we use the convention $\poly(\varepsilon)$ to make the exposition easier to follow.}

\begin{fact}
\label{fact:taylor}
    For every $x\in (-1,1)$, by Taylor series expansion we have that
    $$
    1+x/2-x^2/2\leq \sqrt{1+x}\leq 1+x/2, \text{ and }  1-x/2-x^2/2\leq \sqrt{1-x}\leq 1-x/2.
    $$
\end{fact}

\paragraph{Fourier analysis.} We introduce the basics of Fourier analysis on the Boolean cube here, referring to~\cite{o2014analysis} for more. For functions $f,g:\01^n\rightarrow \R$, define their inner product~as
	$$
	\langle f,g\rangle=\Exp_{x\in \FF^n} [f(x)\cdot g(x)],
	$$ 
	where the expectation is with respect to the uniform distribution over $\FF^n$. For $S\in\FF^n$, the character function corresponding to $S$ is given by $\chi_S(x):=(-1)^{S\cdot x}$, where the dot product $S\cdot x$ is $\sum_{i=1}^n S_ix_i$. Observe that the set of parity functions $\{\chi_S\}_{S\in \FF^n}$ forms an orthonormal basis for the space of all real-valued functions over the Boolean cube. In particular, every $f:\FF^n\rightarrow \R$ can be written uniquely as 
	$$
	f(x)=\sum_{S\in \FF^n} \widehat{f}(S) \chi_S(x) \quad \text{for all }x\in \FF^n,
	$$ 
	where $\widehat{f}(S)=\langle f,\chi_S\rangle=\Exp_x[f(x)\chi_S(x)]$ is called a \emph{Fourier coefficient} of $f$. A well-known result in Fourier analysis is Parseval's theorem that states that $\Exp_{x\in \FF^n}[f(x)^2]=\sum_S \widehat{f}(S)^2$. In particular, if $f:\FF^n\rightarrow \{-1,1\}$, this implies that $\{\widehat{f}(S)^2\}_S$ forms a probability distribution. 
    
\subsection{Interesting concept classes} \label{subsect:conceptclasses}
In this section, we introduce the main concept classes that we will be dealing with in  this work. 

\paragraph{Parities.} This is the concept class defined as 
$$
\Cc_\Par=\{\chi_s: \chi_s(x)=\langle s,x\rangle\}_{s\in \01^n}
$$ where $\langle s,x\rangle=\sum_i s_i x_i \mod 2$. 

\paragraph{Juntas.} We say a Boolean function $c:\01^n\rightarrow \01$ is a $k$-junta if there exists  $S=\{i_1,\ldots,i_k\}\subseteq [n]$ of size $|S|=k$ such that $c(x_1,\ldots,x_n)=g(x_{i_1},\ldots, x_{i_{k}})$ where $g:\01^k\rightarrow \01$ is an arbitrary function on $k$ bits. In relation to the class of disjunctions, note that $\OR_S$ is a $|S|$-junta. 

\paragraph{Decision trees.} A decision tree  $(\DT)$ on $n$ Boolean variables is a binary tree such that the leaves have labels chosen from $\01$ and the internal nodes of the tree have two children, the left child and the right child. On input $x\in \01^n$, an algorithm traverses the binary tree from the root to a leaf by evaluating the node at the $i$th level as follows: if $x_i=0$, go to the left child and if $x_i=1$, go to the right child. The output of the $\DT$ is the label of the leaf that the algorithm reaches. The size of the decision tree is the total number of nodes in the tree.\footnote{We remark that there are some works that call the \emph{number} of leaves as the $\DT$ size, but this is a factor $2$ smaller than the way we define it here.} We say that a function $c:\01^n\rightarrow \01$ is computed by a size-$s$ decision tree, if there exists a size-$s$ $\DT$ such that for every $x$, traversing this $\DT$ and outputting the label of the leaf yields $c(x)$.

\paragraph{Depth-$2$ circuits.} Depth-$2$ circuits consisting of $\AND, \OR, \NOT$ gates are often referred to as \emph{disjunctive normal form} ($\DNF$s) formulas. One also refers to this concept class as $\textsf{AC}_2^0$. In particular, the class of $s$-term $\DNF$ formulas is defined as depth-$2$ circuits where the first layer consists of $s$ $\AND$ gates, each with unbounded fanin (i.e., they take in as input an arbitrary subset of the variables in $x_1,\ldots,x_n$) and the second layer is a single $\textsf{OR}$ gate of fanin $s$. The size of the circuit (or the $\DNF$ formula) is the total number of gates in the circuit, which in this case will be $s+1$.

\paragraph{Depth-$3$ circuits.} In this work, we will consider two different notions of depth-$3$ circuits. The first is $\textsf{AC}_3^0$: these are depth-$3$ circuits where the gates are alternating layers of $\AND$s and $\OR$s. For example, the top gate may be an $\AND$ that takes as input a collection of $\OR$ gates, each of which in turn takes as input a collection of $\AND$ gates. Another type of depth-$3$ circuit we consider is \emph{threshold} of $\DNF$s. To define this, we first define the threshold function.

\begin{definition}
    A threshold function has the form,
    \begin{equation}
        T_k^m(y_1,..,y_m)= \begin{cases}
            1,\text{  if } \sum_{i=1}^m  y_i\geq k\\
            0,\text{ otherwise}
        \end{cases}.
    \end{equation}
\end{definition}
Now, one can define the threshold-of-$\DNF$s class as follows.
\begin{definition}[Threshold-of-$\DNF$s]
     Define $\textsf{TAC}^0_2$ to be the class of depth $3$ circuits on $n$ bits where the top gate is a threshold function whose inputs are $\DNF$ formulas acting on $n$ bits.
\end{definition}
For both definitions, $\AC_3^0$ and $\TAC_2^0$, the size of the corresponding circuit is the total number of $\AND, \OR, \NOT$ gates.

\subsection{Function and state classes}
For notational convenience, we will be explicit about the size in parenthesis, i.e., $\DT(s),\AC_3^0(s),\TAC_2^0(s)$ will be size-$s$ decision trees and circuits, respectively. Throughout the paper we will denote $\Cc_{\Par}$ as the class of parities, $\Cc_{\DT(s)}$ as the class of decision trees of size $s$, $\Cc_{\Jun(k)}$ as the class of $k$-juntas, $\Cc_{\DNF(s)}$ as the class of $s$-term $\DNF$ formulas, $\Cc_{\AC_3^0(s)}$ as the class of $\AC_3^0(s)$ circuits and  $\Cc_{\TAC_2^0(s)}$ as the class of $\TAC_2^0(s)$ circuits. 

For every concept class $\Cc$, we will denote the phase state corresponding to a function $c \in \Cc$ as
\begin{equation}\label{eq:phase_state}
    \ket{\phi_c} = \frac{1}{\sqrt{2^n}} \sum_{x \in \FF^n} (-1)^{c(x)} \ket{x},
\end{equation}
and $\Sh_{\Cc}$ to be the corresponding \emph{state class}, i.e.,
\begin{equation}\label{eq:state_class}
\Sh_\Cc=\Big\{\ket{\phi_c}=\frac{1}{\sqrt{2^n}}\sum_{x\in \01^n}(-1)^{c(x)}\ket{x}:c\in \Cc\Big\}.    
\end{equation}
Furthermore, we define the \emph{fidelity} of an unknown $\ket{\psi}$ with respect to the class $\Sh_\Cc$ as $\calF_\Cc$, i.e., a state $\ket{\psi}$ is said to have $\calF_\Cc(\ket{\psi})=\opt$, if
$$
\max_{c\in\Cc} |\langle \phi_c|\psi\rangle|^2=\opt.
$$

\subsection{Learning models}

\subsubsection{PAC learning}

\paragraph{Classical $\PAC$ learning.} In his seminal paper, Valiant~\cite{DBLP:journals/cacm/Valiant84} introduced the \emph{Probably Approximately Correct} model of learning, often referred to as $\PAC$ learning. In this model, there is a \emph{concept class} $\Cc\subseteq \{c:\01^n\rightarrow \01\}$  which is a collection of Boolean functions. The goal of the learning algorithm is to learn $\Cc$ in the following sense: The learner $\A$ obtains \emph{labeled examples} $(x,c(x))$ where $x\in \01^n$ is uniformly random and $c\in \Cc$ is the \emph{unknown} {target} function promised to lie in $\Cc$.\footnote{In the general $\PAC$ learning model, there is an \emph{unknown} distribution $D:\01^n\rightarrow [0,1]$ from which $x$ is drawn.  In this paper we will only be concerned with uniform-distribution $\PAC$ learning, i.e., $D$ is the uniform distribution.} The goal of an $(\varepsilon,\delta)$-learner $\A$ is as follows: for every $c\in \Cc$, given labeled examples $\{(x^i,c(x^i))\}_i$, with probability $\geq 1-\delta$ (over the randomness of the labeled examples and the  randomness of the learner), output a \emph{hypothesis} $h:\01^n\rightarrow \01$ such that $\Pr_x [c(x)=h(x)]\geq 1-\varepsilon$. In other words, with \emph{probability} $1-\delta,$ the hypothesis $h$ $\varepsilon$-\emph{approximates} $c$. The $(\varepsilon,\delta)$-sample complexity of a learning algorithm $\A$ is the maximal number of labeled examples used for the hardest concept, i.e., maximized over all $c\in \Cc$. The $(\varepsilon,\delta)$-sample complexity of learning $\Cc$ is the \emph{minimal} sample complexity  over all $(\varepsilon,\delta)$-learners for $\Cc$. Similarly the $(\varepsilon,\delta)$-time complexity  of learning $\Cc$ is the total number of time steps used by an optimal $(\varepsilon,\delta)$-learner for $\Cc$.  We say a learner is \emph{proper} if the output hypothesis $h$ lies within the concept class $\Cc$ and otherwise it is referred to as \emph{improper}. Throughout the paper, we present improper learners, with the exception of the parity learner, which is proper.

\paragraph{Quantum $\PAC$ learning.}  The quantum $\PAC$ model was introduced by  Bshouty and Jackson~\cite{bshouty1995learning} wherein they allowed the algorithm access to quantum examples of the form
$$
\ket{\psi_c}=\frac{1}{\sqrt{2^n}}\sum_{x\in \01^n}\ket{x,c(x)}.
$$
In particular, the learner is given copies of $\ket{\psi_c}$ and is allowed to perform arbitrary measurements on those copies. 
Note that measuring $\ket{\psi_c}$ in the computational basis produces a classical labeled example, so quantum examples are \emph{at least} as strong as classical examples. Understanding the strengths and weaknesses of quantum examples has been looked at by several works (we refer an interested reader to the survey~\cite{arunachalam2017guest}). As with the classical complexities, one can similarly define the $(\varepsilon,\delta)$-sample and time complexity for learning $\Cc$ as the quantum sample complexity (i.e., number of quantum examples $\ket{\psi_c}$ used) and quantum time complexity (i.e., number of one and two-qubit quantum gates used in the algorithm) of an optimal $(\varepsilon,\delta)$-learner for $\Cc$. 

\subsubsection{Classical agnostic learning}
Let $D$ be a distribution $D:\FF^n\rightarrow [0,1]$ and  $\phi:\FF^n\rightarrow [-1,1]$. We say $A=(D,\phi)$ is a distribution on $\FF^{n+1}$ satisfying: the marginal on the first $n$ bits of $A$ is given by the distribution $D$ and the distribution of the last bit is described by the distribution $(1+\phi(x))/2,(1-\phi(x)/2)$. More formally, for the distribution  $A=(D, \phi)$, we have $D(z) = \Pr_{(x,b) \sim A}[x = z]$ and
\[
\phi(z) = \Exp_{(x,b) \sim A}[b \mid z = x].
\]
Formally, for a Boolean function $h$ and a distribution $D$, we define
\[
\Delta(A, h) = \Pr_{(x,b) \sim A}[h(x) \ne b].
\]
Furthermore, we have the following simple equality 
\begin{align}
\Delta(D, h) = (1 - \langle \phi, h \rangle_D)/2=(1 -\Exp_{x\sim D} [\phi(x)h(x)])/2.
\end{align}
For a concept class $\Cc$, define 
$$
\Delta(A, \Cc) = \min_{h \in \Cc} \{\Delta(A, h)\}
$$
Now one can formally define agnostic learning as follows
\begin{definition}[\cite{kearns1992toward}]
    {An algorithm $\mathcal{A}$ agnostically learns a class $\Cc\subseteq \{h:\FF^n\rightarrow \{-1,1\}\}$ by a representation class $H$ if for every $\varepsilon,\delta > 0$, distribution $A$ over $\FF^n \times \{-1,1\}$, $\mathcal{A}$, given access to examples drawn randomly from $A$, outputs, with probability at least $1 - \delta$, a hypothesis $h \in H$ such that $\Delta(A, h) \le \Delta(A, C) + \varepsilon$.}
\end{definition} 
As is often the case, we limit ourselves to the scenario in which the distribution on the first $n$ bits is uniform. In that case, the goal of the learner is to output an $h:\FF^n\rightarrow \{-1,1\}$ such that
$$
(1 -\Exp_{x}[ \phi(x)h(x)])/2 \le \min_{c\in \Cc}\{(1 -\Exp_{x} [\phi(x)c(x)])/2\} + \varepsilon \implies \Exp_{x}[ \phi(x)h(x)]\geq \max_{c\in \Cc}\Exp_{x}[ \phi(x)c(x)]-2\varepsilon.
$$
\subsubsection{Quantum  agnostic learning}
\paragraph{Distributional agnostic learning.} 
Like in the classical model, let $A=(D,\phi)$ be a distribution on $\FF^n$. The quantum learning algorithm is given copies of
$$
\sum_{(x,b)\in \FF^{n+1}}\sqrt{A(x,b)}\ket{x,b}.
$$ 
In the case where $D$ is the uniform distribution, one can view $A(x,b)=2^{-n}\cdot (1+(-1)^{b}\phi(x))/2$ in the expression above. Hence, the learning algorithm is given copies of
    $$
   \frac{1}{\sqrt{2^n}}\sum_x \ket{x} \otimes \Big(\sqrt{\frac{1+\phi(x)}{2}}\ket{0}+\sqrt{\frac{1-\phi(x)}{2}}\ket{1}\Big),
    $$
    Like in the classical setting, the goal is to output a $h$ such that
    $$
    \Exp_{x}[ \phi(x)h(x)]\geq \max_{c\in \Cc}\Exp_{x}[ \phi(x)c(x)]-\varepsilon.
    $$

\paragraph{State agnostic learning.}
In this model, the hypothesis class is a set of quantum state $\Sh=\{\ket{\psi_1},\ldots,\ket{\psi_m}\}$.\footnote{We remark that we define this model for pure states for simplicity since that is the focus of this work. One could similarly define a hypothesis class of mixed states $\{\rho_1,\ldots,\rho_m\}$.} The agnostic learning algorithm is given copies of an unknown $\ket{\phi}$ and the goal is to output an $\ket{\phi'}$ such that
$$
|\langle \psi|\phi'\rangle|^2\geq \max_{i\in [m]}|\langle \psi|\phi_i\rangle|^2-\varepsilon.
$$
The sample complexity of learning is the total number of copies used by the algorithm to satisfy the above guarantee and the time complexity is the total time.
We say an algorithm is \emph{sample and time efficient} these complexities  scale polynomial in $n,1/\varepsilon$ and the the description size of the class. If the learner outputs $\ket{\psi'}\in \Sh$, then the learner is called \emph{proper}, else its an \emph{improper} learner. 
As far as we are aware, there are only a handful of works that have considered quantum state agnostic learning~\cite{chen2024stabilizer,DBLP:conf/stoc/BadescuO21,bakshi2024learning} wherein they considered interesting classes of states such as stabilizer states, product states, matrix product states and proved results for this model. As we mentioned in the introduction, in this work, we will be concerned with the concept class of states~being
$$
\Sh=\Big\{\ket{\psi}=\frac{1}{\sqrt{2^n}}\sum_x c(x)\ket{x}: c\in \Cc\Big\},
$$
where $\Cc$ is a Boolean-valued concept class of interest.

\section{Quantum agnostic boosting}\label{sec:agnostic_boosting}
In this section, we introduce the framework of quantum agnostic boosting and prove one of our main theorems.
In order to define the boosting algorithm, we first give the following definition of a weak agnostic learner.   
\begin{definition}[Weak agnostic leaner]\label{def:WAL}
Let $\tau \in (0,1)$ and $\eta(\cdot)$ be a function of $\tau$. Suppose $\ket{\phi}$ is an arbitrary $n$-qubit quantum state such that $\calF_{\Cc}(\ket{\phi}) \geq \tau$. We say $\calA_{\WAL}$ is a weak agnostic learner for $\calS_\Cc$ \emph{with promise $\eta(\cdot)$} if, given copies of $\ket{\phi}$, outputs $\ket{\chi} \in \Cc_{\Par}$ such that $|\langle \phi| \chi \rangle|^2 \geq \eta(\tau)$ with probability $\geq 1-\delta$. Let $S_{\WAL}$ and  $T_{\WAL}$ be the sample and time complexity of $\A_{\WAL}$ respectively.\footnote{We remark that in this work we only consider weak agnostic learners whose output class will be the class of parity states, hence why we define it this way.}
\end{definition}

We are now ready to state our main theorem which comments on obtaining a strong (improper) agnostic learner from a weak agnostic learner via boosting.
\begin{theorem}[Quantum agnostic boosting]
\label{thm:agnostic_boosting}
Let $\varepsilon,\delta,\eta_1\in (0,1)$. Let $\eta_2\geq 1$ be a universal constant and $\eta: [0,1] \rightarrow [0,1]$ be defined as $\eta(\tau) := \eta_1 \tau^{\eta_2}$. Let $\Cc$ be a concept class and $\ket{\psi}$ be an unknown $n$-qubit  state with $\calF_{\Cc}(\ket{\psi}) = \opt$. 

Let $\calA_{\WAL}$ be a weak agnostic learner for $\calS_\Cc$ with promise of $\eta(\cdot)$, with sample complexity $S_{\WAL}$  and time complexity $T_{\WAL}$. Then, there is an algorithm $\calL$ that with probability $\geq 1-\delta$, outputs a state $\ket{\widehat{\phi}}$ expressed as
$$
\ket{\widehat{\phi}} = \sum_{i=1}^\kappa \beta_i \ket{\chi_i},
$$
where $\beta \in \calB_\infty^k$, $\kappa \leq O(1/(\varepsilon^2 \cdot \upsilon))$ with $\upsilon=\eta(C\varepsilon^2/16)$, $C = (2/3)^{1/\eta_2 + 1}$ and $\{\ket{\chi_i}\}_{i\in [\kappa]}$ are parity states. Furthermore $\ket{\widehat{\phi}}$~satisfies 
$$
|\langle \widehat{\phi} | \psi \rangle|^2 \geq \opt - \varepsilon.
$$
This algorithm $\calL$ invokes $\calA_{\WAL}$ $\kappa$~times. The overall complexity of this algorithm is as follows
\begin{align*}
&\text{Sample complexity: } \widetilde{O}\Big(1/(\varepsilon^2 \upsilon)\cdot S_{\WAL} + 1/(\varepsilon^{16}\upsilon^7)\log(1/ \delta)\Big) \\
&\text{Time complexity: } \widetilde{O}\Big(1/(\varepsilon^2 \upsilon)\cdot T_{\WAL} + n^2/(\varepsilon^{16}\upsilon^7)\log(1/ \delta)\Big).
\end{align*} 
\end{theorem}

\subsection{Useful subroutines and lemmas}
In this section, we will provide a few definitions and lemmas that we use in our algorithm. As we mentioned in the high-level idea in the introduction (Section~\ref{sec:intro}), the boosting algorithm will be executed across multiple iterations and at every iteration, we will apply projections of the quantum state in hand onto parity states. We now formally define these projections.

\paragraph{Projection.} For a set of $k$ states $\{\ket{\chi_i}\}_{i \in [k]}$, let its span be denoted as $T = \spann(\{\ket{\chi_i}\}_{i \in [k]})$. Let $\Lambda_T$ be a projection onto $T$. The projection of $\ket{\psi}$ onto $T$  is given by
\begin{equation}
\label{eq:projection_state}
    \Lambda_T \ket{\psi} = \sum_{i=1}^k \beta_i \ket{\chi_i} \enspace \text{ s.t. } \enspace \{\beta_i\}_{i \in [k]} = \argmin_{\alpha_1,\ldots,\alpha_k \in \mathbb{C}} \normstate{\ket{\psi} - \sum \limits_{i=1}^k \alpha_i \ket{\chi_i}}_2
\end{equation}
For the special case of parity states, the projection has a simpler expression. Given a set of parity states $\{\ket{\chi_i}\}_{i \in [k]}$ with $\ket{\chi_i} \in \calS_{\Cc_{\Par}}$, the projection is
\begin{equation}\label{eq:projection_state_when_parities}
    \Lambda_T \ket{\psi} = \sum_{i=1}^k \beta_i \ket{\chi_i} \enspace{ s.t. } \enspace \beta_i =  \langle \chi_i | \psi \rangle,
\end{equation}
where $T = \{\ket{\chi_i}\}_{i \in [k]}$ and using that the parity states are orthogonal to one another\footnote{In particular, note that $\langle Q|\textsf{Had}^{\otimes n}\textsf{Had}^{\otimes n}|R\rangle=0$ if $Q\neq R$.}. In other words, we can represent the projector $\Lambda_T$ in terms of the parity states $\{\ket{\chi_i}\}_{i \in [k]}$ as
\begin{equation}\label{eq:projector_parities}
    \Lambda_T = \sum_{i=1}^k \ket{\chi_i} \bra{\chi_i},
\end{equation}
and the solution to the optimization problem of Eq.~\eqref{eq:projection_state} is when the coefficients are inner products of the basis elements and $\ket{\psi}$. Note that $\Lambda_T$ is a projector since
$$
(\Lambda_T)^2=\sum_{i,j}\ketbra{\chi_i}{\chi_i}\ketbra{\chi_j}{\chi_j}=\sum_i\ketbra{\chi_i}{\chi_i}=\Lambda_T,
$$
since parity states are orthogonal. Additionally, we have the following fact regarding the residual $(\id - \Lambda_T)\ket{\psi}$, which we will use often in our analysis below.
\begin{fact}\label{fact:projection}
Let  $\{\ket{\chi_i}\}_{i \in [k]}$ be a set of  parity states and $T = \spann(\{\ket{\chi_i}\}_{i})$.  Every state $\ket{\psi}$ can be written~as
$$
\ket{\psi} = \Lambda_T \ket{\psi} + \alpha \ket{\phi^\perp},
$$
where $\langle \phi^\perp | \chi_i \rangle = 0$ and $\alpha = \sqrt{1 - \sum_{i=1}^k |\langle \chi_i | \psi \rangle|^2}$.
\end{fact}
\begin{proof}
We can express any arbitrary $\ket{\psi}$ as
\begin{align} \label{eq:bestrepofpsi}
\ket{\psi}= \Lambda_T \ket{\psi} + (\id-\Lambda_T) \ket{\psi} = \Lambda_T \ket{\psi} + \alpha \ket{\phi^\perp},
\end{align}
where we have used that $\Lambda_T$ is an orthogonal projector and $\alpha \ket{\phi^\perp} = (\id - \Lambda_T) \ket{\psi}$ with $\alpha \in \calB_\infty$. Note that $\ket{\phi^\perp}$ is a valid quantum state orthogonal to $\ket{\chi_i}$ for all $i \in [k]$ since $\Lambda_T(\id-\Lambda_T) = 0$.  
Moreover, by Pythagoras' theorem, we have
$$
1 = \norm{\ket{\psi}}_2^2 = \norm{\Lambda_T \ket{\psi}}_2^2 + |\alpha|^2 \cdot \norm{\ket{\phi^\perp}}_2^2 = \sum_{i=1}^{k} |\langle \psi| \chi_i \rangle|^2+|\alpha|^2 \implies \alpha = \sqrt{1 - \sum_{i=1}^{k} |\langle \psi| \chi_i \rangle|^2},
$$ 
where we used that $\ket{\phi^\perp}$ is orthogonal to $\Lambda_T \ket{\psi}$ in the second equality, used Eq.~\eqref{eq:projection_state_when_parities} in the third equality along with the fact that  parity states are orthogonal.
\end{proof}

\paragraph{Subroutines.} Before introducing the algorithm, we present a couple of useful subroutines and lemmas that will be necessary below.
\begin{lemma}[\textsf{SWAP} test]
\label{lem:swap_test}
Let $\varepsilon,\delta \in (0,1)$. Given two arbitrary $n$-qubit quantum states $\ket{\psi}$ and~$\ket{\phi}$, there is a quantum algorithm that estimates $|\langle \psi | \phi \rangle|^2$ up to error $\varepsilon$ with probability at least $1-\delta$ using $O(1/\varepsilon^2\cdot \log(1/\delta))$ copies of $\ket{\psi},\ket{\phi}$ and which runs in $O(n/\varepsilon^2\cdot \log(1/\delta))$ time.
\end{lemma}

We will require the following characterization regarding superpositions of stabilizer states being stabilizer states as well.
\begin{lemma}[{\cite[Lemma~2]{garcia2014geometry}}]
\label{lem:superpos_stabilizers}
Let $\ket{\phi}$ and $\ket{\varphi}$ be $n$-qubit stabilizer states such that $\langle \phi | \varphi \rangle \neq 1$. Then $(\ket{\phi} + i^\ell \ket{\varphi})/\sqrt{2}$ for $\ell \in \{0,1,2,3\}$ is a stabilizer state if and only if $\langle \phi | \varphi \rangle = 0$ and there exists an $n$-qubit Pauli operator $P$ such that $\ket{\varphi} = P \ket{\phi}$.
\end{lemma}

We use the following lemma that allows for the preparation of an arbitrary  stabilizer state.
\begin{lemma}[Clifford synthesis~\cite{dehaene2003clifford,patel2003efficient}]\label{lem:clifford_synthesis}
Given the classical description of an $n$-qubit stabilizer state $\ket{\phi}$, there is a quantum algorithm that outputs a Clifford circuit $C$ that prepares $\ket{\phi}$, using $O(n^2)$ single and two-qubit Clifford gates.
\end{lemma}

\subsection{Algorithm}
In this section, we present our main boosting algorithm (Algorithm~\ref{algo:agnostic_boosting}). We refer the reader to Section~\ref{sec:intro_boosting} for a high-level description and intuition regarding our approach. The algorithm has two stages: $(i)$ \emph{structure learning} in which we will learn a set of parities $\{\ket{\chi_i}\}_{i \in [\kappa]}$ across $\kappa$ many iterations such that $|\langle \psi | (\Lambda_T \ket{\psi})|^2 = \opt$ where $T = \spann(\{\ket{\chi_i}\}_{i \in [\kappa]})$, and $(ii)$ \emph{parameter learning} where we learn the coefficients corresponding to the parity states $\ket{\chi_i}$ and thereby learn a state which is a good approximation to $\Lambda_T\ket{\psi}$. We now describe the notation and execution of different steps in these two stages below before presenting their analysis.

\vspace{-0.1in}

\paragraph{Stage 1: Structure learning.} 
In each iteration, we will denote the \emph{residual vector} as
\begin{equation}\label{eq:residual_vector}
\Psi_{t+1} = \ket{\psi} - \Lambda_{T(t)} \ket{\psi} = \ket{\psi} - \sum_{i=1}^t \beta_i \ket{\chi_i},    
\end{equation} 
and the corresponding (normalized) state upon which we carry out agnostic learning as
\begin{equation}\label{eq:residual_state}
    \ket{\psi_{t+1}} = \Psi_{t+1}/\alpha_{t+1},
\end{equation}
where we have used $\alpha_{t+1} = \norm{\Psi_{t+1}}_2$. Note that $\ket{\psi_{t+1}}$ is what we prepare during the course of Algorithm~\ref{algo:agnostic_boosting}. Considering the unnormalized vector $\Psi_{t+1}$ will, however, be useful for our \emph{analysis}.

\emph{Stopping condition.} First observe that we stop at the end of iteration $t \geq 1$ when either $|\alpha_{t+1}|^2 < \varepsilon$ or $\calF_{\Cc}(\ket{\psi_{t+1}}) < \varepsilon$. This implies that
\begin{align}
\label{eq:stoppingcondition}
|\alpha_{t+1}|^2 \cdot \calF_{\Cc}(\ket{\psi_{t+1}}) < \varepsilon.
\end{align}
If we do not stop and proceed with iteration $(t+1)$, then both $|\alpha_{t+1}|^2 \geq \varepsilon$ and $\calF_{\Cc}(\ket{\psi_{t+1}}) \geq \varepsilon$. 

\emph{State update.} At each iteration, the state in consideration is $\ket{\psi_{t}} = (\id-\Lambda_{T(t-1)})\ket{\psi}/\alpha_t$.  We now prepare the state $\ket{\psi_{t+1}}$ given  copies of $\ket{\psi}$ as follows: consider the two-outcome measurement
$$
\Big\{\id-\Lambda_{T(t)}, \Lambda_{T(t)}\Big\}.
$$
The probability of this $2$-outcome POVM giving the first outcome is given by 
$$
\langle \psi|\id-\Lambda_{T(t)}|\psi\rangle=\langle \psi|(\id-\sum_{i\in [t]} \ketbra{\chi_i}{\chi_i})|\psi\rangle=1-\sum_{i\in [t]}|\langle \psi|\chi_i\rangle|^2=1-\sum_{i\in [t]}\beta_i^2=|\alpha_{t+1}|^2\geq \varepsilon,
$$
where  the inequality used that we are in the $t$th iteration only if we did not exit the loop in the previous iterations, which occurs only if $|\alpha_{t+1}|^2>\varepsilon$. Hence with probability $\varepsilon$ we succeed in step $(10)$ in preparing the quantum state corresponding to $\ket{\psi_{t+1}}$. At this point, one can run the  weak agnostic learner on copies of $\ket{\psi_{t+1}}$  to learn the next parity state $\ket{\phi_{t+1}}$ and update $\Lambda_{T(t)}\rightarrow \Lambda_{T(t+1)}$.

\paragraph{Stage 2: Parameter learning.} In the previous stage we learned the new parity function, but in order to update our  current state $\ket{\widehat{\phi}^{(t)}}=\sum_{i\in [t]}\beta_i\ket{\chi_i}$ to $\ket{\widehat{\phi}^{(t+1)}}=\sum_{i\in [t+1]}\beta_i\ket{\chi_i}$, we also need to learn the coefficient $\beta_{t+1}$.  To do so, one could compute $\beta_{t+1} = \langle \phi_{t+1} | \psi \rangle$ via the Hadamard test using the state preparation unitaries (and their controlled versions) of $\ket{\phi_{t+1}}$ and $\ket{\psi}$. However, we show that we can avoid the need for state preparation unitary access and accomplish the task of agnostic learning by determining $\Lambda_{T(t)} \ket{\psi}$ up to a global phase, using only copies of $\ket{\psi}$. 
This is formally stated in Lemma~\ref{lem:estimate_projection_psi}.  We are now ready to present the algorithm.

\begin{myalgorithm}
\begin{algorithm}[H]
\label{algo:agnostic_boosting}
\caption{Quantum agnostic boosting}
\setlength{\baselineskip}{1.6em} 
\DontPrintSemicolon 
\KwInput{$\varepsilon \in (0,1)$, copies of  $\ket{\psi}$, weak learner $\calA_{\WAL}$ (Def.~\ref{def:WAL}) with promise $\eta(\tau) = \eta_1 \tau^{\eta_2}$.}
\KwOutput{List of parity states $L = \{\ket{\chi_i}\}_{i \in [\kappa]}$, coefficients $B = \{\beta_i\}_{i \in [\kappa] }$}
\vspace{2mm}
Set error parameters $\varepsilon_s = (2/3)^{1/\eta_2 + 1}  \varepsilon^2/16$ and $\varepsilon_p = \varepsilon/2$. \\
\Comment*[l]{Stage 1: Structure learning (Theorem~\ref{thm:structure_learning})}
Set $\ket{\psi_1}=\ket{\psi}$, $\alpha_1 = 1$, $L=\varnothing$. \\
Set parameter $\eta = \eta(\varepsilon_s)$ with $\eta(\cdot)$ being the promise of $\calA_{\WAL}$ (Theorem~\ref{thm:agnostic_boosting}). \\
Set $t_{\max} = 4/(\varepsilon_s \eta(\varepsilon_s))$, $\delta'=\delta/(3t_{\mathrm{max}})$, $\kappa=0$. \\[2mm]
\For{$t=1$ \KwTo $t_{\max}$}{\vspace{2mm}
    Run the weak agnostic learner $\calA_{\WAL}$ on $S_{\WAL}$ copies of $\ket{\psi_{t}}$ to learn a parity state~$\ket{\chi_t}$. \\
    Run \SWAP test on $O(1/\eta^2 \log(t_\mathrm{max}/\delta))$ copies of $\ket{\psi_t},\ket{\chi_t}$ such that with probability $\geq 1-\delta'$, one obtains an $\eta/2$ approximation of $|\langle \psi_t | \chi_t \rangle|^2$. Call the estimate $\nu_t$. \label{algo_step:est_fidelity}\\
    \lIf{$\nu_t < \eta$  \label{algo_step:stop_cond_fidelity}}{break loop.}
    Update $L \leftarrow L \cup \{\ket{\chi_t}\}$ and $\kappa \leftarrow \kappa + 1$. \\
    Set $\Lambda_{T(t)} = \sum_{i=1}^t \ket{\chi_i}\bra{\chi_i}$. \\ 
    Let $\widehat{\alpha}_{t+1}^2$ be an $\varepsilon_s/2$ approximation of $\alpha_{t+1}^2 := \norm{(\id-\Lambda_{T(t)})\ket{\psi}}_2^2$ by measuring $\ket{\psi}$ in the basis~$\{\id-\Lambda_{T(t)} , \Lambda_{T(t)}\}$, $O(1/\varepsilon_s^2\log(1/\delta'))$ many times. \label{algo_step:step_est_alphat}\\
    \lIf{$\widehat{\alpha}_{t+1}^2 < \varepsilon_s$}{break loop. \label{algo_step:stop_cond_ST}}
    Prepare $S_{\WAL}$ copies of $\ket{\psi_{t+1}}= (\id-\Lambda_{T(t)})\ket{\psi}/\alpha_{t+1}$ by measuring $O(S_{\WAL}/\varepsilon_s\log(1/\delta'))$ copies of $\ket{\psi}$ in the basis $\{\id-\Lambda_{T(t)}, \Lambda_{T(t)}\}$ and post-selecting for the first outcome. \\
}
\vspace{2mm}
\Comment*[l]{Stage 2: Parameter learning (Theorem~\ref{thm:parameter_learning})} \label{algo_step:parameter_learning}
Set error parameters $\upsilon_1 = (\varepsilon_p \cdot \eta)/(63\kappa)$, $\upsilon_2 = (\varepsilon_p \cdot \sqrt{\eta})/(18\kappa)$, $\upsilon' = (\varepsilon_p \cdot \sqrt{\eta})/(36\kappa)$. \\
Estimate $\xi_1$ of $|\langle \psi | \chi_1 \rangle|$ using the \textsf{SWAP} test up to error $\upsilon_1$. \\
Estimate $\xi_j$ of $|\langle\psi | \chi_j \rangle|$ for all $j \geq 2$ using the \textsf{SWAP} test up to error $\upsilon_2$. \\
\For{$j=2$ \KwTo $\kappa$}{\vspace{2mm}
    Prepare copies of $\ket{\chi_j^R} = (\ket{\chi_1} + \ket{\chi_j})/\sqrt{2}$ and $\ket{\chi_j^I} = (\ket{\chi_1} + i \ket{\chi_j})/\sqrt{2}$ which are promised to be stabilizer states, using Lemma~\ref{lem:clifford_synthesis}. \\
    Estimate $\gamma_j^R$ of $|\langle \chi_j^R | \psi \rangle|$ using $\textsf{SWAP}$ test up to error $\upsilon'$. \\
    Estimate $\gamma_j^I$ of $|\langle \chi_j^I | \psi \rangle|$ using $\textsf{SWAP}$ test up to error $\upsilon'$. \\
    Set $$a_j = \frac{2 (\gamma_j^R)^2 - \xi_1^2 - \xi_j^2}{2 \xi_1}, \enspace \text{ and } \enspace b_j = \frac{2 (\gamma_j^I)^2 - \xi_1^2 - \xi_j^2}{2 \xi_1}.$$ \\
    Set $\widehat{\beta}_j = a_j + i b_j$. \\
}
Set $\widehat{\beta}_j \leftarrow \widehat{\beta}_j/\beta$, where $\beta = \norm{\widehat{\beta}}_2$. \\
\Return List of $\kappa$ parity states $L=\{\ket{\chi_i}\}_{i}$ and their coefficients $B=\{\widehat{\beta}_i\}_i$.
\end{algorithm}
\end{myalgorithm}

\subsection{Structure learning}
In this section, we will analyze \emph{structure learning}, which is step $1$ of Algorithm~\ref{algo:agnostic_boosting}, and prove the following theorem.
\begin{theorem}[Structure learning]
\label{thm:structure_learning}
Let $\varepsilon_s,\eta_1,\delta\in (0,1)$ and $\eta_2 \geq 1$. Let $\Cc_{\Par}$ be the class of parities and let $\Cc$ be a specified function class. Suppose $\ket{\psi}$ is an unknown $n$-qubit state such that $\calF_{\Cc}(\ket{\psi}) = \opt$. 

Let $\calA_{\WAL}$ be a weak agnostic learner as defined in Theorem~\ref{thm:agnostic_boosting} with sample complexity $S_{\WAL}$, time complexity $T_{\WAL}$, and the corresponding  $\eta:[0,1] \rightarrow [0,1]$ be defined as $\eta(\tau) = \eta_1 \tau^{\eta_2}$. Then, there exists an algorithm that with probability $\geq 1-\delta$ determines a list of $\kappa \leq 4/(\varepsilon_s \cdot \eta(\varepsilon_s))$ parity states $\{\ket{\chi_i}\}_{i \in [\kappa]}$ such that $|\langle \chi_i | \psi \rangle|^2 \geq \varepsilon_s \cdot \eta(\varepsilon_s)/4$ for all $i \in [\kappa]$ and $\ket{\psi}$ can be expressed as
$$
\ket{\psi} = \Lambda_T \ket{\psi} + \alpha_{\kappa+1} \ket{\psi_{\kappa + 1}} \enspace \text{ where } \enspace |\alpha_{\kappa+1}|^2 \cdot \calF_{\Cc}(\ket{\psi_{\kappa + 1}}) < \varepsilon_s',
$$
where $T=\spann\big(\{\ket{\chi_i}\}_{i \in [\kappa]}\big)$, $\ket{\psi_{\kappa+1}}$ is orthogonal to each parity state $\ket{\chi_i}$, and $\varepsilon_s' = (3/2)^{1/\eta_2 + 1}\varepsilon_s$. Additionally, the state $\ket{\widehat{\phi}}:=\Lambda_T \ket{\psi}/\norm{\Lambda_T \ket{\psi}}$ satisfies
\begin{align}
\label{eq:guaranteeofstructurelearning}
|\langle \widehat{\phi} | \psi \rangle|^2 \geq \opt - 2 \sqrt{\varepsilon_s'}.
\end{align}
The overall complexity of this algorithm is as follows
\begin{align*}
&\text{Sample complexity: } \kappa S_{\WAL} + \widetilde{O}(\kappa/\eta(\varepsilon_s)^2 \log(1/\delta)), \\
&\text{Time complexity: } \kappa T_{\WAL} + \widetilde{O}(\kappa n/\eta(\varepsilon_s)^2 \log(1/\delta)).
\end{align*} 
\end{theorem}
The algorithm corresponding to Theorem~\ref{thm:structure_learning} is stage $1$ of Algorithm~\ref{algo:agnostic_boosting}. The proof follows the analysis in Arunachalam and Dutt~\cite{ad2025structure} which showed how to learn structured stabilizer decompositions of quantum states, and which itself is inspired by the analysis from structured decomposition results from additive combinatorics~\cite{green2006montreal,tulsiani2014quadratic,kim2023cubic}. This can also be viewed as bringing arguments from Feldman's work on classical agnostic boosting~\cite{feldman2009distribution}, applicable to Boolean functions, to the quantum setting.

We break down the proof of the theorem into two parts: we first provide an upper bound on the number of iterations the algorithm runs for, and then we prove Eq.~\eqref{eq:guaranteeofstructurelearning}, the main guarantee of the structure learning theorem.

\paragraph{Iterations.}
We need to upper bound the maximum number of iterations $\kappa$, the boosting algorithm runs for. To this end, we have the following observations regarding the promise of the parity state $\ket{\chi_i}$ learned in each iteration, and the residual vectors (Eq.~\eqref{eq:residual_state}) across consecutive iterations before we stop.

\begin{claim}\label{claim:WAL_promise}
Consider the context of Theorem~\ref{thm:structure_learning}. Let $\delta' \in (0,1)$ be the failure probability of any iteration in Algorithm~\ref{algo:agnostic_boosting}. Using a sample complexity of $O(1/\eta(\varepsilon_s)^2\log(1/\delta'))$ for the $\SWAP$ test in step~\ref{algo_step:est_fidelity} and $O(1/\varepsilon_s^2\log(1/\delta'))$ for estimating probability of preparing $\psi_{t}$ in step~\ref{algo_step:step_est_alphat}, we ensure that for each $t \leq \kappa$, we~have
$$
|\beta_t|^2 = |\langle \chi_t | \psi \rangle|^2 \geq \varepsilon_s \cdot \eta(\varepsilon_s)/4,
$$
with probability $\geq 1-\delta'$.
\end{claim}
\begin{proof}
Consider iteration $t \geq 1$. We obtain an estimate of $|\langle \phi_t | \psi_t \rangle|^2$, denoted via $\nu_t$, up to error $\eta(\varepsilon_s)/2$ with probability $\geq 1-\delta'/2$ using the $\SWAP$ with $O(1/\eta(\varepsilon_s)^2 \log(1/\delta'))$ sample complexity and $O(n/\eta(\varepsilon_s)^2 \log(1/\delta'))$ time. We also obtain an estimate of $\alpha_t := \norm{(\id - \Lambda_{T(t-1)})\ket{\psi}}_2$, denoted by $\widehat{\alpha}_t$, up to error $\varepsilon_s/2$ with probability $1-\delta'/2$ using $O(1/\varepsilon_s^2 \log(1/\delta'))$ samples and $O(n/\varepsilon_s^2 \log(1/\delta'))$ time. By a union bound, we thus ensure with probability $\geq 1-\delta'$ that 
$$
\Big| |\widehat{\alpha}_{t}| - |\alpha_t| \Big| \leq \varepsilon_s/2, \enspace \text{and} \enspace \Big| \nu_{t} - |\langle \chi_t|\psi_t \rangle|^2 | \Big| \leq \eta(\varepsilon_s)/2.
$$
If we have not exited from the loop i.e., $|\langle \chi_t | \psi_t \rangle| \geq \eta(\varepsilon_s)$ and $\widehat{\alpha}_t^2 \geq \varepsilon_s$, then the true values satisfy
\begin{equation}\label{eq:promise_true_values}
|\langle \chi_t | \psi_t \rangle|^2 \geq \eta(\varepsilon_s), \enspace \text{and} \enspace \alpha_t^2 \geq \varepsilon_s/2.    
\end{equation}
Now, using the definition of the residual state $\ket{\psi_t}$ (Eq.~\eqref{eq:residual_state}), we note that
$$
\langle \chi_t | \psi_t \rangle = \frac{\langle \chi_t | \psi_t \rangle - \langle \chi_t | (\Lambda_{T(t-1)} \ket{\psi})}{\alpha_t} = \frac{\langle \chi_t | \psi \rangle}{\alpha_t} \implies |\langle \chi_t | \psi \rangle|^2 = \alpha_t^2 |\langle \chi_t | \psi_t \rangle|^2 \geq \varepsilon_s \cdot \eta(\varepsilon_s)/4,
$$
where we have used that $\ket{\chi_t}$ is orthogonal to $\Lambda_{T(t-1)} \ket{\psi}$ (which is a linear combination of parity states $\{\ket{\chi_i}\}_{i \in [t-1]}$, distinct from $\ket{\chi_t}$) and used Eq.~\eqref{eq:promise_true_values} for the implication. This gives us the desired~result.
\end{proof}

The above claim allows us to comment on the progress made in each iteration before we stop as follows.
\begin{claim}\label{claim:diff_l2_norm_residuals}
Consider the context of Theorem~\ref{thm:structure_learning}. Let $\delta' \in (0,1)$ be the failure probability of any iteration in Algorithm~\ref{algo:agnostic_boosting}. For each $t \leq \kappa$, we have with probability $\geq 1 - \delta'$ that
$$
\norm{\ket{\psi_t}}_2^2 - \norm{\ket{\psi_{t+1}}}_2^2 \geq \varepsilon_s \eta(\varepsilon_s)/4.
$$    
\end{claim}
\begin{proof}
By direct computation, we obtain that
\begin{align*}
\norm{\ket{ \psi_{t}}}_2^2 - \norm{\ket{\psi_{t+1}}}_2^2
&= \norm{\ket{ \psi_{t}} - \ket{\psi_{t+1}} + \ket{ \psi_{t+1}}}_2^2 - \norm{\ket{\psi_{t+1}}}_2^2 \\
&= \norm{\ket{ \psi_{t}} - \ket{\psi_{t+1}}}_2^2 + 2 \mathrm{Re}(\langle \psi_{t} - \psi_{t+1} | \psi_{t+1} \rangle) + \norm{\ket{\psi_{t+1}}}_2^2 - \norm{\ket{\psi_{t+1}}}_2^2 \\
&= \norm{\ket{ \psi_{t}} - \ket{\psi_{t+1}}}_2^2 \\
&= \norm{ \beta_t \ket{\chi_t}}_2^2 \\
& \geq \varepsilon_s \eta(\varepsilon_s)/4,
\end{align*}
where we used $\langle \psi_{t} - \psi_{t+1} | \psi_{t+1} \rangle = \beta_t \alpha_{t+1} \langle \chi_t | \psi_{t+1} \rangle = 0$ as $\ket{\psi_{t+1}}$ is orthogonal to $\ket{\chi_t}$ by construction. The final inequality follows from the promise of the weak agnostic learner and the fact that the algorithm did not break in the current and any of the previous $t$ iterations (as was being checked in the algorithm), implying $|\beta_t|^2 \geq \varepsilon_s \eta(\varepsilon_s)/4$ by Claim~\ref{claim:WAL_promise}. This completes the proof.
\end{proof}

Using the above claim, we can now provide an upper bound on $\kappa$.
\begin{claim}\label{claim:ub_kappa}
Consider the context of Theorem~\ref{thm:structure_learning}. The maximum number of iterations $\kappa$ in the structure learning algorithm of Theorem~\ref{thm:structure_learning} with probability $\geq 1-\delta$ is bounded as 
$$
\kappa \leq 4/(\varepsilon_s \eta(\varepsilon_s)).
$$    
\end{claim}
\begin{proof}
Suppose the algorithm ran for $\kappa$ many iterations before stopping. Then, we have that
$$
1\geq \norm{\ket{\psi_1}}_2^2-\norm{\ket{\psi_{\kappa+1}}}_2^2= \sum_{t=1}^{\kappa} \norm{\ket{\psi_t}}_2^2- \norm{\ket{ \psi_{t+1}}}_2^2 \geq \kappa \varepsilon_s \eta(\varepsilon_s)/4 \implies \kappa \leq 4/(\varepsilon_s \eta(\varepsilon_s)),
$$
where we used that $\ket{\psi_1} = \ket{\psi}$ in the first inequality and Claim~\ref{claim:diff_l2_norm_residuals} in the third inequality. This is true with success probability $\geq 1 - \kappa \delta'$ where $\delta'$ is the failure probability of Claim~\ref{claim:diff_l2_norm_residuals}. Setting $\delta'=\delta \varepsilon_s \eta(\varepsilon_s)/4$ gives us the desired success probability. This proves the desired result.
\end{proof}

\paragraph{Guarantee of the algorithm.}
We now prove the guarantee of the algorithm as promised by Theorem~\ref{thm:structure_learning}.
\begin{proof}[Proof of Theorem~\ref{thm:structure_learning}]
Suppose the algorithm stops after $\kappa$ iterations. From Claim~\ref{claim:ub_kappa}, we have that with probability $1-\delta$, $\kappa \leq 4/(\varepsilon_s \cdot \eta(\varepsilon_s))$ and the output of the algorithm is a set of parity states $\{\ket{\chi_i}\}_{i \in [\kappa]}$. Let $T=\spann(\{\ket{\chi_i}\}_{i \in [\kappa]})$ and denote the corresponding projection of $\ket{\psi}$ on $T$ as 
\begin{equation}
\Lambda_T \ket{\psi} = \sum_{i=1}^\kappa \beta_i\ket{\chi_i},    
\end{equation}
where we have used Eq.~\eqref{eq:projection_state_when_parities} and denoted $\beta_i = \langle \chi_i | \psi \rangle$. By Fact~\ref{fact:projection}, we can express $\ket{\psi}$ as
\begin{equation}\label{eq:interim_decomp_psi}
\ket{\psi} = \Lambda_T \ket{\psi} + \alpha_{\kappa + 1} \ket{\psi_{\kappa+1}},    
\end{equation}
where $\psi_{\kappa+1}$ is the residual state (Eq.~\eqref{eq:residual_state}) after $\kappa$ iterations and is orthogonal to $\Lambda_T \ket{\psi}$. Since the algorithm stopped, we must have (steps~\ref{algo_step:est_fidelity} and \ref{algo_step:step_est_alphat} of Algorithm~\ref{algo:agnostic_boosting}) that
$$
\nu_{\kappa + 1} < \eta(\varepsilon_s) \enspace \text{or} \enspace \widehat{\alpha}_t^2 < \varepsilon_s \implies |\langle \phi_{\kappa + 1} | \psi_{\kappa + 1} \rangle|^2 < (3/2)\eta(\varepsilon_s) \enspace \text{or} \enspace \alpha_t^2 < (3/2)\varepsilon_s
$$
where we have used that $\nu_{\kappa + 1}$ is an $\eta(\varepsilon_s)/2$ estimate of $|\langle \phi_{\kappa + 1} | \psi_{\kappa + 1} \rangle|^2$ and $\widehat{\alpha}_{\kappa + 1}^2$ is an $\varepsilon_s/2$ estimate of $\alpha_{\kappa + 1}^2$. Recall that the promise of $\calA_{\WAL}$ is that if $\calF_{\Cc}(\ket{\psi_t}) \geq \tau$ then it will output $\ket{\chi_t}$ such that $|\langle \chi_t | \psi_t \rangle|^2 \geq \eta(\tau)$. Here, since $|\langle \chi_{\kappa + 1} | \psi_{\kappa + 1} \rangle|^2 < 3\eta(\varepsilon_s)/2$, we must have $\calF_{\Cc}(\ket{\psi_{\kappa+1}}) \leq (3/2)^{1/\eta_2} \varepsilon_s$ (where we have used the expression of $\eta$ assumed in the statement of Theorem~\ref{thm:agnostic_boosting}). This then implies that
$$
\alpha_{\kappa + 1}^2 \cdot \calF_{\Cc}(\ket{\psi_{\kappa + 1}} \leq (3/2)^{1/\eta_2 + 1} \varepsilon_s.
$$
This proves the first part of the theorem. Let us denote $\varepsilon_s' = (3/2)^{1/\eta_2 + 1} \varepsilon_s$ from now onward.

Towards proving the second part, let $\ket{\varphi} \in \calS_{\Cc}$ be the phase state that achieves maximal fidelity with $\ket{\psi}$ i.e., $|\langle \varphi| \psi \rangle|^2 = \opt$. By the decomposition of $\ket{\psi}$ in Eq.~\eqref{eq:interim_decomp_psi}, we  have
\begin{align*}
    |\langle \varphi| \psi \rangle| \leq |\langle \varphi|(\Lambda_T \ket{\psi})| + |\alpha_{\kappa+1}| \cdot | \langle \varphi | \psi_{\kappa + 1} \rangle| < |\langle \varphi|(\Lambda_T \ket{\psi})| + \sqrt{\varepsilon_s'} \implies |\langle \varphi| \psi \rangle| - |\langle \varphi|(\Lambda_T \ket{\psi})| < \sqrt{\varepsilon_s'},
\end{align*}
where we used $|\alpha_{\kappa+1}|\cdot | \langle \varphi | \psi_{\kappa + 1} \rangle| \leq |\alpha_{\kappa+1}| \cdot \sqrt{\calF_{\Cc}(\ket{\psi_{\kappa+1}})} < \sqrt{\varepsilon_s'}$. We can then immediately show
\begin{align}
    & |\langle \varphi| \psi \rangle|^2 - |\langle \varphi|(\Lambda_T \ket{\psi})|^2 = \Big(|\langle \varphi| \psi \rangle| + |\langle \varphi|(\Lambda_T \ket{\psi})|\Big)\Big(|\langle \varphi| \psi \rangle| - |\langle \varphi|(\Lambda_T \ket{\psi})|\Big) \leq 2 \sqrt{\varepsilon_s'} \\
    &\implies |\langle \varphi|(\Lambda_T \ket{\psi})|^2 \geq |\langle \varphi| \psi \rangle|^2 - 2 \sqrt{\varepsilon_s'} = \opt - 2 \sqrt{\varepsilon_s'},
    \label{eq:fidelity_promise_projection}
\end{align}
where we have used $|\langle \varphi| \psi \rangle|, |\langle \varphi|(\Lambda_T \ket{\psi})| \leq 1$ and $|\langle \varphi| \psi \rangle|^2 = \opt$ in the final implication.

In order to solve the task of agnostic learning, define the quantum state $\ket{\widehat{\phi}}=\Lambda_T \ket{\psi}/\norm{\Lambda_T \ket{\psi}}$. Now observe that
\begin{align*}
|\langle \widehat{\phi}|\psi\rangle|^2 = \left|\langle \widehat{\phi}|\Lambda_T \ket{\psi}\rangle + r_\kappa \langle \widehat{\phi}|\psi_{\kappa+1}\rangle \right|^2
&= |\langle \widehat{\phi}| \Lambda_T \ket{\psi}|^2 \\
&= |\langle \widehat{\phi}|\widehat{\phi}\rangle| \cdot \norm{ \Lambda_T \ket{\psi}}_2^2 \\
&\geq |\langle \widehat{\phi}|\varphi\rangle|^2 \cdot \norm{ \Lambda_T \ket{\psi}}_2^2 \\
&=\frac{|\langle\varphi| \Lambda_T \ket{\psi}|^2}{\norm{ \Lambda_T \ket{\psi}}_2^2} \cdot \norm{ \Lambda_T \ket{\psi}}_2^2 \\
&=|\langle\varphi| \Lambda_T \ket{\psi}|^2 \\
&\geq \opt - 2 \sqrt{\varepsilon_s'},
\end{align*}
where the first equality used the definition of $\ket{\psi}$ in the theorem statement, second equality used that $ \Lambda_T \ket{\psi},\ket{\psi_{\kappa+1}}$ are orthogonal, third equality used the definition of $\ket{\widehat{\phi}}$, the inequality works for \emph{every} state (and in particular $\ket{\varphi}$) and the last inequality used Eq.~\eqref{eq:fidelity_promise_projection}. This proves Eq.~\eqref{eq:guaranteeofstructurelearning} in the theorem statement.

To conclude the theorem proof, we observe that the main contribution to the sample complexity is running $\calA_{\WAL}$ in each of the $\kappa$ many iterations which consumes $O(\kappa S_{\WAL})$ overall, $\SWAP$ tests (Claim~\ref{claim:WAL_promise}) which consumes $O(\kappa/\eta(\varepsilon_s)^2\log(\kappa/\delta))$  and estimation of the $\norm{(\id - \Lambda_{T(t)})\ket{\psi}}$ (Claim~\ref{claim:WAL_promise}) which consumes $O(\kappa/\varepsilon_s\log(\kappa/\delta))$. The corresponding time complexities are $\kappa T_{\WAL}$, $O(n\kappa/\eta(\varepsilon_s)^2\log(\kappa/\delta))$, and $O(n \kappa/\varepsilon_s\log(\kappa/\delta))$ respectively. The overall sample and time complexities are then as stated. This completes the proof.
\end{proof}

\subsection{Parameter learning}
In the previous section we showed how to learn a set of parities $\{\ket{\chi_i}\}_{i\in [\kappa]}$ which were used as a basis to construct the state $\ket{\widehat{\phi}}$ that achieved the $\opt-\varepsilon$ fidelity lower bound. In this section, we show how to learn the \emph{coefficients} of these parities in order to construct the state $\ket{\widehat{\phi}}$. Crucial to proving our main theorem is the following lemma. 
In particular, for an arbitrary quantum state $\ket{\psi}$ we show how to determine the projection of the state onto  $T = \spann(\{\ket{\chi_i}\}_{i \in [k]})$ i.e., $\Lambda_T \ket{\psi}$ (Eq.~\eqref{eq:projection_state_when_parities}) up to a global phase using only copies of $\ket{\psi}$. We state the lemma in full generality below since we will use it as a blackbox in another context.

\begin{lemma}
\label{lem:estimate_projection_psi}
Let $\kappa \in \mathbb{N}$ and $\varepsilon,\eta,\delta \in (0,1)$. Suppose $\ket{\psi}$ is an unknown $n$-qubit state. Let $\{\ket{\chi_i}\}_{i \in [k]}$ be a list of known parity states such that $|\langle\chi_i | \psi \rangle|^2 \geq \mu$ for all $i\in [k]$. There is an algorithm that, with probability $\geq 1 - \delta$, outputs  $\widehat{\beta}\in \calB_\infty^k$
such~that
$$
\left|\langle \psi | \left(\sum_{i=1}^k \widehat{\beta_i} \ket{\chi_i}\right)\right|^2 \geq |\langle \psi | (\Lambda_T\ket{\psi})|^2 - \varepsilon,
$$
where $T = \spann(\{\ket{\chi_i}\}_{i \in [k]})$. Additionally, $\norm{\widehat{\beta}}_2^2 \leq |\langle \psi | (\Lambda_T \ket{\psi})|^2 + \varepsilon$. The complexity of the algorithm is: 
\begin{align*}
&\text{ Sample complexity: }O(k^3/(\varepsilon^2 \cdot \mu^2) \log(k/\delta))\\
&\text{ Time complexity: } O(k^3 n^2/(\varepsilon^2 \cdot \mu^2) \cdot \log(k/\delta)).
\end{align*}
\end{lemma}
We now state the main theorem, which is given as stage $2$ of Algorithm~\ref{algo:agnostic_boosting}. 
\begin{theorem}[Parameter learning]
\label{thm:parameter_learning}
Let $\kappa \in \mathbb{N}$ and $\varepsilon_p,\mu,\delta \in (0,1)$. Suppose $\ket{\psi}$ is an unknown state such that $\calF_{\Cc}(\ket{\psi}) = \opt$. Let $\{\ket{\chi_i}\}_{i \in [\kappa]}$ be a list of known parity states such that $|\langle\chi_i | \psi \rangle|^2 \geq \mu$ and $\ket{\phi} := \Lambda_T\ket{\psi}/\norm{\Lambda_T \ket{\psi}}$ where $T=\spann\big(\{\ket{\chi_i}\}_{i \in [\kappa]}\big)$, satisfies
$$
|\langle \phi | \psi \rangle |^2 \geq \opt - \varepsilon_p.
$$
Then, there exists an algorithm outputs coefficients $\{c_i\}_{i \in [\kappa]}$ such that $\ket{\widehat{\phi}} = \sum_{i=1}^\kappa c_i \ket{\chi_i}$ satisfies
$$
|\langle \widehat{\phi} | \psi \rangle|^2 \geq \opt - 2 \varepsilon_p,
$$
with probability at least $1-\delta$. The complexity of the algorithm is as~follows:
\begin{align*}
    &\text{ Sample complexity: }\widetilde{O}(\kappa/(\varepsilon_p^2 \cdot \mu^6) \log(1/\delta))\\
    &\text{ Time complexity: }  \widetilde{O}(\kappa n^2/( \varepsilon_p^2 \cdot \mu^6) \log(1/\delta)).
\end{align*}
\end{theorem}

\begin{proof}
Recall that in~Eq.~\eqref{eq:projection_state}  we defined $\Lambda_T \ket{\psi} = \sum_{i=1}^\kappa \beta_i \ket{\chi_i}$, where $T = \spann(\{\ket{\chi_i}\}_{i \in [\kappa]})$ and $\beta_i = \langle \chi_i | \psi \rangle$. We are promised $|\langle \psi | \chi_i \rangle|^2 \geq \mu$ for all $i \in [\kappa]$ and that the state $\ket{\phi} = \Lambda_T\ket{\psi}/\norm{\Lambda_T \ket{\psi}}$ solves the task of agnostic learning i.e., $|\langle \widehat{\phi} | \psi \rangle|^2 \geq \opt - \varepsilon_p$. The idea is to then use Lemma~\ref{lem:estimate_projection_psi} to obtain an approximation of $\ket{\widehat{\phi}}$ which will be the eventual output of the agnostic learner.

Let $\gamma \in (0,1)$ be a parameter to be decided later. Using Lemma~\ref{lem:estimate_projection_psi} with $O(\kappa^3/(\gamma^2 \cdot \mu^2) \log(\kappa/\delta))$ sample complexity and $O(\kappa^3 n^2/( \gamma^2 \cdot \mu)^2 \log(\kappa/\delta))$ time complexity, we can determine a list of coefficients $\{\widehat{\beta}_i\}_{i \in [\kappa]}$ such that
\begin{equation}
\label{eq:interim_fidelity}
\left|\langle \psi \Big| \Big(\sum_{i=1}^k \widehat{\beta}_i \ket{\chi_i} \Big) \right|^2 \geq |\langle \psi | (\Lambda_T \ket{\psi}) \rangle|^2 - \gamma,    
\end{equation}
and we are guaranteed
\begin{align}\label{eq:ub_L2_norm_beta}
    \norm{\widehat{\beta}}_2^2 \leq |\langle \psi | (\Lambda_T \ket{\psi}) \rangle|^2 + \gamma .
\end{align}
Consider the state $\ket{\widehat{\phi}}$ defined as
\begin{equation}
    \ket{\widehat{\phi}} = \sum_{i=1}^\kappa c_i \ket{\chi_i},
\end{equation}
where $c_i = \widehat{\beta}_i/\norm{\widehat{\beta}}_2, \, \forall i \in [\kappa]$. Note that $\ket{\widehat{\phi}}$ is a valid normalized state as $\norm{\ket{\widehat{\phi}}} = 1$. We then have
\begin{align*}
    |\langle \psi | \widehat{\phi} \rangle|^2 = \frac{\left|\langle \psi \Big| \Big(\sum_{i=1}^k \widehat{\beta}_i \ket{\chi_i} \Big) \right|^2}{\norm{\widehat{\beta}}_2^2} 
    &\geq \frac{|\langle \psi | (\Lambda_T \ket{\psi}) \rangle|^2 - \gamma}{|\langle \psi | (\Lambda_T \ket{\psi}) \rangle|^2 + \gamma} \\
    &= \frac{|\langle \psi | (\Lambda_T \ket{\psi}) \rangle|^2 - \gamma}{\norm{\Lambda_T \ket{\psi}}_2^4 + \gamma} \\
    &\geq \frac{|\langle \psi | (\Lambda_T \ket{\psi}) \rangle|^2 - \gamma}{\norm{\Lambda_T \ket{\psi}}_2^2 + \gamma} \\
    &= \frac{|\langle \psi | (\Lambda_T \ket{\psi}) \rangle|^2 - \gamma}{\norm{\Lambda_T \ket{\psi}}_2^2\Big(1 + \frac{\gamma}{\norm{\Lambda_T \ket{\psi}}_2^2}\Big)} \\
    &\geq \frac{|\langle \psi | (\Lambda_T \ket{\psi}) \rangle|^2 - \gamma}{\norm{\Lambda_T \ket{\psi}}_2^2} \Big(1 - \frac{\gamma}{\norm{\Lambda_T \ket{\psi}}_2^2}\Big) \\
    &\geq \frac{|\langle \psi | (\Lambda_T \ket{\psi}) \rangle|^2}{\norm{\Lambda_T \ket{\psi}}_2^2} - \frac{\gamma}{\norm{\Lambda_T \ket{\psi}}_2^2} - \frac{\gamma|\langle \psi | (\Lambda_T \ket{\psi}) \rangle|^2}{\norm{\Lambda_T \ket{\psi}}_2^4} \\
    &\geq \opt - \varepsilon_p - \frac{\gamma}{\kappa \mu} - \frac{\gamma}{\kappa^2\mu^2} \\
    &\geq \opt - \varepsilon_p - 2 \frac{\gamma}{\kappa \mu^2},
\end{align*}
where we used Eq.~\eqref{eq:interim_fidelity} and Eq.~\eqref{eq:ub_L2_norm_beta} in the second inequality in the first line, the fact that $|\langle \psi |(\Lambda_T \ket{\psi})| = \norm{\Lambda_T \ket{\psi}}_2^2$ (which can be observed from Fact~\ref{fact:projection}) in the second line, $\norm{\Lambda_T \ket{\psi}}_2^2 \leq 1$ in the third line, $1/(1+x) \geq 1 - x, \,\,\forall x \geq 0$ in the fifth line,
Eq.~\eqref{eq:fidelity_promise_projection} in the seventh inequality along with the observation
$$
\norm{\Lambda_T \ket{\psi}}_2^2 = \sum_{i\in [\kappa]} |\beta_i|^2 \geq \kappa \mu,
$$
since we are given $|\beta_i|^2 \geq \mu$ for all $i \in [\kappa]$, and noting that $\mu \in (0,1]$ in the final inequality. Setting $\gamma = \varepsilon_p \kappa \mu^2 /2$ gives us the desired result. The sample complexity and the time complexity is due to the use of Lemma~\ref{lem:estimate_projection_psi} with error parameter of $\gamma$ as decided. 
\end{proof}

It remains to prove Lemma~\ref{lem:estimate_projection_psi} which we do now.

\begin{proof}[{Proof of Lemma~\ref{lem:estimate_projection_psi}}]
Recall from Eq.~\eqref{eq:projection_state_when_parities} that the projection of $\ket{\psi}$ onto $\spann(\{\ket{\chi_i}\}_i)$  is
\begin{equation*}
    \Lambda_T \ket{\psi} = \sum_{i=1}^k \beta_i \ket{\chi_i} \enspace{ where } \enspace \beta_i =  \langle \chi_i | \psi \rangle.
\end{equation*}
Let $\tau = |\langle \psi | \left(\sum_{i=1}^k \beta_i \ket{\chi_i}\right)|^2$. Denoting $\beta_1 = |\beta_1|e^{i \theta_1}$ where $\theta_1$ is the angle corresponding to the phase of $\beta_1$, we observe the following is true as well
\begin{equation}
\label{eq:tau_pseudo_global_phase}
    \tau = \left|\langle \psi \Big| \left(\sum_{i=1}^k \beta_i e^{-i \theta_1} \ket{\chi_i}\right)\right|^2,
\end{equation}
since $e^{-i \theta_1}$ is simply a global phase. Let us denote $\widetilde{\beta}_i = \beta_i e^{-i \theta_1}$. Note that in particular, $\widetilde{\beta}_1 = |\beta_1|$. From Eq.~\eqref{eq:tau_pseudo_global_phase}, we have that the coefficients $\widetilde{\beta}_i$ also satisfy
\begin{equation}\label{eq:fidelity_betatilde}
|\langle \psi | (\Lambda_T\ket{\psi})|^2 = \left|\langle \psi | \left(\sum_{i=1}^k \widetilde{\beta_i} \ket{\chi_i}\right)\right|^2 =\tau.   
\end{equation}
It is then enough to obtain estimates of $\widetilde{\beta}_i$, which we will denote as $\widehat{\beta}_i$ that satisfies the guarantee of the theorem. To this end, let $\upsilon_1 \in (0,1)$ be a fixed error parameter to be decided later and in particular we will choose it to be $\leq \sqrt{\mu}/2$ as will be seen shortly. We will use the following algorithm to estimate $\widetilde{\beta}_j$.

For $j=1$, obtain an estimate $\widehat{\beta}_1$ of $|\beta_1| = |\langle \chi_1 | \psi \rangle|$ using the \SWAP test that uses $O(1/\upsilon_1^2 \log(k/\delta))$ copies of $\ket{\psi}$ and with probability at least $1-  \delta/k$, outputs an estimate of $|\beta_1|$ up to error $\upsilon_1$ .
    
For $j\geq2$, we obtain estimates $\widehat{\beta}_j$ of $\widetilde{\beta}_j$ using the following procedure. For all $j\geq 2$, define
\begin{equation}\label{eq:stabilizers_phi1_phij}
    \ket{\chi_j^R} = \frac{\ket{\chi_1} + \ket{\chi_j}}{\sqrt{2}}, \quad \ket{\chi_j^I} = \frac{\ket{\chi_1} + i \ket{\chi_j}}{\sqrt{2}}.    
\end{equation}
Note that $\langle \chi_1 | \chi_j \rangle = 0$ as these are distinct parities, hence the ``real" and ``imaginary" quantum states $\ket{\chi_j^R} ,\ket{\chi_j^I} $ are valid quantum states. Moreover there exists an $a \in \{0,1\}^n$ such that $\ket{\chi_j} = Z^a \ket{\chi_1}$, which implies that $\ket{\chi_j^R}$ and $\ket{\chi_j^I}$ are stabilizer states (using Lemma~\ref{lem:superpos_stabilizers}), which can be prepared efficiently using Lemma~\ref{lem:clifford_synthesis}.

We now observe that
\begin{align}
    |\langle \chi_j^R | \psi \rangle|^2 = \frac{1}{2} \left| \langle \chi_1 | \psi \rangle + \langle \chi_j | \psi \rangle \right|^2 &= \frac{1}{2}\left[ |\langle \chi_1 | \psi \rangle|^2 + |\langle \chi_j | \psi \rangle|^2 + 2\mathrm{Re}(\langle \chi_j | \psi \rangle \overline{\langle \chi_1 | \psi \rangle}) \right] \\
    &= \frac{1}{2}\left[ |\beta_1|^2 + |\beta_j|^2 + 2\mathrm{Re}(\beta_j |\beta_1| e^{-i \theta_1}) \right],
\end{align}
which after rearrangement gives
\begin{equation}\label{eq:est_real_betaj}
      \frac{2 |\langle \chi_j^R | \psi \rangle|^2 - |\langle \chi_1 | \psi \rangle|^2 - |\langle \chi_j | \psi \rangle|^2}{2|\langle \chi_1 | \psi \rangle|}=\mathrm{Re}(\beta_j e^{-i \theta_1})=\mathrm{Re(\widetilde{\beta}_j)}, 
\end{equation}
where the second equality is by definition of $\widetilde{\beta}_j$.      Thus one obtain an estimate $\mathrm{Re}(\widehat{\beta_j})$ of $\mathrm{Re(\widetilde{\beta}_j)} $ using the expression above     and estimating each term $|\langle \chi_j^R | \psi \rangle|^2$, $|\langle \chi_1 | \psi \rangle|^2$, $|\langle \chi_j | \psi \rangle|^2$ by using the \SWAP test between corresponding states in each term. Similarly, we can obtain an estimate $\mathrm{Im}(\widehat{\beta_j})$ of $\mathrm{Im(\widetilde{\beta}_j)} = \mathrm{Im}(\beta_j e^{-i \theta_1})$ using the expression
\begin{equation}\label{eq:est_imag_betaj}
    \mathrm{Im}(\beta_j e^{-i \theta_1}) = \frac{2 |\langle \chi_j^I | \psi \rangle|^2 - |\langle \chi_1 | \psi \rangle|^2 - |\langle \chi_j | \psi \rangle|^2}{2|\langle \chi_1 | \psi \rangle|},
\end{equation}
and again estimating each term involved with the \SWAP test. 

In an ideal world when one can run the $\SWAP$ test without errors, the above procedures suffice to estimate $\widetilde{\beta_j}$ and Eq.~\eqref{eq:fidelity_betatilde} implies the theorem statement. However, note that $\SWAP$ test has measurement errors and we discuss the errors up to which we should estimate these terms and the sample complexity required so that one can upper bound
\begin{align}
\label{eq:maintermtobound}
 |\langle \psi | (\Lambda_T\ket{\psi})|^2-\left|\langle \psi | \left(\sum_{i=1}^k \widetilde{\beta_i} \ket{\chi_i}\right)\right|^2.
\end{align}
To this end, let us first bound the error $|\widehat{\beta}_j - \widetilde{\beta}_j|$. Suppose we run $\SWAP$ test to estimate  $|\langle \chi_j^R | \psi \rangle|, |\langle \chi_j | \psi \rangle|$ upto error $\upsilon_j \in (0,1)$ and $\upsilon_j' \in (0,1)$ (which we fix later) respectively. In particular, if the $\SWAP$ test outputs $\gamma_j,\alpha_j$ respectively, then we have that
\begin{equation}
\label{eq:error_estimates}
    \Big|\gamma_j - |\langle \chi_j^R | \psi \rangle|\Big| \leq \upsilon_j', \, \quad \Big|\alpha_j - |\langle \chi_j | \psi \rangle|\Big| \leq \upsilon_j, \quad \text{ for all } j \in [k],
\end{equation}
 We then have from Eq.~\eqref{eq:est_real_betaj} that
\begin{align}\label{eq:error_est_real_betaj}
  \Big|  \mathrm{Re}(\widehat{\beta}_j) - \mathrm{Re}(\widetilde{\beta}_j)\Big| &= \Big|\left[\frac{\gamma_j^2}{\alpha_1} - \frac{\alpha_1}{2} - \frac{\alpha_j^2} {2\alpha_1}\right]-\left[ \frac{|\langle \chi_j^R | \psi \rangle|^2}{|\langle \chi_1 | \psi \rangle|} - \frac{|\langle \chi_1 | \psi \rangle|}{2}-\frac{|\langle \chi_j | \psi \rangle|^2}{2|\langle \chi_1 | \psi \rangle|} \right]\Big|\\
    &\leq \underbrace{\Big|\left[\frac{\gamma_j^2}{\alpha_1} - \frac{|\langle \chi_j^R | \psi \rangle|^2}{|\langle \chi_1 | \psi \rangle|}\right]\Big|}_{(i)} + \underbrace{\Big|\left[- \frac{\alpha_1}{2} + \frac{|\langle \chi_1 | \psi \rangle|}{2}\right]\Big|}_{(ii)} + \underbrace{\Big|\left[- \frac{\alpha_j^2} {2\alpha_1} + \frac{|\langle \chi_j | \psi \rangle|^2}{2|\langle \chi_1 | \psi \rangle|} \right]\Big|}_{(iii)}.
\end{align}
Let us now bound each pair of terms defined by $(i),(ii),(iii)$ above, individually.

\noindent $(i)$ From direct computation,
\begin{align}
    \frac{\gamma_j^2}{\alpha_1} - \frac{|\langle \chi_j^R | \psi \rangle|^2}{|\langle \chi_1 | \psi \rangle|} 
    \leq \frac{\gamma_j^2}{|\langle \chi_1 | \psi \rangle| - \upsilon_1} - \frac{|\langle \chi_j^R | \psi \rangle|^2}{|\langle \chi_1 | \psi \rangle|} 
    &= \frac{\gamma_j^2}{|\langle \chi_1 | \psi \rangle|\left(1 - \frac{\upsilon_1}{|\langle \chi_1 | \psi \rangle|}\right)} - \frac{|\langle \chi_j^R | \psi \rangle|^2}{|\langle \chi_1 | \psi \rangle|} \\ 
    &\leq \frac{\gamma_j^2}{|\langle \chi_1 | \psi \rangle|}\left(1 + \frac{2\upsilon_1}{|\langle \chi_1 | \psi \rangle|}\right) - \frac{|\langle \chi_j^R | \psi \rangle|^2}{|\langle \chi_1 | \psi \rangle|} \\
    &= \frac{\gamma_j^2 - |\langle \chi_j^R | \psi \rangle|^2}{|\langle \chi_1 | \psi \rangle|} + \frac{2\upsilon_1 \gamma_j^2}{|\langle \chi_1 | \psi \rangle|^2} \\
    &\leq \frac{2 \upsilon_j'}{\sqrt{\mu}} + \frac{2\upsilon_1}{\mu}, \label{eq:ub_star}
\end{align}
where in the second line we used that $1/(1-x) \leq 1 + 2x$ for all $x \in [0,1/2]$ and noted that by choosing $\upsilon_1 \leq \sqrt{\mu}/2$, we can ensure $\upsilon_1/|\langle \chi_1 | \psi \rangle| \leq 1/2$ as $|\langle \chi_1 | \psi \rangle| \geq \sqrt{\mu}$ by assumption (in the lemma statement). In the final line, we used that $\gamma_j^2 \leq 1$ and
$$
\gamma_j^2 - |\langle \chi_j^R | \psi \rangle|^2 = (\gamma_j - |\langle \chi_j^R | \psi \rangle|)\cdot (\gamma_j + |\langle \chi_j^R | \psi \rangle|) \leq 2 \upsilon_j',
$$ 
by Eq.~\eqref{eq:error_estimates}     and $|\langle \chi_1 | \psi \rangle|^2 \geq \mu$. We proceed similarly to obtain a lower bound:
\begin{align}
    \frac{\gamma_j^2}{\alpha_1} - \frac{|\langle \chi_j^R | \psi \rangle|^2}{|\langle \chi_1 | \psi \rangle|} 
    \geq \frac{\gamma_j^2}{|\langle \chi_1 | \psi \rangle| + \upsilon_1} - \frac{|\langle \chi_j^R | \psi \rangle|^2}{|\langle \chi_1 | \psi \rangle|} 
    &= \frac{\gamma_j^2}{|\langle \chi_1 | \psi \rangle|\left(1 + \frac{\upsilon_1}{|\langle \chi_1 | \psi \rangle|}\right)} - \frac{|\langle \chi_j^R | \psi \rangle|^2}{|\langle \chi_1 | \psi \rangle|} \\ 
    &\geq \frac{\gamma_j^2}{|\langle \chi_1 | \psi \rangle|}\left(1 - \frac{\upsilon_1}{|\langle \chi_1 | \psi \rangle|}\right) - \frac{|\langle \chi_j^R | \psi \rangle|^2}{|\langle \chi_1 | \psi \rangle|} \\
    &= \frac{\gamma_j^2 - |\langle \chi_j^R | \psi \rangle|^2}{|\langle \chi_1 | \psi \rangle|} - \frac{\upsilon_1 \gamma_j^2}{|\langle \chi_1 | \psi \rangle|^2} \\
    &\geq \frac{-2 \upsilon_j'}{\sqrt{\mu}} - \frac{\upsilon_1}{\mu},\label{eq:lb_star}
\end{align}
where in the second line we used $1/(1+x) \geq 1-x$ for $x \geq 0$. In the final line, we used that $\gamma_j^2 \leq 1$, 
$$
\gamma_j^2 - |\langle \chi_j^R | \psi \rangle|^2 = (\gamma_j - |\langle \chi_j^R | \psi \rangle|)\cdot (\gamma_j + |\langle \chi_j^R | \psi \rangle|) \geq - \upsilon_j' (\gamma_j + |\langle \chi_j^R | \psi \rangle|) \geq - 2 \upsilon_j',
$$
by Eq.~\eqref{eq:error_estimates} and $|\langle \chi_1 | \psi \rangle|^2 \geq \mu$. Combining Eq.~\eqref{eq:ub_star} and Eq.~\eqref{eq:lb_star}, we have
\begin{equation}\label{eq:bound_star}
    (i) \leq \frac{2 \upsilon_j'}{\sqrt{\mu}} + \frac{2\upsilon_1}{\mu}.
\end{equation}
$(ii)$ We immediately have that
\begin{equation}\label{eq:bound_star2}
    (ii) = \Big| \alpha_1 -  |\langle \chi_1 | \psi \rangle| \Big|/2 \leq \upsilon_1/2.
\end{equation}
$(iii)$ Again by direct computation, we have
\begin{align}
    - \frac{\alpha_j^2} {2\alpha_1} + \frac{|\langle \chi_j | \psi \rangle|^2}{2|\langle \chi_1 | \psi \rangle|} 
    \leq - \frac{\alpha_j^2} {2|\langle \chi_1 | \psi \rangle| + 2 \upsilon_1} + \frac{|\langle \chi_j | \psi \rangle|^2}{2|\langle \chi_1 | \psi \rangle|}
    &= - \frac{\alpha_j^2} {2|\langle \chi_1 | \psi \rangle|\left(1 + \frac{\upsilon_1}{|\langle \chi_1 | \psi \rangle|}\right)} + \frac{|\langle \chi_j | \psi \rangle|^2}{2|\langle \chi_1 | \psi \rangle|} \\
    &\leq - \frac{\alpha_j^2} {2|\langle \chi_1 | \psi \rangle|}\left(1 - \frac{\upsilon_1}{|\langle \chi_1 | \psi \rangle|}\right) + \frac{|\langle \chi_j | \psi \rangle|^2}{2|\langle \chi_1 | \psi \rangle|} \\
    &= \frac{- \alpha_j^2 + |\langle \chi_j | \psi \rangle|^2} {2|\langle \chi_1 | \psi \rangle|} + \frac{\alpha_j^2\upsilon_1}{2|\langle \chi_1 | \psi \rangle|^2} \\
    &\leq \frac{\upsilon_j} {\sqrt{\mu}} + \frac{\upsilon_1}{2\mu}, \label{eq:ub_star3}
\end{align}
where we used in the second line that $1/(1+x) \geq 1-x, $ for all $x \geq 0$. In the final line, we used $\alpha_j^2 \leq 1$, $- \alpha_j^2 + |\langle \chi_j | \psi \rangle|^2 = (|\langle \chi_j | \psi \rangle| - \alpha_j)\cdot (|\langle \chi_j | \psi \rangle| + \alpha_j) \leq  2 \upsilon_j$ and $|\langle \chi_1 | \psi \rangle|^2 \geq \mu$. We again proceed similarly to obtain a lower bound:
\begin{align}
    - \frac{\alpha_j^2} {2\alpha_1} + \frac{|\langle \chi_j | \psi \rangle|^2}{2|\langle \chi_1 | \psi \rangle|} 
    \geq - \frac{\alpha_j^2} {2|\langle \chi_1 | \psi \rangle| - 2 \upsilon_1} + \frac{|\langle \chi_j | \psi \rangle|^2}{2|\langle \chi_1 | \psi \rangle|}
    &= - \frac{\alpha_j^2} {2|\langle \chi_1 | \psi \rangle|\left(1 - \frac{\upsilon_1}{|\langle \chi_1 | \psi \rangle|}\right)} + \frac{|\langle \chi_j | \psi \rangle|^2}{2|\langle \chi_1 | \psi \rangle|} \\
    &\geq - \frac{\alpha_j^2} {2|\langle \chi_1 | \psi \rangle|}\left(1 + \frac{2\upsilon_1}{|\langle \chi_1 | \psi \rangle|}\right) + \frac{|\langle \chi_j | \psi \rangle|^2}{2|\langle \chi_1 | \psi \rangle|} \\
    &= \frac{- \alpha_j^2 + |\langle \chi_j | \psi \rangle|^2} {2|\langle \chi_1 | \psi \rangle|} - \frac{2\alpha_j^2\upsilon_1}{2|\langle \chi_1 | \psi \rangle|^2} \\
    &\geq - \frac{\upsilon_j} {\sqrt{\mu}} - \frac{\upsilon_1}{\mu}, \label{eq:lb_star3}
\end{align}
where in the second line we again used that $1/(1-x) \leq 1 + 2x$ for all $x \in [0,1/2]$ and noted that by choosing $\upsilon_1 \leq \sqrt{\mu}/2$, we can ensure $\upsilon_1/|\langle \chi_1 | \psi \rangle| \leq 1/2$ as $|\langle \chi_1 | \psi \rangle| \geq \sqrt{\mu}$ by assumption. In the final line, we used that $\alpha_j^2 \leq 1$,
$$ 
- \alpha_j^2 + |\langle \chi_j | \psi \rangle|^2 = (- \alpha_j + |\langle \chi_j | \psi \rangle|)\cdot (\alpha_j + |\langle \chi_j | \psi \rangle|) \geq - \upsilon_j (\alpha_j + |\langle \chi_j^R | \psi \rangle|) \geq - 2 \upsilon_j,
$$
and $|\langle \chi_1 | \psi \rangle|^2 \geq \mu$. Combining Eq.~\eqref{eq:ub_star3} and Eq.~\eqref{eq:lb_star3}, we have
\begin{equation}\label{eq:bound_star3}
    (iii) \leq \frac{\upsilon_j} {\sqrt{\mu}} + \frac{\upsilon_1}{\mu}.
\end{equation}
Finally, substituting Eqs~\eqref{eq:bound_star},\eqref{eq:bound_star2},\eqref{eq:bound_star3}, into Eq.~\eqref{eq:error_est_real_betaj}, we have
\begin{equation}\label{eq:bound_est_real_betaj}
    |\mathrm{Re}(\widehat{\beta}_j) - \mathrm{Re}(\widetilde{\beta}_j)| \leq (i)+(ii)+(iii) \leq \frac{2 \upsilon_j'}{\sqrt{\mu}} + \frac{2\upsilon_1}{\mu} + \frac{\upsilon_1}{2} + \frac{\upsilon_j} {\sqrt{\mu}} + \frac{\upsilon_1}{\mu} \leq \frac{2 \upsilon_j'}{\sqrt{\mu}} + \frac{\upsilon_j} {\sqrt{\mu}} + \frac{7\upsilon_1}{2\mu}.
\end{equation}
By choosing $\upsilon_j' = (\varepsilon \cdot \sqrt{\mu})/(36k),\upsilon_j = (\varepsilon \cdot \sqrt{\mu})/(18k)$ for all $j \geq 2$, and $\upsilon_1 = (\varepsilon \cdot \mu)/(63k)$, we obtain
\begin{equation}\label{eq:final_error_bound_est_real_betaj}
    |\mathrm{Re}(\widehat{\beta}_j) - \mathrm{Re}(\widetilde{\beta}_j)| \leq \frac{\varepsilon}{6k}.
\end{equation}
Similarly, it can be shown that by estimating $|\langle \chi_j^I | \psi \rangle|^2$ to error $\upsilon_j' = (\varepsilon \cdot \sqrt{\mu})/(12k)$ for all $j \in [k]$ as just defined, we would have
\begin{equation}\label{eq:final_error_bound_est_imag_betaj}
    |\mathrm{Im}(\widehat{\beta}_j) - \mathrm{Im}(\widetilde{\beta}_j)| \leq \frac{\varepsilon}{6k}.
\end{equation}
We then have
\begin{equation}\label{eq:est_L1_norm_beta}
    \sum_{j=1}^k \Big| \widehat{\beta}_i - \widetilde{\beta}_i \Big| \leq \sum_{j=1}^k \left(|\mathrm{Re}(\widehat{\beta}_j) - \mathrm{Re}(\widetilde{\beta}_j)| + |\mathrm{Im}(\widehat{\beta}_j) - \mathrm{Im}(\widetilde{\beta}_j)| \right) \leq \varepsilon/3.
\end{equation}
Let us define the states $\ket{\widehat{\phi}} = \sum_{i=1}^k \widehat{\beta}_i \ket{\chi_i}$ and $\ket{\widetilde{\chi}} = \sum_{i=1}^k \widetilde{\beta}_i \ket{\chi_i}$, which are not necessarily normalized. Eq.~\eqref{eq:est_L1_norm_beta} then implies
\begin{equation}\label{eq:diff_inner_prods}
    |\langle \psi | \widehat{\phi} \rangle - \langle \psi | \widetilde{\phi} \rangle| = \Big|\sum_{i=1}^k (\widehat{\beta}_i - \widetilde{\beta}_i) \langle \psi | \chi_i \rangle\Big| \leq \sum_{i=1}^k |\widehat{\beta}_i - \widetilde{\beta}_i| \leq \varepsilon/3.
\end{equation}
We now note that
\begin{align}
    |\langle \psi | \widetilde{\phi} \rangle|^2 - |\langle \psi | \widehat{\phi} \rangle|^2 
    &= \Big(|\langle \psi | \widetilde{\phi} \rangle| + |\langle \psi | \widehat{\phi} \rangle|\Big)\Big(|\langle \psi | \widetilde{\phi} \rangle| - |\langle \psi | \widehat{\phi} \rangle|\Big) \\
    &\leq \Big(1 + |\langle \psi | \widehat{\phi} \rangle|\Big)\Big(|\langle \psi | \widetilde{\phi} \rangle - \langle \psi | \widehat{\phi} \rangle|\Big) \\
    &\leq (2 + \varepsilon/3) \cdot (\varepsilon/3) \\
    &\leq 2 \varepsilon/3 + (\varepsilon/3)^2 \\
    &\leq \varepsilon, \label{eq:final_fidelity_error}
\end{align}
where in the first inequality  we used that $|\langle \psi | \widehat{\phi} \rangle| = |\langle \psi | (\Lambda_T \ket{\psi}) \rangle| \leq 1$. In the second inequality, we used Eq.~\eqref{eq:diff_inner_prods} and by  the reverse triangle inequality  $|a| - |b| \leq ||a| - |b|| \leq |a - b|$, we have
$$
|\langle \psi | \widehat{\phi} \rangle| = |\langle \psi | \widehat{\phi} \rangle - \langle \psi | \widetilde{\phi} \rangle + \langle \psi | \widetilde{\phi} \rangle| \leq |\langle \psi | \widehat{\phi} \rangle - \langle \psi | \widetilde{\phi} \rangle| + |\langle \psi | \widetilde{\phi} \rangle| \leq \varepsilon/3 + 1,
$$
where we again used Eq.~\eqref{eq:diff_inner_prods} and $|\langle \psi | \widehat{\phi} \rangle| = |\langle \psi | (\Lambda_T \ket{\psi}) \rangle| \leq 1$.

Finally, noting $|\langle \psi | \widetilde{\phi} \rangle|^2 = |\langle \psi | (\Lambda_T \ket{\psi}) \rangle|^2$ from Eq.~\eqref{eq:fidelity_betatilde}, we can upper bound Eq.~\eqref{eq:final_fidelity_error} by $\varepsilon$. This proves our desired lemma statement 
$$
|\langle \psi | \widehat{\phi} \rangle|^2 \geq |\langle \psi | (\Lambda_T \ket{\psi}) \rangle|^2 - \varepsilon.
$$
 Additionally, we can upper bound the $\ell_2$-norm of $\widehat{\beta} = (\widehat{\beta}_1,\ldots,\widehat{\beta}_k)$ as
\begin{align*}
\sum_{j=1}^k |\widehat{\beta}_i|^2 = \sum_{j=1}^k |\widehat{\beta}_i - \widetilde{\beta}_i + \widetilde{\beta_i}|^2 
&= \sum_{j=1}^k \left(|\widehat{\beta}_i - \widetilde{\beta}_i|^2 + 2|\widehat{\beta}_i - \widetilde{\beta}_i|\cdot|\widetilde{\beta_i}| +  |\widetilde{\beta_i}|^2\right) \\
&\leq \left(\sum_{i=1}^k |\widehat{\beta}_i - \widetilde{\beta}_i|\right)^2 + 2 \sum_{i=1}^k |\widehat{\beta}_i - \widetilde{\beta}_i| +  \sum_{i=1}^k |\widetilde{\beta_i}|^2 \\
&\leq \varepsilon^2/9 + 2 \varepsilon/3 +  |\langle \psi | (\Lambda_T \ket{\psi})|^2 \\
&\leq |\langle \psi | (\Lambda_T \ket{\psi})|^2 + \varepsilon,
\end{align*}
where we have used in the second line that $|\widetilde{\beta}_i| = |\beta_i| = |\langle \psi | \chi_i \rangle| \leq 1$ and $\sum_i |a_i|^2 \leq (\sum_i |a_i|)^2$. In the third line, we used Eq.~\eqref{eq:est_L1_norm_beta} and noted that $\sum_{i=1}^k |\widetilde{\beta_i}|^2 = \sum_{i=1}^k |\beta_i|^2 = |\langle \psi | (\Lambda_T \ket{\psi})|^2$.

The contribution to sample complexity is due to estimation of $|\langle \chi_j^R | \psi \rangle|$ and $|\langle \chi_j^I | \psi \rangle|$ up to error $\upsilon_j' = (\varepsilon \cdot \sqrt{\mu})/(36k)$, $|\langle \phi_j | \psi \rangle|$ up to error $\upsilon_j = (\varepsilon \cdot \sqrt{\mu})/(18k)$, and $|\langle \phi_j | \psi \rangle|$ up to error $\upsilon_1 = (\varepsilon \cdot \mu)/(63k)$. So by taking $O(k^2/(\varepsilon^2\cdot \mu^2) \log(k/\delta))$ copies of $\ket{\psi}$ and performing each $\SWAP$ test so that it succeeds with probability $1-O(\delta/k)$, so that after a union bound, the estimates in the previous analysis are met with overall probability $\geq 1-\delta$. 
The main contribution to time complexity is the preparation of the stabilizer states $\ket{\chi_j^R}$ and $\ket{\chi_j^I}$ which requires $O(n^2)$ gates each. The total time complexity~is
$$
O(k^3 n^2/(\varepsilon^2 \cdot \mu^2) \cdot \log(k/\delta)),
$$
hence proving the lemma statement.
\end{proof}

\subsection{Overall correctness and complexity} 
The proof of Theorem~\ref{thm:agnostic_boosting} regarding the correctness and complexity of the agnostic boosting protocol follows immediately from putting together our theorems regarding structure learning~(Theorem~\ref{thm:structure_learning}) and parameter learning~(Theorem~\ref{thm:parameter_learning}). 
\begin{proof}[Proof of Theorem~\ref{thm:agnostic_boosting}]
Let $\varepsilon_s,\varepsilon_p$ be parameters to be decided later. On input copies of $\ket{\psi}$, we use Theorem~\ref{thm:structure_learning} with error parameter instantiated as $\varepsilon_s$ and failure probability $\delta/2$ to learn a set of $\kappa \leq 4/(\varepsilon_s \cdot \eta(\varepsilon_s))$ parity states $\{\ket{\chi_i}\}_{i\in [\kappa]}$ such that $|\langle \chi_i | \psi \rangle|^2 \geq \varepsilon_s \eta(\varepsilon_s)/4$, and the state $\ket{\phi}:= \Lambda_T \ket{\psi}/\norm{\Lambda_T \ket{\psi}}$ (with $T=\spann(\{\ket{\chi_i}\}_{i\in [\kappa]})$) satisfies
$$
|\langle \psi | \phi \rangle|^2 \geq \opt - 2\sqrt{\varepsilon_s'},
$$
where $\varepsilon_s' = (3/2)^{1/\eta_2 + 1} \varepsilon_s$. We then utilize Theorem~\ref{thm:parameter_learning} instantiated with error parameter $\varepsilon_p = 2 \sqrt{\varepsilon_s'}$ and failure probability $\delta/2$, to learn a set of coefficients $\{\widehat{\beta}_i\}_{i \in [\kappa]}$ such that the state $\ket{\widehat{\phi}} = \sum_{i=1}\widehat{\beta}_i \ket{\chi_i}$ satisfies
$$
|\langle \widehat{\phi} | \psi \rangle|^2 \geq \opt - 2\varepsilon_p.
$$
Setting $\varepsilon_s = (2/3)^{1/\eta_2 + 1} \varepsilon^2/16$ and $\varepsilon_p = \varepsilon/2$ gives us the desired result. The overall sample complexity and time complexity is then evident from instantiating Theorems~\ref{thm:structure_learning},\ref{thm:parameter_learning}.
\end{proof}

\section{Learning algorithms}\label{sec:learning_algos}
In this section, we show how quantum agnostic boosting (Algorithm~\ref{algo:agnostic_boosting}) can be utilized for improper agnostic learning of decision trees, juntas and $\DNF$s. Finally, we will give a learning protocol of depth-$3$ circuits, based on boosting, in the uniform PAC model given quantum examples.

\subsection{Agnostic learning parities}
As a preliminary step to all subsequent learning algorithms, we first present a proper agnostic learning algorithm for parities. This algorithm is fairly simple and does not rely on boosting.

\begin{theorem}
\label{thm:parities}
Let $\opt \geq  \varepsilon > 0$. Suppose $\ket{\psi}$ is an unknown $n$-qubit state with unknown optimal fidelity $\calF_{\Par}(\ket{\psi}) = \opt$. Then, there is a 
$\widetilde{O}(n/\varepsilon^3\cdot  \log 1/ \delta)$-time proper agnostic learner that, with probability $\geq 1-\delta$, outputs $\ket{\phi} \in \Sh_{\Cc_{\Par}}$ such that $|\langle \phi | \psi \rangle|^2 \geq \opt - \varepsilon$.    
\end{theorem}

 To prove the theorem, we prove a lemma where we first assume that  $\calF_{\Cc_{\Par}}(\ket{\psi})$ is~known.
\begin{lemma}
\label{lem:AL_parities_tau}
Let $\tau \geq  \varepsilon > 0$. Suppose $\ket{\psi}$ is an unknown $n$-qubit state with fidelity $\calF_{\Cc_{\Par}}(\ket{\psi}) \geq \tau$. Then, there is a 
$\widetilde{O}(n/(\tau \cdot \varepsilon^2)\cdot  \log 1/\delta)$-time proper agnostic learner that, with probability $\geq 1-\delta$, outputs $\ket{\phi} \in \Sh_{\Cc_{\Par}}$ such that $|\langle \phi | \psi \rangle|^2 \geq \tau - \varepsilon$.    
\end{lemma}
\begin{proof}
Let $\ket{\chi_z}=\Had^{\otimes n}\ket{z}$. We will use the following algorithm for agnostic learning. 

\begin{myalgorithm}
\label{algo:AL_parities}
\begin{algorithm}[H]
    \caption{Agnostic learning of parity states}
    \setlength{\baselineskip}{1.5em} 
    \DontPrintSemicolon 
    \KwInput{Copies of $\ket{\psi}$, $\tau \in (0,1)$, $\varepsilon \in (0,1)$, $\delta  \in (0,1)$. } 
    \KwOutput{$\ket{\phi} \in \{\Sh_{\Par(\varepsilon)}\}$.}
    
    Measure $\Had^{\otimes n}\ket{\psi}$  in the computational basis $t=O(1/\varepsilon \log(2/\delta))$ many times, and collect the strings in $L=\{z_1,\ldots,z_t\}$.\\[1mm]
    Obtain $\varepsilon/2$-approximate estimates of $|\langle \chi_z|\psi\rangle|^2, \, \, \forall z \in L$ with probability $\geq 1 - \delta/(2|L|) $ using the $\SWAP$ test and $O(1/\varepsilon^2\log(|L|/\delta))$ copies of $\ket{\psi}$ for each $z \in L$. Let $\widehat{z}$ be the one that maximizes the fidelity. \\[1mm]
\Return  $\ket{\chi_{\widehat{z}}}\in \Sh_{\Par}$.
\end{algorithm}
\end{myalgorithm}
We now argue the correctness of the above protocol. Let $z^\star \in \{0,1\}^n$ be such that $\ket{\chi_z} \in \calS_{\Cc_{\Par}}$ maximizes fidelity with $\ket{\psi}$ i.e., $\calF_{\Cc_{\Par}}(\ket{\psi}) = |\langle \chi_z | \psi \rangle|^2 \geq \tau$. By measuring $\Had^{\otimes n}\ket{\psi}$ in the computational basis, the probability of obtaining a measurement outcome $z \in \{0,1\}^n$ (Step $(1)$ in the algorithm above) coinciding with $z^\star$ is
$$
\Pr[z=z^\star] = |\langle z^\star | \Had^{\otimes n}| \psi \rangle|^2 = |\langle \chi_{z^\star} | \psi \rangle|^2 \geq \tau.
$$
Repeating Step $(1)$, $O(1/\tau\cdot \log(1/\delta))$ many times, we ensure that $z^\star \in L$ with probability $1-\delta/2$. 

In Step $(2)$, for each distinct $z \in L$, we estimate the fildelity $\langle \chi_z | \psi \rangle|^2$ up to error $\varepsilon/2$ with success probability $\geq 1 -\delta/(2|L|)$ by using the $\SWAP$ test, which consumes $O(|L|/\varepsilon^2 \log(|L|/\delta))$ copies of $\ket{\psi}$. We then output $\ket{\chi_{\widehat{z}}}$ for the string $\widehat{z}$, that maximized the fidelity. By the guarantee of Step $(1)$ and a union bound, we will have $|\langle \chi_{\widehat{z}}|\psi \rangle|^2 \geq \tau - \varepsilon$ with probability $\geq 1 - \delta$. 

Step $(1)$ consumes $O(1/\tau \log(1/\delta))$ sample complexity and $O(n/\tau \log(1/\delta))$ time complexity. Step $(2)$ consumes $O(1/(\tau \cdot \varepsilon^2) \log(1/(\delta \cdot \tau)))$ sample complexity after noting that $|L| = O(1/\tau \log(1/\delta))$, and $O(n/(\tau \cdot \varepsilon^2) \log(1/(\delta \cdot \tau)))$ time complexity. The overall time complexity is thus $O(n/(\tau \cdot \varepsilon^2) \log(1/(\delta \cdot \tau)))$. This completes the proof of the lemma.
\end{proof}
The proof of Theorem~\ref{thm:parities} then follows.
\begin{proof}[Proof of Theorem~\ref{thm:parities}]
We instantiate Lemma~\ref{lem:AL_parities_tau} with $\tau$ set to be $\varepsilon$. Note that in either case of the unknown optimal fidelity $\opt \geq \varepsilon$ or $\opt < \varepsilon$, the outputted state $\ket{\phi}$ from Lemma~\ref{lem:AL_parities_tau} satisfies the guarantee of the theorem. This gives us the desired result.    
\end{proof}

\subsection{Agnostic learning decision trees }
Recall that we denote the class of decision trees of size $s$ as $\Cc_{\DT(s)}$ and define 
\begin{equation}
    \calF_{\Cc_{\DT(s)}}(\ket{\psi}) = \max_{f \in \Cc_{\DT(s)}} |\langle \phi_f | \psi \rangle|^2,
\end{equation}
where $\ket{\phi_f}$ is the phase state (Eq.~\eqref{eq:phase_state}) corresponding to $f$.

We have the following main theorem regarding the agnostic learnability of decision trees.
\begin{theorem}
\label{thm:agnostic_decision_trees}
Let $\varepsilon,\delta \in (0,1)$. Suppose $\ket{\psi}$ is an $n$-qubit state with unknown optimal fidelity $\calF_{\Cc_{\DT(s)}}(\ket{\psi}) = \opt$. Then, there is a quantum algorithm consuming $\poly(n, s, 1/\varepsilon, 1/\delta)$ copies of $\ket{\psi}$ and runs in $\poly(n, s, 1/\varepsilon, 1/\delta)$ time to output a state $\ket{\widehat{\phi}}$ such that
$$
|\langle \widehat{\phi} | \psi \rangle|^2 \geq \opt - \varepsilon,
$$
with probability $\geq 1-\delta$. Moreover, $\ket{\widehat{\phi}}$ can be expressed as $\ket{\widehat{\phi}} = \sum_{i=1}^{\kappa} \beta_i \ket{\chi_i}$ with $\beta \in \calB_\infty^k$ being coefficients corresponding to $\ket{\chi_i}$, which are parities, and $\kappa =\poly(s/\varepsilon)$.
\end{theorem}

To prove the above theorem, we will instantiate the quantum agnostic boosting algorithm (Algorithm~\ref{algo:agnostic_boosting}) and then use Theorem~\ref{thm:agnostic_boosting}. To use the boosting protocol, we need to define a weak agnostic learner $\calA_{\WAL}$ (Definition~\ref{def:WAL}) of $\DT(s)$. Towards obtaining a $\calA_{\WAL}$ for $\DT(s)$, we will require the following result from Kushilevitz and Mansour~\cite{kushilevitz1993decision}.
\begin{lemma}[\cite{kushilevitz1993decision}]\label{lem:ub_L1_norm_DT}
If $f\in \DT(s)$, then the $\ell_1$ norm of its Fourier coefficients satisfies
$
\sum_{\alpha} |\widehat{f}(\alpha)| \leq~s.
$
\end{lemma}

We can now show a weak agnostic learner for $\DT(s)$.
\begin{lemma}\label{WAL_DT}
Let $s \in \mathbb{N}$, $\tau,\delta \in (0,1)$ and $\varepsilon \in (0,\tau/s^2)$. Suppose $\ket{\psi}$ is an unknown $n$-qubit state satisfying $\calF_{\Cc_{\DT(s)}}(\ket{\psi}) \geq \tau$. Then, there is a quantum algorithm that outputs a parity state $\ket{\phi} \in \Cc_{\Par}$ such that
$$
|\langle \phi | \psi \rangle|^2 \geq \tau/s^2 - \varepsilon,
$$
with probability $\geq 1-\delta$. The algorithm consumes 
$\widetilde{O}(s^2/(\tau \cdot \varepsilon^2)\cdot  \log1/ \delta)$ copies of $\ket{\psi}$ and runs in $\widetilde{O}(n s^2/(\tau \cdot \varepsilon^2)\cdot  \log 1/ \delta)$ time.
\end{lemma}
\begin{proof}
Let $\ket{\phi_f}$ be the phase state corresponding to $f\in \DT(s)$ such that $|\langle \psi | \phi_f \rangle|^2 \geq \tau$. This then implies that
\begin{align*}
    \sqrt{\tau} \leq  |\langle \psi|\phi_f\rangle| =|\langle \psi|\textsf{Had}^{\otimes n} \cdot \textsf{Had}^{\otimes n} |\phi_f\rangle|
    &=\Big|\langle \psi|\textsf{Had}^{\otimes n} \cdot \sum_{\alpha \in \FF^n} \widehat{f}(\alpha)\ket{\alpha} \Big|\\
    &=\Big|\sum_{\alpha \in \{0,1\}^n} \widehat{f}(\alpha)\langle \psi|\textsf{Had}^{\otimes n}| \alpha \rangle\Big|\\
    &\leq \sum_{\alpha \in \{0,1\}^n} |\widehat{f}(\alpha)|\cdot |\langle \psi|\textsf{Had}^{\otimes n}| \alpha \rangle|\\
    &\leq \max_{\alpha \in \{0,1\}^n} |\langle \psi|\textsf{Had}^{\otimes n}| \alpha \rangle|\cdot \sum_{\alpha \in \{0,1\}^n} |\widehat{f}(\alpha)|\\
    &\leq \max_{\alpha \in \{0,1\}^n} |\langle \psi|\textsf{Had}^{\otimes n}| \alpha \rangle| \cdot s,
\end{align*}
where we have used the triangle inequality in the third line and Lemma~\ref{lem:ub_L1_norm_DT} in the last line. Noting that 
$
\textsf{Had}^{\otimes n} \ket{\alpha} = \ket{\chi_\alpha},
$
is a parity, we then have
$$
\calF_{\Cc_{\Par}}(\ket{\psi}) \geq \tau/s^2.
$$ 
Given that the fidelity of $\ket{\psi}$ with the class of parity states $\calS_{\Cc_{\Par}}$ is high, we can use Lemma~\ref{lem:AL_parities_tau} with error set to $\varepsilon$ and the lower bound on fidelity set to $\tau/s^2$ to learn a parity state $\ket{\phi}$ such that
$$
|\langle \psi | \phi \rangle|^2 \geq \tau/s^2 - \varepsilon,
$$
with probability $\geq 1 - \delta$. This consumes $O(s^2/(\tau \cdot \varepsilon^2) \log(s/(\tau \cdot \delta)))$ sample complexity and $O(n s^2/(\tau \cdot \varepsilon^2) \log(s/(\tau \cdot \delta)))$ time complexity. This completes the proof.
\end{proof}

The proof of Theorem~\ref{thm:agnostic_decision_trees} is then immediate from the instantiation of Theorem~\ref{thm:agnostic_boosting}.

\begin{proof}[Proof of Theorem~\ref{thm:agnostic_decision_trees}]
We instantiate the agnostic boosting algorithm (Algorithm~\ref{algo:agnostic_boosting}), as in Theorem~\ref{thm:agnostic_boosting}, using the weak agnostic learner $\calA_\WAL$ for $\Cc_{\DT(s)}$ from Lemma~\ref{WAL_DT}. The corresponding promise is $\eta(\tau) = \tau/(2s^2)$ with $\eta_1 = 1/(2s^2)$ and $\eta_2 = 1$ (as defined in Theorem~\ref{thm:agnostic_boosting} for $\varepsilon$ of Lemma~\ref{lem:AL_parities_tau} set to $\tau/(2s^2)$. The corresponding sample complexity is $S_{\WAL} = \widetilde{O}(s^6/\varepsilon_s^3 \cdot  \log1/ \delta)$ and $T_{\WAL} = \widetilde{O}(n s^6/\varepsilon_s^3 \cdot  \log 1/ \delta)$ for error instantiated as $\varepsilon_s$ in Algorithm~\ref{algo:agnostic_boosting} (Theorem~\ref{thm:structure_learning}). The output of Theorem~\ref{thm:agnostic_boosting} is then a strong (improper) agnostic learner consuming $\poly(s, 1/\varepsilon, \log(1/\delta))$ copies and $\poly(n, s, 1/\varepsilon, \log(1/\delta))$ time.
\end{proof}

\paragraph{Agnostic learning juntas.} We obtain an improper learning algorithm for $k$-juntas by instantiating Theorem~\ref{thm:agnostic_decision_trees} with size $s=2^k$ since any $k$-junta admits a decision tree of size $\mathcal{O}(2^k)$. This is summarized below.
\begin{corollary}
\label{corr:agnostic_juntas}
Let $\varepsilon \in (0,1)$ and $k \in \mathbb{N}$. Suppose $\ket{\psi}$ is an $n$-qubit state with unknown optimal fidelity $\calF_{\Jun(k)}(\ket{\psi}) = \opt$. Then, there is a quantum algorithm consuming $\poly(n, 2^k, 1/\varepsilon)$ copies of the state to output a state $\ket{\widehat{\phi}}$ such that
$$
|\langle \widehat{\phi} | \psi \rangle|^2 \geq \opt - \varepsilon,
$$
where $\kappa = \poly(2^k/\varepsilon)$ and $\ket{\widehat{\phi}} = \sum_{i=1}^{\kappa} \beta_i \ket{\chi_i}$ with $\beta \in \calB_\infty^k$ being coefficients for the parities $\ket{\chi_i}$.
\end{corollary}
We remark that in Appendix~\ref{app_sec:agnostic_juntas_no_boost}, we describe an (improper) agnostic learning algorithm of $k$-junta phase states that does not utilize boosting and is instead inspired by ideas from \cite{gopalan2008agnostically}.

\subsection{Agnostic learning DNFs}
We now show how the quantum agnostic boosting protocol (Algorithm~\ref{algo:agnostic_boosting}) can be utilized for (improper) agnostic learning of $s$-term DNFs, which we will denote by $\DNF(s)$. Particularly, we establish the following result.\footnote{We remark that one could have obtained a similar result for size-$s$ read-$k$ $\DNF$s with complexity $\poly(n,(s/\varepsilon)^{\log \log k\cdot \log 1/\varepsilon})$ using recent Fourier concentration bounds for these functions~\cite{lecomte2022sharper}.}
\begin{theorem}
\label{thm:agnostic_DNF}
Let $s \in \mathbb{N}$ and $\varepsilon,\delta \in (0,1)$. Suppose $\ket{\psi}$ is an $n$-qubit state with unknown optimal fidelity $\calF_{\DNF(s)}(\ket{\psi}) = \opt$. Then, there is a quantum algorithm that, given access to  $\poly((s/\varepsilon)^{\log \log (s/\varepsilon) \cdot \log(1/\varepsilon)}, 1/\varepsilon, \log(1/\delta))$ copies of $\ket{\psi}$ and running in time $\poly(n, (s/\varepsilon)^{\log \log (s/\varepsilon)\cdot \log (1/\varepsilon)}, 1/\varepsilon, \log(1/\delta))$ outputs, with probability at least $1-\delta$, a state $\ket{\widehat{\phi}}$ such that
$$
|\langle \widehat{\phi} | \psi \rangle|^2 \geq \opt - \varepsilon,
$$
where $\ket{\widehat{\phi}} = \sum_{i=1}^{\kappa} \beta_i \ket{\chi_i}$ with coefficients $\beta \in \calB_\infty^k$ for the parities $\ket{\chi_i}$ and $\kappa = \poly((s/\varepsilon)^{\log \log (s/\varepsilon) \cdot \log(1/\varepsilon)})$.
\end{theorem}

Recall from Theorem~\ref{thm:agnostic_boosting} that we require the input of a weak agnostic learner $\calA_{\WAL}$ (Definition~\ref{def:WAL}). To obtain $\calA_{\WAL}$ for $\Cc_{\DNF(s)}$, our starting point is the following theorem due to Mansour which shows that the Fourier spectrum of $\DNF$s concentrate. 
\begin{theorem}[\cite{mansour1992n}]
\label{lemma:l1normofdnf}
Let $s\in \mathbb{N}$, $\gamma \in (0,1)$. For $f \in \Cc_{\DNF(s)}$,  there exists $\mathcal{T}_\gamma\subseteq \F_2^n$ such that $|\mathcal{T_\gamma}|\leq (s/\gamma)^{O\big((\log \log s/\gamma)\cdot (\log 1/\gamma)\big)}$~and 
$$
\sum_{T\in \mathcal{T}_\gamma}\widehat{f}(T)^2\geq 1-\gamma.
$$
\end{theorem}

This allows us to give a weak agnostic learner for $\DNF(s)$.
\begin{lemma}
\label{lem:WAL_DNF}
Let $s\in \mathbb{N}$, $\tau,\delta \in (0,1]$ and $\varepsilon > 0$. Suppose $\ket{\psi}$ is an unknown $n$-qubit state satisfying
$$
\max_{f\in \Cc_{\DNF(s)}}|\langle \psi|\phi_f\rangle|^2 \geq \tau.
$$
Let $s^*=(s/\tau)^{O(\log \log (s/\tau) \cdot \log (1/\tau))}$.
Then, there exists an algorithm that with probability $\geq 1 - \delta$, outputs a parity state $\ket{\phi}\in \calS_{\Cc_{\Par}}$ satisfying
$$
|\langle \psi|\phi\rangle|^2\geq \tau/s^* - \varepsilon.
$$
The algorithm consumes $$
\widetilde{O}\Big(s^*/(\tau \cdot \varepsilon^2)\cdot  \log(1/\delta))\Big)
$$ 
copies of $\ket{\psi}$  in total and runs in  time 
$$
\widetilde{O}\Big(n\cdot s^*/(\tau \cdot \varepsilon^2) \cdot \log(1/\delta)\Big).
$$
\end{lemma}
\begin{proof}
Let $\ket{\phi_f}$ be the phase state corresponding to $f\in \DNF(s)$ such that $|\langle \psi | \phi_f \rangle|^2 \geq \tau$, and $\mathcal{T}\subseteq \F_2^n$ be as in Theorem~\ref{lemma:l1normofdnf}. Then we have that,  
\begin{align*}
\sqrt{\tau}\leq |\langle \psi|\Had^{\otimes n}\cdot \Had^{\otimes n}|\phi_f\rangle|&=|\langle \psi'|\sum_S \widehat{f}(S)|S\rangle|\\
&=|\sum_{S\in \mathcal{T}}\widehat{f}(S)\langle \psi'|S\rangle+\langle \psi'|\cdot \sum_{S\notin \mathcal{T}}\widehat{f}(S)|S\rangle|\\
&\leq \max_{S}|\langle \psi'|S\rangle| \cdot |\mathcal{T}|+ \Big|\langle \psi'|\cdot \sum_{S\notin \mathcal{T}}\widehat{f}(S)|S\rangle\Big|\\
&\leq  \max_{S}|\langle \psi'|S\rangle| \cdot |\mathcal{T}|+ \big \|\ \ket{\psi'} \big\| \cdot \Big\|\sum_{S\notin \mathcal{T}}\widehat{f}(S)|S\rangle \Big \|\\
&=  \max_{S}|\langle \psi'|S\rangle| \cdot |\mathcal{T}|+ \Big(\sum_{S\notin \mathcal{T}}\widehat{f}(S)^2 \Big)^{1/2}\\
&\leq  \max_{S}|\langle \psi'|S\rangle| \cdot |\mathcal{T}|+ \sqrt{\gamma},
\end{align*}
\noindent where we use the triangle inequality in the third line, Cauchy-Schwarz in the fourth, and Parseval's identity in combination with Theorem \ref{lemma:l1normofdnf} in the last one. This implies that,
$$
\calF_{\Par}(\ket{\psi}) \geq (\sqrt{\tau}-\sqrt{\gamma})^2\cdot (s/\gamma)^{-O(\log \log (s/\gamma) \log(1/\gamma)}\geq \tau\cdot \underbrace{(s/\tau)^{-O(\log \log (s/\tau) \cdot \log (1/\tau))}}_{:=1/s^*},
$$
where we let $\gamma= \tau/8$. So by applying our agnostic parity learning algorithm (Lemma~\ref{lem:AL_parities_tau}) with the lower bound there set to $\tau/s^*$, we can learn a state $\ket{\phi} \in \calS_{\Cc_{\Par}}$ such that
$$
|\langle \psi|\phi\rangle|^2\geq \tau/s^*-\varepsilon.
$$
with probability $\geq 1 - \delta$. The corresponding sample and complexity time complexity follows from Theorem~\ref{thm:parities} for $\varepsilon<\tau/s^*$, completing the proof. 
\end{proof}

The proof of Theorem~\ref{thm:agnostic_DNF} is then immediate from the instantiation of Theorem~\ref{thm:agnostic_boosting}.
\begin{proof}[Proof of Theorem~\ref{thm:agnostic_DNF}]
We instantiate the agnostic boosting algorithm (Algorithm~\ref{algo:agnostic_boosting}), as in Theorem~\ref{thm:agnostic_boosting}, using the weak agnostic learner $\calA_\WAL$ for $\Cc_{\DNF(s)}$ from Lemma~\ref{lem:WAL_DNF}. The corresponding promise is $\eta(\tau) = \tau/(2s^\star)$ with $s^*=(s/\tau)^{O(\log \log (s/\tau) \cdot \log (1/\tau))}$, $\eta_1 = 1/(2 s^\star)$ and $\eta_2 = 1$ (as defined in Theorem~\ref{thm:agnostic_boosting} for $\varepsilon$ of Lemma~\ref{lem:AL_parities_tau} set to $\tau/(2s^\star)$). Note that the function $\tau/s^*$ is an increasing function of $\tau$ (as required by the boosting algorithm). The corresponding sample complexity is $S_{\WAL} = \widetilde{O}((s/\varepsilon_s)^{\log \log (s/\varepsilon_s)\log (1/\varepsilon_s)}/\varepsilon_s^3 \cdot  \log1/ \delta)$ and $T_{\WAL} = \widetilde{O}(n (s/\varepsilon_s)^{\log \log (s/\varepsilon_s) \log (1/\varepsilon_s)}/\varepsilon_s^3 \cdot  \log 1/ \delta)$ for error instantiated as $\varepsilon_s$ in Algorithm~\ref{algo:agnostic_boosting} (Theorem~\ref{thm:structure_learning}). The output of Theorem~\ref{thm:agnostic_boosting} is then a strong (improper) agnostic learner consuming $\poly(s^*, 1/\varepsilon, \log(1/\delta))$ copies and $\poly(n, s^*, 1/\varepsilon, \log(1/\delta))$ time with $s^*=(s/\varepsilon)^{O(\log \log (s/\varepsilon) \cdot \log (1/\varepsilon))}$.
\end{proof}

\subsection{PAC learning depth-3 circuits}\label{sec:pac_learn_depth3}
In this section, we finally show how to quantum $\PAC$ learn depth-$3$ circuits. As mentioned earlier in the introduction (Section~\ref{sec:intro}), the current state-of-the-art algorithm for learning depth-$3$ circuits in the $\PAC$ model with only classical examples has a time complexity of $n^{O(\log^2 n)}$. Here, we show how to quantum $\PAC$ learn these circuits in $n^{O(\log n)}$ time. In order to prove our result, we will use the well-known discriminator lemma by Hajnal et al.~\cite{hajnal1993threshold} (which we reprove below specialized to our setting). Using our discriminator at each step of our boosting algorithm (in order to construct our weak learner), we are able to show that  depth-$3$ circuits are learnable in the quantum $\PAC$  model with the desired time complexity as stated above.

\subsubsection{Discriminator lemma}

We begin by formalizing the notion of a discriminator. 

\begin{definition}\label{def:discriminator}
Let $C:\{0,1\}^n\mapsto \{0,1\}$ be a circuit on $n$ bits, and let $A,B\subseteq \{0,1\}^n$ be disjoint sets. Let $D_A$ (resp $D_B$) be  distributions supported on $A$ (resp $B$). We say $C$ is a $\varepsilon$-discriminator for $A$ and $B$ over $D_A$ and $D_B$ if
    \begin{equation}
        |\Exp_{x\sim D_A}[C(x)]-\Exp_{x\sim D_B}[C(x)] |\geq \varepsilon.
    \end{equation}
\end{definition}

Next, we extend the discriminator of \cite{hajnal1993threshold} for circuits with a final threshold layer to the case where the distributions $A$ and $B$ are no longer uniform.

\begin{lemma}\label{lemma:th_disc}
Let $f=T_k^m(C_1,...,C_m)$ be a circuit on $n$ bits. There is a $C_i$ that is $(1/m)$-discriminator for $f^{-1}(1)$ and $f^{-1}(-1)$.
\end{lemma}
\begin{proof}
Let $D_A, D_B$ be distributions over $A=f^{-1}(1)$ and $B=f^{-1}(-1)$ respectively, then 
\begin{equation*}
\Exp_{x\sim D_A}\left[ \sum_{i=1}^m C_i(x)\right]\geq k,\text{ and } \Exp_{x\sim D_B}\left[ \sum_{i=1}^m  C_i(x)\right ]\leq k-1,
\end{equation*}
since by definition for every $x\in f^{-1}(1)$ (resp.~$x\in f^{-1}(-1)$), we know that $\sum_i C_i(x)\geq k$ (resp.~$\sum_i C_i(x)\leq k-1$). Therefore,
\begin{align*}
    1 \leq \Exp_{x\sim D_A}\left[ \sum_{i=1}^m  C_i(x)\right]-\Exp_{x\sim D_B}\left[ \sum_{i=1}^m  C_i(x)\right]=\sum_{i=1}^m  \left (\Exp_{x\sim D_A}[ C_i(x)]-\Exp_{x\sim D_B}[C_i(x)] \right )\leq m \max_i \left|\Exp_{x\sim D_A}[ C_i(x)]-\Exp_{x\sim D_B}[C_i(x)]\right |
\end{align*}
Therefore for a circuit $C_i$ with $1\leq i \leq m $ we have that,
\begin{equation*}
\Big|\Exp_{x\sim D_A}[ C_i(x)]-\Exp_{x\sim D_B}[C_i(x)]\Big| \geq 1/m,
\end{equation*}
proving the lemma statement. 
\end{proof}

Using this, we prove that at least one of the circuits feeding into the final threshold gate attains a nontrivial correlation with the gate’s output.

\begin{lemma}\label{lemma:discriminator}
Let $f=T_k^m(C_1,...,C_m)$ and $D$ be an arbitrary distribution. There exists $i\in [m]$ such that 
\begin{equation}
    |\langle f,C_i \rangle_D|=|\Exp_{x\sim D}[f(x)C_i(x)]|\geq \frac{1}{2m}.
\end{equation}
\end{lemma}
\begin{proof}
Let $A=f^{-1}(1)$ and $B=f^{-1}(-1)$. Let $\alpha=\sum_{x\in A} D(x)$. Without loss we can assume $\alpha\geq 1/2$, if not we can let $\alpha=\sum_{x\in B} D(x)$. For an arbitrary $C_j$, we have
\begin{align*}
    \Exp_{x\sim D}[f(x)C_j(x)]=\sum_{x\in A} C_j(x)\cdot D(x)-\sum_{x\in B} C_j(x)\cdot D(x)
    =\alpha \Big(\Exp_{x\sim D_A}\left[ C_j(x)\right]+\Exp_{x\sim D_B}\left[ C_j(x)\right]\Big)-\Exp_{x\sim D_B}\left[ C_j(x)\right]
\end{align*}
where $\alpha=\sum_{x\in A}D(x)$ and $D_A(x)=D(x)/\alpha, D_B(x)=D(x)/(1-\alpha)$. Now Lemma~\ref{lemma:th_disc} shows the existence of $C_i$ such that for the previous distributions $D_A$ and $D_B$ over $A,B$ respectively, we have
\begin{equation}
\label{eq:implicationoflemma}
\Big|\Exp_{x\sim D_A}[ C_i(x)]-\Exp_{x\sim D_B}[C_i(x)]\Big| \geq 1/m.
\end{equation}
Now, if $\alpha>1/2$, using Eq.~\eqref{eq:implicationoflemma}, we have that 
$\Exp_{x\sim D}[f(x)C_i(x)]\geq 1/(2m)$, proving the~lemma.
\end{proof}

\begin{fact}\label{fact:CNF_DNF_duality}
By De Morgan’s laws, the complement of any CNF with at most $k$ clauses is logically equivalent to a DNF with at most $k$ terms, and vice versa.
\end{fact}

\subsubsection{Learning algorithm}
In this part of the section, we prove the following main result regarding learnability of depth-$3$ circuits in the quantum \PAC~model.
\begin{theorem}\label{thm:pac_learn_depth3}
Let $s,m \in \mathbb{N}$. Suppose $f\in \textsf{TAC}^0_2$ where the fanin of the threshold gate is $m$ and the size of the circuit is $s$. Then, given quantum examples $\ket{\psi_f} = \frac{1}{\sqrt{2^n}}\sum_x f(x)\ket{x}$, with probability $\geq 1-\delta$,  can output $g : \FF^n \rightarrow \{-1,1\}$ such that $\Pr_{x\sim \mathcal{U}}[g(x)=f(x)]\geq 1-\varepsilon$.  The runtime of the algorithm is 
$$\poly(n,m,(s/\varepsilon)^{(\log s/\varepsilon\cdot\log \log s/\varepsilon) }, \log(1/\delta)).
$$
\end{theorem}

To prove this, we will show that applying improper agnostic learning of $\calS_{\Cc_{\DNF(s)}}$ when the input state $\ket{\psi_f}$ is promised to a phase state corresponding to a depth-$3$ circuit, accomplishes the task of state tomography. To that end, we will revisit the quantum agnostic boosting argument in Section~\ref{sec:agnostic_boosting} when the goal is to do well against the class $\Cc_{\DNF(s)}$. For the argument below, we now specialize our discussion to when the input is a phase state $\ket{\psi_f}$  since we are in the quantum $\PAC$ setting. Recall that in the agnostic boosting algorithm (Algorithm~\ref{algo:agnostic_boosting}), we had a running estimate $\ket{\widehat{\phi}^{(t)}}$ defined as the projection of the input state $\ket{\psi_f}$ on to the set of basis states $\{\ket{\chi_i}\}_{i \in [t]}$ learned so far:
$$
\ket{\widehat{\phi}^{(t)}} = \Lambda_{T(t)} \ket{\psi_f},
$$
where $T(t) = \mathrm{span}(\{\ket{\chi_i}\}_{i \in [t]})$. We  stop at the end of the $\kappa$th iteration as part of structure learning (see steps \ref{algo_step:est_fidelity} and \ref{algo_step:step_est_alphat}, Algorithm~\ref{algo:agnostic_boosting}) if 
\begin{equation}
\alpha_{t+1}^2 := \norm{(\mathbb{I} - \Lambda_{T(t)})\ket{\psi_f}}_2^2 < \varepsilon_s, \enspace \text{or} \enspace \calF_{\Cc}(\ket{\psi_{t+1}}) < \varepsilon_s,    
\end{equation}
where the residual state is $\ket{\psi_{t+1}} = (\mathbb{I} - \Lambda_{T(t)})\ket{\psi_f}/\alpha_{t+1}$ and $\varepsilon_s$ is an user-defined input error parameter to structure learning (see Theorem~\ref{thm:structure_learning}). However, we have not yet exploited the fact that the unknown input state $\ket{\psi_f}$ to Algorithm~\ref{algo:agnostic_boosting} is promised to be a phase state corresponding to a depth-$3$ circuit. In this case, we will now show that we stop at the end of the $\kappa$th iteration as part of structure learning only if both conditions are simultaneously satisfied. In particular, we have the following claim.
\begin{claim}\label{claim:stop_cond_depth3}
Consider the context of Theorem~\ref{thm:pac_learn_depth3}. Let $\Cc_{\DNF(s)}$ be the class of interest and $\calA_{\WAL}$ of Lemma~\ref{lem:WAL_DNF} be the weak agnostic learner. Suppose we apply agnostic boosting (Algorithm~\ref{algo:agnostic_boosting} and Theorem~\ref{thm:structure_learning}) to $\ket{\psi_f}$ then the following is true. If $|\alpha_t|^2 \geq \varepsilon$, then $\calF_{\Cc_{\DNF(s)}}(\ket{\psi_t}) \geq \varepsilon/4m^2$.
\end{claim}
\begin{proof}
Suppose as part of structure learning of the agnostic boosting algorithm (Algorithm~\ref{algo:agnostic_boosting}), we have carried out $t-1$ iterations so far and learned the parity states $\{\ket{\chi_i}\}_{i \in [t-1]}$. 
{Recall that until the $t-1$th iteration, the learner has obtained $\beta_1,\ldots,\beta_t$ and also the functions $\{\chi_i\}_{i\in [t]}$. Given explicit descriptions of these, the learner can implement the map
$$
O_{t-1}:\ket{x}\ket{0}\mapsto \ket{x}\ket{\textsf{ceil}\Big(\sum_{i\in [t-1]}\beta_i\chi_i(x)\Big)},
$$
where $\textsf{ceil}(a)=a$ if $a\in [-1,1]$,  $\textsf{ceil}(a)=-1$ if $a\leq -1$  and $\textsf{ceil}(a)=1$ if $a\geq 1$. The oracle can be implemented in time $O(t)$ by basic arithematic operations. Define $g(x)=\textsf{ceil}(\sum_i\beta_i\chi_i(x))$ and $\ket{\psi_g}=\frac{1}{\sqrt{2^n}}\sum_xg(x)\ket{x}$ (which can be prepared by making $O(1)$ calls to the oracle $O_{t-1})$. Defining $\Lambda_{g,t}=\ketbra{\psi_g}{\psi_g}$. Now measuring $\ket{\psi_f}$ in the basis $\{\Lambda_g,\mathbb{I}-\Lambda_g\}$, the post measurement state if we obtain the second outcome is given by  
$$
\ket{\psi_{t}}= \frac{1}{\alpha_t}\Big(\ket{\psi_f} - \Lambda_{g,t} \ket{\psi_f} \Big) = \frac{1}{\alpha_t}\Big(\ket{\psi_f} - \frac{1}{\sqrt{2^n}}\sum_x g(x) \ket{x}\Big),
$$}
where $T(t-1) = \mathrm{span}(\{\ket{\chi_i}\}_{i \in [t-1]})$, $\beta_i = \langle \chi_i | \psi_f \rangle$ (Eq.~\eqref{eq:projection_state_when_parities}) and $\alpha_t = \sqrt{1 - \sum_{i=1}^{t-1} |\beta_i|^2}$ (Fact~\ref{fact:projection}). Let us denote the states
$$
\ket{\psi_f} = \frac{1}{\sqrt{2^n}} \sum_{x \in \{0,1\}^n} f(x) \ket{x}, \enspace \ket{\chi_i} = \frac{1}{\sqrt{2^n}} \sum_{x \in \{0,1\}^n} \chi_i(x) \ket{x},
$$
where $f:\{0,1\}^n \rightarrow \{-1,1\}$ and $\chi_i(x) = (-1)^{\alpha_i \cdot x}$ for some $\alpha_i \in \{0,1\}^n, \,\, \forall i \in [t-1]$. Let us define the distribution
{$$
D_{f,t}(x)=\Big|f(x)-\textsf{ceil}(\sum_{i=1}^{t-1} \beta_i\chi_{i}(x))\Big|\cdot 2^{-n}\cdot \Delta^{-1},
$$}
where {$\Delta=\sum_x |f(x)-\textsf{ceil}(\sum_{i=1}^{t-1} \beta_i \chi_{i}(x))|\cdot 2^{-n}$} and $\chi_i$ is the parity function corresponding to the parity state $\ket{\chi_i}$. At this point, let us consider the inner product between the unknown $f$ and a $\DNF(s)$ formula $c$ (with $\ket{\phi_c} := \frac{1}{\sqrt{2^n}} \sum_x c(x) \ket{x}$ being the corresponding phase state),
{\begin{align*}
\Exp_{x\sim D_{f,t}}[f(x)c(x)]&=  \sum_{x}D_{f,t}(x)[f(x)\cdot c(x)]\\
 &=2^{-n}\cdot \Delta^{-1}\cdot\sum_{x} \left|f(x)-\textsf{ceil}(\sum_{i=1}^{t-1} \beta_i\chi_{i}(x))\right| \cdot f(x)\cdot c(x) \\
 &=2^{-n}\cdot \Delta^{-1}\cdot\sum_{x} \left(f(x)-\textsf{ceil}(\sum_{i=1}^{t-1} \beta_i\chi_{i}(x)) \right)\cdot c(x)\\
 &=\Delta^{-1}\cdot \alpha_t \cdot \langle \psi_{t}| \phi_c\rangle,
\end{align*}}
where the third equality used the fact that, for $a\in \{-1,1\}$ and $b\in [-1,1]$, we have that $|a-b|=a\cdot (a-b)$.\footnote{To see this: if $a=1$ then $a \geq b$ and $|a-b| = a - b = 1- b = a \cdot (a - b).$ If $a=-1,$ then $a \leq b$ and $|a-b| = b - a = b + 1 = a \cdot (a - b).$} One can then use Lemma~\ref{lemma:discriminator} to show that 
$$
|\langle \psi_t|\phi_c\rangle|=\Big|\frac{\Delta}{\alpha_t}\cdot  \Exp_{x\sim D_{f,t}}[f(x)c(x)]\Big|\geq \frac{\Delta}{m \alpha_t} \geq \frac{\sqrt{\varepsilon}}{m},
$$
where we used
{$$
\Delta = 2^{-n} \sum_x \left|f(x) - \textsf{ceil}(\sum_{i=1}^{t-1} \beta_i \chi_i(x)) \right| \geq 2^{-n} \sum_x \left|f(x) - \textsf{ceil}(\sum_{i=1}^{t-1} \beta_i \chi_i(x)) \right|^2/2 \geq \alpha_t^2/2,
$$
}
and $\alpha_t \geq \sqrt{\varepsilon}$. This in particular implies
$$
\calF_{\Cc}(\ket{\psi_{t}}) \geq |\langle \psi_t | \psi_c \rangle|^2 \geq \varepsilon/4m^2.
$$
This concludes the proof.
\end{proof}

\begin{proof}[Proof of Theorem~\ref{thm:pac_learn_depth3}]
Let $f$ be as in the theorem statement. We will employ the quantum agnostic boosting algorithm (Algorithm~\ref{algo:agnostic_boosting}) on input of copies of $\ket{\psi_f}$ against the class $\calC_{\DNF(s)}$ and use the weak agnostic learner $\calA_{\WAL}$ of Lemma~\ref{lem:WAL_DNF}.\footnote{Our notation here is similar to the one in Section~\ref{sec:agnostic_boosting}.} Suppose we are at the end of the $t$th iteration of the agnostic boosting algorithm. We first observe that by Claim~\ref{claim:stop_cond_depth3}, if $\alpha_{t+1}^2 \geq \varepsilon$, then $\calF_{\Cc_{\DNF(s)}}(\ket{\psi_{t+1}}) \geq \varepsilon/(4m^2)$. If we were to then use $\calA_{\WAL}$ of Lemma~\ref{lem:WAL_DNF}, we are guaranteed to learn a parity state $\ket{\chi_{t+1}}$ that satisfies $| \langle \chi_{t+1} | \psi_{t+1} \rangle|^2 \geq \varepsilon/(4m^2 s^\star)$ where we have denoted $s^\star = (s/\varepsilon)^{(\log \log s/\varepsilon) \cdot \log 1/\varepsilon}$. Note that the function $\varepsilon/s^*$ is an increasing function of $\varepsilon$ (as required by the boosting algorithm). Thus, we no longer need to check the fidelity of the residual state with $\calC_{\DNF(s)}$. We now present the simplified stage $1$ of the algorithm in Algorithm~\ref{algo:structure_learning_depth3}
\begin{myalgorithm}
\begin{algorithm}[H]
\label{algo:structure_learning_depth3}
\caption{Structure learning for depth-$3$ circuits}
\setlength{\baselineskip}{1.65em} 
\DontPrintSemicolon 
\KwInput{Parameters $s,m \in \mathbb{N}$, $\varepsilon \in (0,1)$, copies of  $\ket{\psi}$, weak learner $\calA_{\WAL}$ of Lemma~\ref{lem:WAL_DNF}}
\KwOutput{List of parities $L = \{\ket{\chi_i}\}_{i \in [\kappa]}$}

Set error parameter $\varepsilon_s = \varepsilon/9$. \\
Set $\ket{\psi_1}=\ket{\psi}$, $\alpha_1 = 1$, $L=\varnothing$. \\
Set parameter $\eta = \eta(\varepsilon_s)$ with $\eta(\cdot)$ being the promise of Lemma~\ref{lem:WAL_DNF}. (Theorem~\ref{thm:agnostic_boosting}). \\
Set $t_{\max} = 4/(\varepsilon_s \eta(\varepsilon_s))$, $\delta'=\delta/(3t_{\mathrm{max}})$, $\kappa=0$. \\[2mm]
\For{$t=1$ \KwTo $t_{\max}$}{\vspace{2mm}
    Run the weak agnostic learner $\calA_{\WAL}$ on $S_{\WAL}$ copies of $\ket{\psi_{t}}$ to learn a parity state~$\ket{\chi_t}$. \\
    Update $L \leftarrow L \cup \{\ket{\chi_t}\}$ and $\kappa \leftarrow \kappa + 1$. \\
    Set $\Lambda_{T(t)} = \sum_{i=1}^t \ket{\chi_i}\bra{\chi_i}$. \\ 
    Let $\widehat{\alpha}_{t+1}^2$ be an $\varepsilon_s/2$ approximation of $\alpha_{t+1}^2 := \norm{(\id-\Lambda_{g,t})\ket{\psi}}_2^2$ by measuring $\ket{\psi}$ in the basis~$\{\id-\Lambda_{g,t} , \Lambda_{g,t}\}$, $O(1/\varepsilon_s^2\log(1/\delta'))$ many times.\\
    \lIf{$\widehat{\alpha}_{t+1}^2 < \varepsilon_s$}{break loop. }
    Prepare $S_{\WAL}$ copies of $\ket{\psi_{t+1}}= (\id-\Lambda_{T(t)})\ket{\psi}/\alpha_{t+1}$ by measuring $O(S_{\WAL}/\varepsilon_s\log(1/\delta'))$ copies of $\ket{\psi}$ in the basis $\{\id-\Lambda_{T(t)}, \Lambda_{T(t)}\}$ and post-selecting for the first outcome.  \label{eq:stateupdate}\\
}\vspace{2mm}
\Return List of $\kappa$ parity states $L=\{\ket{\chi_i}\}_{i}$.
\end{algorithm}
\end{myalgorithm}

In Algorithm~\ref{algo:structure_learning_depth3}, we set the relevant error parameter $\varepsilon_s=\varepsilon/9$. The promise of the $\calA_{\WAL}$ here is then $\eta(\varepsilon_s) = \varepsilon/(36 m^2 s^\star)$ (with $\eta_1 = 1/(36 m^2 s^\star)$ and $\eta_2=1$ considering the definitions in Theorem~\ref{thm:agnostic_boosting}). By Claim~\ref{eq:stoppingcondition}, we are guaranteed that we stop after $\kappa \leq 4/(\varepsilon_s \eta(\varepsilon)) = O(m^2 s^\star/\varepsilon^2)$ many iterations in the structure learning algorithm. Moreover, at the end of the $\kappa$th iteration, we have the following decomposition of $\ket{\psi_f}$ from Theorem~\ref{thm:structure_learning}:
$$
\ket{\psi_f} = \Lambda_{T(\kappa)} \ket{\psi_f} + \alpha_{\kappa + 1} \ket{\psi_{\kappa + 1}},
$$
where $|\alpha_{\kappa + 1}|^2 < \varepsilon/4$, $\Lambda_{T(\kappa)} \ket{\psi}$ is the projection of $\ket{\psi}$ on to $T(\kappa) = \spann(\{\chi_t\}_{t \in [\kappa]}$~(Eq.~\eqref{eq:projection_state_when_parities}), and $\ket{\psi_{\kappa + 1}}$ is orthogonal to $\Lambda_{T(\kappa)} \ket{\psi_f}$. We then have from Fact~\ref{fact:projection}
$$
|\langle \psi_f | (\Lambda_{T(\kappa)} \ket{\psi_f})| = 1 - |\alpha_{\kappa+1}|^2 \geq 1 - \varepsilon/4 \implies |\langle \psi_f | (\Lambda_{T(\kappa)} \ket{\psi_f})|^2 \geq 1 - \varepsilon/2.
$$
Moreover, the state $\ket{\phi} := \Lambda_{T(\kappa)} \ket{\psi_f}/\norm{\Lambda_{T(\kappa)} \ket{\psi_f}}_2$ will satisfy
$$
|\langle \psi_f | \phi \rangle|^2 \geq 1 - \varepsilon/2,
$$
since $\norm{\Lambda_{T(\kappa)} \ket{\psi_f}}_2 \leq 1$.

At this point, we utilize Theorem~\ref{thm:parameter_learning} with the error parameter $\varepsilon_p$ set to $\varepsilon/2$, to learn $\{\beta_i\}_{i \in [\kappa]}$ corresponding to parity states $\{\ket{\chi_i}\}_{i \in [\kappa]}$ such that $\ket{\widehat{\phi}} := \sum_{i=1}^\kappa \beta_i \ket{\chi_i}$ satisfies
\begin{equation}\label{eq:good_job1}
|\langle \psi_f | \widehat{\phi} \rangle|^2 \geq 1 - \varepsilon.    
\end{equation}
We will now show that we have in fact accomplished the task of \PAC~learning. Let $h(x)=\sum_{i=1}^t \beta_i \chi_{i}(x)$. Eq.~\eqref{eq:good_job1} then implies that $(\Exp_x[f(x)h(x)])^2\geq 1-\varepsilon$. Let $g(x)=\sign(h(x))$ and  observe
\begin{align*}
    \Pr_{x\in \01^n}[f(x)\neq g(x)]&=\Exp_{x}[f(x)\neq \sign(h(x))]\\ &\leq \Exp_x[|f(x)-h(x)|^2]\\
    &=1+\Exp_x[h(x)^2]-2\Exp_x[f(x)h(x)]\\
    &\leq 1+1-2(1-\varepsilon/2)=\varepsilon,
\end{align*}
where the second inequality used that $\Exp_x[h(x)^2]=\sum_i\beta_i^2\leq 1$ and the assumption of this case that $\Exp_x[f(x)h(x)]\geq\sqrt{1-\varepsilon}\geq 1-\varepsilon/2$. Hence the \emph{Boolean} function $g$ is an $\varepsilon$-approximator for the unknown $f\in \textsf{TAC}^0_2$ and we are done. Finally note that the learning algorithm \emph{knows} explicitly $\ket{\widehat{\phi}}$, so it can output $g$ as well. 
 The main contribution to time complexity is due to running Algorithm~\ref{algo:structure_learning_depth3} and utilizing Theorem~\ref{thm:parameter_learning}. 
\end{proof}

\section{Relating distributional and state agnostic learning}
\label{sec:distributionagnostic}

A natural question left open by the results of the previous sections is whether there any connections between quantum distributional agnostic learning and quantum state agnostic learning. Although the input state in quantum distribution agnostic learning is more structured than in quantum state agnostic learning, observe that the output of the former model is more stringent than the latter model. So it is unclear if these two models are equivalent. In this section, we will show that for distributions $A=(\mathcal{U},\phi)$ where $\phi$ is ``well-bounded" (as we make clear in the statements below), then one can use quantum \emph{state} agnostic learning algorithms even when given as input $\ket{\psi_D}$, the input in distributional state agnostic learning.
To show this, we first prove the lemma below which will immediately imply the main theorem.

\begin{lemma}
\label{lem:goingbetweenmodelsmath}
Let $\gamma\in [0,1]$.  Let $\phi:\FF^n\rightarrow [-1,1]$ be such that $\Exp_x[\phi(x)^2]=\gamma$. Let $h:\FF^n\rightarrow \{0,1\}$.~Let 
    $$
\ket{\psi_1}=\frac{1}{\sqrt{2^n}}\sum_x \ket{x} \sum_{b\in \FF}(-1)^{b}\sqrt{\frac{1+(-1)^{b}\phi(x)}{2}}, \quad \ket{\psi_2}=\frac{1}{\sqrt{2^n}}\sum_x (-1)^{h(x)}\ket{x}.
    $$
    Then we have that
    $$
    \langle \psi_1|\psi_2\rangle \in \frac{1}{\sqrt{2}}\cdot  \Big[\Exp_x[(-1)^{h(x)}\phi(x)]-\gamma/2,\Exp_x[(-1)^{h(x)}\phi(x)]+\gamma/2\Big].
    $$
\end{lemma}

\begin{proof}
    The proof essentially involves writing out using Fact~\ref{fact:taylor} (which follows from the Taylor series of $\sqrt{1\pm x}$).
    \begin{align*}
    &\langle \psi_1|\psi_2\rangle\\
    &=\frac{1}{2^n}\sum_x (-1)^{h(x)}\Big(\sqrt{\frac{1+\phi(x)}{2}}-\sqrt{\frac{1-\phi(x)}{2}}\Big)\\
    &=\frac{1}{2^n}\sum_{x:h(x)=0} \Big(\sqrt{\frac{1+\phi(x)}{2}}-\sqrt{\frac{1-\phi(x)}{2}}\Big)-\frac{1}{2^n}\sum_{x:h(x)=1} \Big(\sqrt{\frac{1+\phi(x)}{2}}-\sqrt{\frac{1-\phi(x)}{2}}\Big)\\
    &\leq \frac{1}{\sqrt{2}}\frac{1}{2^n}\sum_{x:h(x)=0} (1+\phi(x)/2)-(1-\phi(x)/2-\phi(x)^2/2)-\frac{1}{\sqrt{2}}\frac{1}{2^n}\sum_{x:h(x)=1} (1+\phi(x)/2-\phi(x)^2/2)-(1-\phi(x)/2)\\
    &= \frac{1}{\sqrt{2}}\frac{1}{2^n}\sum_{x:h(x)=0} \phi(x)+\phi(x)^2/2-\frac{1}{\sqrt{2}}\frac{1}{2^n}\sum_{x:h(x)=1} \phi(x)-\phi(x)^2/2\\
    &= \frac{1}{\sqrt{2}}\frac{1}{2^n}\sum_{x}(-1)^{h(x)} \phi(x)+(-1)^{h(x)}\phi(x)^2/2\\
     &= \frac{1}{\sqrt{2}}\Exp_{x}[(-1)^{h(x)} \phi(x)]+\frac{1}{\sqrt{2}}\Exp_{x}[(-1)^{h(x)}\phi(x)^2/2]\\
     &\leq \frac{1}{\sqrt{2}}\Exp_{x}[(-1)^{h(x)} \phi(x)]+\frac{1}{\sqrt{2}}\Exp_{x}[\phi(x)^2/2]\\
      &\leq \frac{1}{\sqrt{2}}\Exp_{x}[(-1)^{h(x)} \phi(x)]+\gamma/(2\sqrt{2}).
\end{align*}
Similarly, one can also show a \emph{lower bound} using the same reasoning as above to get
$$
 \frac{1}{\sqrt{2}}\Exp_{x}[(-1)^{h(x)} \phi(x)]-\gamma/(2\sqrt{2}),
$$
hence proving the lemma statement.
\end{proof}

\begin{theorem}
Let $\alpha,\beta,\gamma \geq 0$ and $\Cc\subseteq \{c:\FF^n\rightarrow \FF\}$ be a concept class.

If there is a learning algorithm that satisfies the following: given copies of an unknown  $n$-qubit $\ket{\chi}$, outputs a phase state $\ket{\psi_h}=\frac{1}{\sqrt{2^n}}\sum_x(-1)^{h(x)}\ket{x}$ for $h:\FF^n\rightarrow \FF$, such that 
$$
\langle \chi|\psi_h\rangle\geq \alpha \cdot \opt-\beta,
$$
where $\opt=\max_{c\in \Cc}|\langle \chi|\psi_c\rangle|$. Suppose the sample and gate complexity is $T, G$ respectively.

Let $A=(D,\phi)$ be a distribution such that  $D$ is the uniform distribution and $\Exp_x[\phi(x)^2]=~\gamma$. Then, there is an algorithm that outputs a $h:\FF^n\rightarrow \{-1,1\}$ satisfying
$$
\Exp_{x}[h(x) \phi(x)]\geq \alpha\cdot \opt-(1+\alpha)/2\cdot \gamma-\beta
$$
using  $O(T/\gamma^2)$  copies of $\ket{\psi_D}$ and $O(Gn/\gamma^2)$ gates.
\end{theorem}
\begin{proof}
Consider the following algorithm: the learner obtains $\ket{\psi_D}$, applies Hadamard on the final qubit and measures, if it obtains $1$ carries on (with the resulting state $\ket{\psi'_D}$), else discards. We  first show that the probability of obtaining $1$ is $\gamma^2/2$. This is simple to analyze, first observe that after the Hadamard gate, $\ket{\psi_D}$  can be written as
$$
\frac{1}{\sqrt{2^n}}\sum_x \ket{x} \sum_{b\in \FF}\sqrt{\frac{1+(-1)^{b}\phi(x)}{2}}\ket{b}\rightarrow  \frac{1}{\sqrt{2^n}}\sum_x \ket{x} \sum_{b,c\in \FF}(-1)^{b\cdot c}\sqrt{\frac{1+(-1)^{b}\phi(x)}{2}}\ket{c}
$$
and the probability of obtaining $1$ is given by
$$
\frac{1}{2^n}\sum_x \Big(\sqrt{\frac{1+\phi(x)}{2}}-\sqrt{\frac{1-\phi(x)}{2}}\Big)^2=\Exp_x[1-\sqrt{1-\phi(x)^2}]\geq \Exp_x[1-(1-\phi(x)^2/2)]=\Exp_x[\phi(x)^2/2]=\gamma/2, 
$$
where the inequality used Fact~\ref{fact:taylor}. So the learning algorithm can obtain \emph{one} copy of $\ket{\psi'_D}$ using $O(1/\gamma^2)$ copies of $\ket{\psi_D}$ and furthermore the algorithm \emph{knows} when it has succeeded. So the algorithm can deterministically obtain $T$ many copies of $\ket{\psi'_D}$, which is given by
$$
\ket{\psi'_D}= \frac{1}{\sqrt{2^n}}\sum_x \ket{x} \sum_{b\in \FF}(-1)^{b}\sqrt{\frac{1+(-1)^{b}\phi(x)}{2}}.
$$
Now, observe that if we define 
$$
\opt=\max_{c\in \Cc}\Exp_x[\phi(x)c(x)],
$$
then by Lemma~\ref{lem:goingbetweenmodelsmath} we have that 
$$
\langle \psi_c|\psi'_D\rangle\geq \opt/\sqrt{2}-\gamma/(2\sqrt{2}).
$$
So, one can use the base algorithm that we assumed to exist in the lemma, that given copies of $\ket{\psi'_D}$, finds a \emph{phase state} $\ket{\psi_h}=\frac{1}{\sqrt{2^n}}\sum_x(-1)^{h(x)}\ket{x}$ that is $(\alpha,\beta)$-close to $\ket{ \psi'_D}$. We now use Lemma~\ref{lem:goingbetweenmodelsmath} again and one can conclude that if $\langle \psi'_D|\psi_h\rangle \geq \opt/\sqrt{2}-\gamma/(2\sqrt{2})$, then that implies that
$$
\Exp_{x}[(-1)^{h(x)} \phi(x)]/\sqrt{2}+\gamma/(2\sqrt{2})\geq \langle \psi_h|\psi'_D\rangle\geq \alpha(\opt/\sqrt{2}-\gamma/(2\sqrt{2}))-\beta,
$$
which implies
$$
\Exp_{x}[(-1)^{h(x)} \phi(x)]\geq \alpha\cdot \opt-(1+\alpha)/2\cdot \gamma-\beta
$$
showing the desired lemma inequality, by just outputting $\sign(h)$ as the output hypothesis.
\end{proof}

\bibliographystyle{alpha}
\bibliography{agnostic}

\appendix
\section{Further results}

\subsection{Bond dimension bounds for   phase states} \label{sect:bonddim}
As mentioned in the introduction, the recent work of Bakshi et al.~\cite{bakshi2024learning} gives an algorithm for agnostic learning matrix product states whose complexity scales polynomially in the bond dimension of those states. In particular, for an $\textsf{MPS}$ on $n$ qubits with bond dimension $r$, their learning algorithm has time complexity $\poly(n,r,1/\varepsilon)$. Since this suggests a natural learning algorithm to try for agnostic learning, it is worthwhile investigating the bond dimension for phase states corresponding to juntas and $\DNF$s which we bound below.

Below, we will use a couple of facts about bond dimension which we state as a blackbox. First,~\cite{verstraete2006matrix} showed that the bond dimension of a quantum state $\ket{\psi}$ is defined as follows
$$
\textsf{bond dimension}(\ket{\psi})=\max_{L,R} \{\textsf{Schmidt-rank}_{L|R}(\ket{\psi})\},
$$
where the maximum is over all possible contiguous cuts, call it $L,R$ for left and right and $\textsf{Schmidt-rank}_{L|R}(\ket{\psi})$ is defined as the Schmidt rank of the state when expressed as $\ket{\psi}=\sum_i c_i \ket{\chi_i}_L\otimes \ket{\gamma_i}_R$ with $\{\ket{\chi_i}_L\}_i,\{\ket{\gamma_i}_R\}_i$ being an orthogonal set of states.  Second, since we are dealing with phase states $\frac{1}{\sqrt{2^n}}\sum_x (-1)^{f(x)}\ket{x}$ in this work, it is not too hard to see the following: defining $M^f_S(a,b)=(-1)^{f(a,b)}$ where $a\in \FF^{S},b\in \FF^{\overline{S}}$, then 
$$
\textsf{bond dimension}(\ket{\psi_f})=\max_{S} \{\textsf{rank}(M^f_S)\}.
$$

\subsubsection*{A.1 Upper bound for junta states}

Let $f:\{0,1\}^{n}\to\{0,1\}$ be a $k$-junta, i.e., for every $x$, $f(x)$  depends only on a subset $S\subseteq [n]$ of size $|S|=k$. Let $\ket{\psi_f}$ be the junta state. Consider a bipartition of the qubits into $L|R$.
Let $S_L = S \cap L,S_R = S \cap R$, with 
$|S_L| = \ell,|S_R| = k - \ell$.
Across this cut, the amplitude tensor is a $2^{|L|}\times 2^{|R|}$ matrix.
Since $f$ does not depend on qubits outside $S$, those qubits contribute a 
rank-one outer-product factor. Hence the Schmidt rank across this cut equals the rank of the
$2^{\ell}\times 2^{k-\ell}$ sign matrix
$$
M_{a,b} = (-1)^{f(a_{S_L}, b_{S_R})},
\qquad a\in\{0,1\}^{S_L},\; b\in\{0,1\}^{S_R}.
$$
This rank is at most $2^{\min\{\ell,\,k-\ell\}}$, because the matrix has at most 
$2^{\min\{\ell,\,k-\ell\}}$ nonzero singular values.  
Maximizing over $\ell$ gives the bound
$$
\textsf{Schmidt-rank}_{L|R}(\lvert\psi\rangle) \;\leq\; 2^{\lfloor k/2 \rfloor}.
$$
This indicates that the (improper) agnostic learning algorithm in \cite{bakshi2024learning} is efficient for juntas when $k=O(\log n)$.

\subsubsection*{A.2 Lower bound for $\DNF$ states}
Here, we show that the bond dimension of $s$-term $\DNF$ states scales as $2^{s}$, in the worst case, thereby making the agnostic learning algorithm of $\textsf{MPS}$ by Bakshi et al.~\cite{bakshi2024learning} too inefficient for our setting.

\begin{lemma} \label{lem:dnfentrank}
    Let $f:\{0,1\}^{2s} \to \{0,1\},$ be a $\DNF$ acting on $2s$ bits partitioned into two halves $x,y \in \{0,1\}^s$. We define $f$ as follows
    $$
    f(x, y) = \bigvee_{i=1}^s (x_i \land y_i).
    $$
The phase state $\ket{\psi_f}$ when viewed as a bipartite entangled state across the $x,$ $y$ registers, has Schmidt rank $2^{s}.$  
\end{lemma}
\begin{proof}
    Let $M\in \mathbb{R}^{2^s\times 2^s}$ be the amplitude matrix $M_{x,y} \coloneqq (-1)^{f(x,y)}$ of the state $\ket{\psi_f}.$ Our goal is to show that $\mathsf{rank}(M)=2^s$. 
To this end, first we rewrite,
\begin{align}
        M_{x,y} &= 1 - 2 \left( \bigvee_{i=1}^s (x_i \land y_i) \right) =2\left (\prod_{i=1}^s (1 - (x_i \land y_i))\right )-1
      = 2 \left( \sum_{S \in [s]} (-1)^{|S|} \prod_{i \in S} x_i \prod_{j \in S}  y_j \right)-1.
\end{align}
In the first equality, we used the relation $(-1)^{f(x,y)}=1 - 2 f(x,y)$ together with the $\DNF$ representation of $f(x,y)$. In the second equality, we replaced the logical $\OR$ via $\bigvee_{i=1}^s z_i=1- \prod_{i=1}^s (1 - z_i)$. Next, we used the identity
 $$\prod_{i=1}^s (1 - z_i)= \sum_{S \in [s]} (-1)^{|S|} \prod_{i \in S} z_i.$$ 
The last equality follows from the identity $x_i \land y_i = x_i y_i$ and from expanding out the product over all $i \in [s].$
 Now, set $\alpha_{\emptyset} := 1$ and $\alpha_S := 2(-1)^{|S|}$ for all non-empty $S \subseteq [s]$. With this notation, we can equivalently express $M$ as
    $$
   M_{x,y} = \sum_{S \subseteq [s]} \alpha_S \left (\prod_{i \in S} x_i \right)  \left (\prod_{j \in S} y_j \right ) .
    $$
In this form, $M$ admits the singular value decomposition $M = U D V^T,$ where 
$$ U_{x,S} = \prod_{i \in S} x_i\ ; \ V_{y,S} = \prod_{j \in S} y_j \text{  and  }D = \diag(\alpha_S: S \subseteq [s]). 
$$
Up to a change of basis, $M$ is equivalent to a diagonal matrix with $2^s$ nonzero diagonal entries. Consequently, $\mathsf{rank}(M)=2^s,$ and therefore $\ket{\psi_f}$ has Schmidt rank $2^s$, proving the lemma.
\end{proof}

\begin{corollary}
    There exists an $s$-term $\DNF$ state having bond dimension $2^{s}.$ 
\end{corollary}
\begin{proof}
    The statement follows from noting that bond dimension can be defined as the maximum Schmidt rank over all bipartitions of the state and then combining this with Lemma~\ref{lem:dnfentrank} which gives an $s$-term $\DNF$ state (on $2s$ qubits) of Schmidt rank $2^{s}.$ 
\end{proof}

\subsection{Agnostic learning juntas without boosting}\label{app_sec:agnostic_juntas_no_boost}
In this section, we give an improper agnostic learner of $k$-junta phase states that does not utilize boosting (Theorem~\ref{thm:agnostic_boosting}). Particularly, we have the following result.
\begin{theorem}\label{thm:agnostic_juntas_no_boost}
Let $k \in \mathbb{N}$ and $\varepsilon,\delta \in (0,1)$. Suppose $\ket{\psi}$ is an unknown $n$-qubit state with unknown optimal fidelity $\calF_{\Cc_{\Jun(k)}}(\ket{\psi}) = \opt$. Then, there is a quantum algorithm that with probability $\geq 1-\delta$, outputs an $n$-qubit state $\ket{\widehat{\phi}}$ which can expressed as a linear combination of $O(k 2^{2k}/\varepsilon)$ parity states and satisfies
$$
|\langle \widehat{\phi} | \psi \rangle|^2 \geq \opt - \varepsilon.
$$
This algorithm does not use boosting (Algorithm~\ref{algo:agnostic_boosting}) and uses $\poly(k,2^k,1/\varepsilon,\log(1/\delta))$ sample complexity while running in $\poly(n,k,2^k,1/\varepsilon,\log(1/\delta))$ time.
\end{theorem}

To prove the above theorem, we require the following characterization of the unknown state $\ket{\psi}$ that is promised to have high fidelity with $\Cc_{\Jun(k)}$.
\begin{lemma}\label{lem:states_close_to_juntas}
Let $k \in \mathbb{N}$ and $\varepsilon \leq \opt \in (0,1]$. Suppose $\ket{\psi}$ is an arbitrary $n$-qubit state with unknown optimal fidelity $\calF_{\Cc_{\Jun(k)}}(\ket{\psi}) = \opt$ attained by $\ket{\phi_{f_S}} \in \calS_{\Cc_{\Jun(k)}}$ i.e., $|\langle \phi_{f_S} | \psi\rangle|^2 = \opt$. Then, there exists a collection of strings $A$ of size $|A| \leq 2^k$ satisfying the following properties:
\begin{enumerate}
    \item[$(i)$] $\sum_{x \in A} |\alpha_x|^2 \geq \opt - \varepsilon,$
    \item[$(ii)$] $\min_{x \in A} |\alpha_x|^2 \geq \max_{y \in S \setminus A} |\alpha_y|^2$,
    \item [$(iii)$] $\min_{x \in A} |\alpha_x|^2 \geq \varepsilon/2^k$,
    \item [$(iv)$] $\sum_{x \in \supp(\ket{\phi_{f_S}} \setminus A} |\alpha_x|^2 \leq \varepsilon$,
\end{enumerate}
where $\alpha_x = \langle x |\Had^{\otimes n} | \psi \rangle$ for all $x \in \{0,1\}^n$, and $\supp(\ket{\phi_{f_S}}):= \{x \in \{0,1\}^n : |\langle x | \Had^{\otimes n} | \phi_{f_S} \rangle| > 0\}$.
\end{lemma}
\begin{proof}
Let $\ket{\phi_{f_S}} \in \calS_{\Cc_{\Jun(k)}}$ be a $k$-junta phase state that maximizes fidelity with $\ket{\psi}$ i.e., $|\langle \phi_{f_S} | \psi\rangle|^2 = \opt$, and correspond to the $k$-junta Boolean function $f_S$ which depends only on bits in $S \subseteq [n]$ of size $|S| = \ell \leq k$. Let $L=2^\ell$. 

We will denote the $n$-bit strings corresponding to set $S$ as $B(S) = \{0,1\}_S^\ell \times 0_{\overline{S}}^{n-\ell}$, where the subscript $S$ indicates that the length-$\ell$ string should be placed in locations corresponding that in $S$ (assuming a fixed ordering) and similarly for subscript $\overline{S}$. We can then express the state $\ket{\phi_{f_S}}$ as 
\begin{equation}\label{eq:fourier_decomp_junta_func}
    \ket{\phi_{f_S}} = \frac{1}{\sqrt{2^n}} \sum_{x \in \{0,1\}^n} (-1)^{f_S(x)} \ket{x} = \frac{1}{\sqrt{2^n}} \sum_{x \in \{0,1\}^n} \sum_{\alpha \in B(S)} \widehat{f}_S(\alpha)(-1)^{\alpha \cdot x} \ket{x},
\end{equation}
where we have noted the Fourier decomposition of the $k$-junta function $f_S$ and denoted its coefficients as $\widehat{f}_S(\alpha)$.  Let $\ket{\phi_{f_S}'}=\textsf{Had}^{\otimes n} \ket{\phi_{f_S}}$ and $\ket{\psi'} = \textsf{Had}^{\otimes n} \ket{\psi}$. Using $|\langle \psi | \phi_{f_S} \rangle|^2 = |\langle \psi' | \phi_{f_S}' \rangle|^2$, we observe that
\begin{align}
   \opt = |\langle \psi' | \phi_{f_S}' \rangle|^2
   = \Big|\sum_{\alpha \in B(S)} \widehat{f}_S(\alpha) \langle \psi' |\alpha \rangle \Big|^2
   \leq \left[\sum_{\alpha \in B(S)} |\widehat{f}_S(\alpha)|^2 \right] \cdot \left[ \sum_{\alpha \in B(S)} |\langle \psi' |\alpha \rangle|^2 \right]
   =\sum_{z \in \{0,1\}^\ell} |\langle {z_S, 0_{\overline{S}}}|\psi'\rangle |^2,
\end{align}
where we have used Eq.~\eqref{eq:fourier_decomp_junta_func} in the second equality, and used Cauchy Schartz inequality in the third inequality as follows: let $u=\{\widehat{f}(\alpha)\}_{\alpha}$ and $v=\{\langle \psi'|\alpha\rangle\}_{\alpha}$, then $|\langle u,v\rangle|^2\leq \|u\|_2^2\cdot \|v\|_2^2$. We used Parseval's identity of $\sum_{\alpha \in B(S)}\widehat{f}(\alpha)^2=1$ in the final equality. Defining $\alpha_x= \langle x| \Had^{\otimes n} | \psi\rangle$, the above then implies
\begin{equation}
\label{eq:sum_square_amplitudes}
    \sum_{x \in \{0,1\}^S \times 0^{\overline{S}}} |\alpha_x|^2 \geq \opt.
\end{equation}
Let us define $\tau := \sum_{x \in \{0,1\}^S \times 0^{\overline{S}}} |\alpha_x|^2$. We now describe how to construct a subset $A \subseteq B(S)$ satisfying the properties indicated as part of the theorem. Consider all the elements $x \in B(S)$ and order them as $x_1,x_2,\ldots,x_{L}$ such that their corresponding amplitudes are non-increasing, i.e., $|\alpha_{x_1}|^2 \geq |\alpha_{x_2}|^2 \geq \cdots \geq |\alpha_{x_{L}}|^2$. Initializing $A = \emptyset$, we place elements into $A$ starting from $x_1$ and progressively going through $x_j$ for increasing $j$ in order. We stop when property $(i)$ is satisfied i.e.,
\begin{equation}\label{eq:construction_A}
    A = \{x_i\}_{i \in [m]} \quad \text{s.t.} \quad \sum_{x \in A} |\alpha_x|^2 \geq \tau - \varepsilon, \quad \text{and} \sum_{x \in A \setminus \{x_m\}} |\alpha_x|^2 < \tau - \varepsilon,
\end{equation}
where we have denoted $m = |A|$. The existence of such a set $A$ is guaranteed by Eq.~\eqref{eq:sum_square_amplitudes}. Property $(i)$ regarding $A$ is then true by construction as $\tau \geq \opt$~(Eq.~\eqref{eq:sum_square_amplitudes}). Property $(ii)$ is also true by construction and by noting that the elements with the top $m$ amplitudes from $S$ are placed in $A$. Furthermore, note that this set $A$ is the \emph{minimal} $A\subseteq B(S)$ for which items $(i,ii)$ hold true. 

In order to prove item $(iii)$, consider the set $V = (B(S) \setminus A) \cup \{x_m\}$. By construction, the element in the set $V$ with the maximum amplitude must be $x_m$ itself (or at least one of the elements with the same value). We now observe
\begin{equation}\label{eq:interim_lb_amplitudes_V}
\sum_{x \in V} |\alpha_x|^2 = \sum_{x \in B(S)} |\alpha_x|^2 - \sum_{x \in A \setminus \{x_m\}} |\alpha_x|^2 > \tau - (\tau - \varepsilon) = \varepsilon,
\end{equation}
where the first inequality used Eq.~\eqref{eq:sum_square_amplitudes} and Eq.~\eqref{eq:construction_A}. Noting that $|V| \leq |S| - 1 \leq 2^{\ell} \leq 2^k$ and $x_m = \arg \max_{x \in V} |\alpha_x|^2$, we also have
$$
\sum_{x \in V} |\alpha_x|^2 \leq 2^k |\alpha_{x_m}|^2.
$$
Combining the above with Eq.~\eqref{eq:interim_lb_amplitudes_V} then immediately implies
$$
\min_{x\in A} \,\, |\alpha_x|^2= |\alpha_{x_m}|^2 \geq \frac{\varepsilon}{2^k},
$$
which proves $(iii)$. For $(iv)$, we observe
$$
\tau = \sum_{x \in B(S)} |\alpha_x|^2 = \sum_{x \in B(S) \setminus A} |\alpha_x|^2 + \sum_{x \in A} |\alpha_x|^2 \geq \sum_{x \in B(S) \setminus A} |\alpha_x|^2 + \tau - \varepsilon \implies \sum_{x \in B(S) \setminus A} |\alpha_x|^2 \leq \varepsilon,
$$
where we have used Eq.~\eqref{eq:construction_A} in the third inequality. This completes the proof.
\end{proof}
We can then show that the projection of $\ket{\psi}$ on to the set $A$ from Claim~\ref{lem:states_close_to_juntas} solves the task of agnostic learning, which is formally stated below.
\begin{corollary}\label{cor:projection_state_close_to_junta}
Let $k \in \mathbb{N}$ and $\varepsilon \in (0,1)$. Suppose $\ket{\psi}$ is an unknown $n$-qubit state with unknown optimal fidelity $\calF_{\Cc_{\Jun(k)}}(\ket{\psi}) = \opt$ and let $A$ be the set from Lemma~\ref{lem:states_close_to_juntas} corresponding to error $\varepsilon/2^k$. Then, any set of parity states $\{\ket{\chi_y}\}_{y \in Y}$ corresponding to $Y \subseteq \{0,1\}^n$ such that $A \subseteq Y$ and $\ket{\phi}:= \Lambda_T \ket{\psi}/\norm{\Lambda_T \ket{\psi}}_2$ (with $T = \mathrm{span}(\{\ket{\phi_y}\}_{y \in Y})$) satisfies 
$$
|\langle \phi | \psi \rangle|^2 \geq \opt - 2\sqrt{\varepsilon},
$$
where $\Lambda_T \ket{\psi}$ is the projection of $\ket{\psi}$ onto $T$ as defined in Eq.~\eqref{eq:projector_parities}.
\end{corollary}
\begin{proof}
Let $\ket{\phi_{f_S}} \in \calS_{\Cc_{\Jun(k)}}$ be a $k$-junta phase state that maximizes fidelity with $\ket{\psi}$ i.e., $|\langle \phi_{f_S} | \psi\rangle|^2 = \opt$, and correspond to the $k$-junta Boolean function $f_S$ which depends only on bits in $S \subseteq [n]$ of size $|S| = \ell \leq k$. Let $L=2^\ell$. 

We will denote the $n$-bit strings corresponding to set $S$ as $B(S) = \{0,1\}_S^\ell \times 0_{\overline{S}}^{n-\ell}$, where the subscript $S$ indicates that the length-$\ell$ string should be placed in locations corresponding that in $S$ (assuming a fixed ordering) and similarly for subscript $\overline{S}$.

Consider the set $A$ from Lemma~\ref{lem:states_close_to_juntas} corresponding to error $\varepsilon/2^k$. Observe that for each string $\alpha \in A$, we can define a parity state $\ket{\chi_\alpha} := \frac{1}{\sqrt{2^n}} \sum_{x \in \{0,1\}^n} (-1)^{\alpha \cdot x} \ket{x}$ such that
$$
|\langle \chi_\alpha | \psi \rangle|^2 \geq \varepsilon/2^{2k},
$$
where this follows from Lemma~\ref{lem:states_close_to_juntas}$(iii)$. We are given that the set $Y$ in hand contains $A$. Define $T = \spann(\{\ket{\phi_y}\}_{y \in Y}$ and $\Lambda_T \ket{\psi}$ the corresponding projection of $\ket{\psi}$ on to $T$. Using Fact~\ref{fact:projection}, we can express $\ket{\psi}$ as
\begin{equation}\label{interim2_decomp_psi}
    \ket{\psi} = \Lambda_T \ket{\psi} + \alpha \ket{\phi^\perp},
\end{equation}
where $\ket{\phi^\perp}$ is orthogonal to $\Lambda_T \ket{\psi}$ and $\alpha = \sqrt{1 - \norm{\Lambda_T \ket{\psi}}_2^2}$. We then observe
\begin{align}
|\langle \phi_{f_S} | \psi \rangle| \leq |\langle \phi_{f_S} | (\Lambda_T \ket{\psi})\rangle| + |\alpha \langle \phi_{f_S} | \phi^\perp \rangle| 
&= |\langle \phi_{f_S} | (\Lambda_T \ket{\psi})\rangle| + \Big| \sum_{x \in B(S) \setminus A} \alpha_x \widehat{f}_S(x) \Big| \\
&\leq |\langle \phi_{f_S} | (\Lambda_T \ket{\psi})\rangle| + \sum_{x \in B(S) \setminus A} |\alpha_x| \\
&\leq |\langle \phi_{f_S} | (\Lambda_T \ket{\psi})\rangle| + \sqrt{\varepsilon},
\end{align}
where we have used the decomposition of $\ket{\psi}$ from Eq.~\eqref{interim2_decomp_psi} in the first inequality. In the second inequality, we used that $\ket{\phi_{f_S}}$ is supported on computational basis states over $B(S)$ whereas $\ket{\phi^\perp}$ is supported over computational basis states not in $Y$. The second line follows from applying triange inequality and then noting that $A \subseteq Y \implies (B(S) \setminus Y) \subseteq (B(S) \setminus A)$. The third line follows from Lemma~\ref{lem:states_close_to_juntas}$(iii)$ which implies $|\alpha_x| \leq \sqrt{\varepsilon}/2^k, \,\,\forall x \in B(S) \setminus A$ and using $|B(S)| \leq 2^k$. This implies that 
$$
|\langle \phi_{f_S} | \psi \rangle| - |\langle \phi_{f_S} | (\Lambda_T \psi)\rangle| \leq \sqrt{\varepsilon}.
$$
We then observe
\begin{align}\label{eq:projection_high_fidelity_with_junta}
&|\langle \phi_{f_S} | \psi \rangle|^2 - |\langle \phi_{f_S} | (\Lambda_T \psi)\rangle|^2 = \Big(|\langle \phi_{f_S} | \psi \rangle| + |\langle \phi_{f_S} | (\Lambda_T \psi)\rangle|\Big) \Big(|\langle \phi_{f_S} | \psi \rangle| - |\langle \phi_{f_S} | (\Lambda_T \psi)\rangle|\Big) \leq 2 \sqrt{\varepsilon}, \\
&\implies |\langle \phi_{f_S} | (\Lambda_T \psi)\rangle|^2 \geq |\langle \phi_{f_S} | \psi \rangle|^2 - 2\sqrt{\varepsilon} = \opt - 2\sqrt{\varepsilon},
\end{align}
where we used $|\langle \phi_{f_S} | (\Lambda_T \psi)\rangle|, |\langle \phi_{f_S} | \psi \rangle| \leq 1$ in the first line and the fact that $|\langle \phi_{f_S} | \psi \rangle|^2 = \opt$ in the implication. Let us now define the state $\ket{\widehat{\phi}} := \Lambda_T \ket{\psi}/\norm{\Lambda_T \ket{\psi}}$. We then obtain that
\begin{align*}
|\langle \widehat{\phi}|\psi\rangle|^2 = \left|\langle \widehat{\phi}|(\Lambda_T \ket{\psi})\rangle + \alpha \langle \widehat{\phi}|\phi^\perp\rangle \right|^2
= |\langle \widehat{\phi}| (\Lambda_T \ket{\psi})|^2 
&= |\langle \widehat{\phi}|\widehat{\phi}\rangle| \cdot \norm{(\Lambda_T \ket{\psi})}_2^2 \\
&\geq |\langle \phi_{f_S} |\widehat{\phi} \rangle|^2 \cdot \norm{(\Lambda_T \ket{\psi})}_2^2 \\
&=\frac{|\langle \phi_{f_S} | (\Lambda_T \ket{\psi})\rangle|^2}{\norm{(\Lambda_T \ket{\psi})}_2^2} \cdot \norm{(\Lambda_T \ket{\psi})}_2^2 \\
&=|\langle \widehat{\phi}^{(\kappa)} |\phi_{f_S}\rangle|^2 \\
&\geq \opt - 2 \sqrt{\varepsilon},
\end{align*}
where the first equality used Eq.~\eqref{interim2_decomp_psi}, second equality used that $\ket{\widehat{\phi}},\ket{\phi^\perp}$ are orthogonal, third equality used the definition of $\ket{\widehat{\phi}}$, the inequality works for \emph{every} phase state $\ket{\phi_f}$, in particular the phase corresponding to $f_S \in \calS_{\Cc_{\Jun(k)}}$ that satisfies $|\langle \psi | \phi_{f_S} \rangle|^2 = \opt$, and the last inequality used Eq.~\eqref{eq:projection_high_fidelity_with_junta}. This completes the proof.
\end{proof}

Note that while Corollary~\ref{cor:projection_state_close_to_junta} could have also been obtained by instantiating Theorem~\ref{thm:structure_learning}, here the parity states are not obtained by boosting but rather using the characterization from Lemma~\ref{lem:states_close_to_juntas}. We now prove Theorem~\ref{thm:agnostic_juntas_no_boost} and describe its corresponding algorithm.
\begin{proof}[Proof of Theorem~\ref{thm:agnostic_juntas_no_boost}]
Let $\ket{\psi'} = \Had^{\otimes n} \ket{\psi}$ and $\alpha_x = \langle x | \psi' \rangle$ i.e., $\alpha_x$ is the amplitude in $\ket{\psi'}$ corresponding to the computational basis state $\ket{x}$. Let $\varepsilon_1 \in (0,1)$ be an error parameter to be fixed later. We will use the following learning algorithm:
\begin{enumerate}[$(1)$]
    \item Measure $\ket{\psi'}$ in the computational basis $M=O(2^k/\varepsilon_1 \cdot (k + \log(1/\delta)))$ many times to obtain a set of $M$ many strings $Y=\{y_i\}_{i \in [M]}$.
    \item Let $\varepsilon_2 = \varepsilon_1/2^{2k}$. Obtain an estimate $|\widehat{\alpha}_y|^2$ of $|\alpha_y|^2$ up to error $\varepsilon_2/4$, with probability $\geq 1 - \delta/(3|Y|)$ using the $\SWAP$ test between $\ket{\psi'}$ and $\ket{y}$ for all $y \in Y$. Let $Y'$ be the subset of strings in $Y$ such that $|\widehat{\alpha}_y|^2 \geq 3\varepsilon_2/4$ and denote $\kappa = |Y'|$.
    \item Consider the set of parity states $\{ \ket{\chi_y} : \ket{\chi_y} = 2^{-n/2} \sum_x (-1)^{y \cdot x} \ket{x} , \,\, \forall y \in Y' \}$. Use the parameter learning algorithm of Theorem~\ref{thm:parameter_learning} to learn coefficients $\{\widehat{\beta}_y\}_{y \in Y'}$ corresponding to $\{\ket{\chi_y}\}_{y \in Y'}$ with error parameter set as $\varepsilon_1$.
    \item Output the state $\ket{\widehat{\phi}} = \sum_{y \in Y'} \widehat{\beta}_y \ket{\chi_y}$. 
\end{enumerate}
Now, we give the correctness of the above algorithm and that it satisfies the guarantees of the stated theorem. Using Lemma~\ref{lem:states_close_to_juntas} instantiated with error $\varepsilon_1/2^k$, we know there exists a set $A$ of size $|A| \leq 2^k$ such that $|\alpha_x|^2 \geq \varepsilon_1/2^{2k}, \,\, \forall x \in A$ and from Corollary~\ref{cor:projection_state_close_to_junta}, we know that there exists a state corresponding to any set $Y$ containing $A$ which would accomplish agnostic learning (i.e., have fidelity promise $\geq \opt - \varepsilon_1$).

In Step $(1)$, by measuring $O(2^{2k}/\varepsilon_1(k + \log(1/\delta)))$ many times, we ensure that $A \subseteq Y$ with probability $\geq 1-\delta/3$. This can be observed by noting that $\min_{x \in A} |\alpha_x|^2 \geq \varepsilon_1/2^{2k}$ and thus for any fixed $a \in A$, we have
$$
\Pr[a \notin Y] = (1 - |\alpha_a|^2)^m \leq \exp(-M |\alpha_a|^2) \leq \exp(-M \varepsilon_1/2^{2k}),
$$
where $M$ is the number of times we measure $\ket{\psi'}$ in the computational basis. Union bound over the $O(2^k)$ elements of $A$ gives us that
$$
\Pr[A \not \subseteq Y] \leq 2^k \exp(-M \varepsilon_1/2^{2k}).
$$
To make this at most $\delta/3$, it suffices that $M = O(2^{2k}/\varepsilon_1 (k + \log(1/\delta)))$.

In Step $(2)$ of the above procedure, we remove all the strings from $Y$ that have low amplitudes by obtaining estimates $|\widehat{\alpha}_y|^2$ of $|\alpha_y|^2$ up to error $\varepsilon_2/4$ (with $\varepsilon_2 = \varepsilon_1/2^{2k}$) with probability $\geq 1 - \delta/(3|Y|)$. Noting that $|Y| \leq O(k2^{2k}/\varepsilon_1)$, this consumes $\widetilde{O}(2^{6k}/\varepsilon_1^3 \cdot(k + \log(1/\delta)))$ sample complexity overall and $\widetilde{O}(n 2^{6k}/\varepsilon_1^3 \cdot(k + \log(1/\delta)))$ time. Taking an union bound over $O(k2^{2k}/\varepsilon_1)$ elements of $Y$, we ensure that with probability $\geq 1 - \delta/3$ that
$$
\Big| |\widehat{\alpha}_y|^2 - |\alpha_y|^2 \Big| \leq \varepsilon_2/4, \,\, \forall y \in Y.
$$
This in particular implies that for all $y \in A$ which are guaranteed to have $|\alpha_y|^2 \geq \varepsilon_1/2^{2k} = \varepsilon_2$, their estimates satisfy $|\widehat{\alpha}_y|^2 \geq 3\varepsilon_2/4$. Thus, even after removing all elements in $y \in Y$ with $|\widehat{\alpha}_y|^2 < 3\varepsilon_2/4$ to create the set $Y'$, we ensure that $A \subseteq Y'$ and all strings $y \in Y'$ satisfy $|\alpha_y|^2 \geq \varepsilon_2/2$.

In Step $(3)$, we consider the set of parity states $\{ \ket{\chi_y}\}_{y \in Y'}$ (with $\kappa = |Y'| \leq O(k2^{2k}/\varepsilon_1)$). Applying Corollary~\ref{cor:projection_state_close_to_junta}, we are promised that the state $\ket{\phi}:= \Lambda_T \ket{\psi}/\norm{\Lambda_T\ket{\psi}}_2$ satisfies
$$
|\langle \psi | \phi \rangle|^2 \geq \opt - 2\sqrt{\varepsilon_1},
$$
where $T = \mathrm{span}(\{\ket{\chi_y}\}_{y \in Y'})$ and $\Lambda_T \ket{\psi}$ is the projection of $\ket{\psi}$ onto $T$ as defined in Eq.~\eqref{eq:projector_parities}. Applying Theorem~\ref{thm:parameter_learning} with failure probability set to $\delta/3$, the error $\varepsilon_p$ set  to $2\sqrt{\varepsilon_1}$ and $\mu = \varepsilon_1/2^{2k+1}$ (as $|\langle \chi_y | \psi \rangle|^2 \geq \varepsilon_2/2 = \varepsilon_1/2^{2k+1} ,\,\,\forall y \in Y'$), we learn coefficients corresponding to $\{\widehat{\beta}_y\}_{y \in Y'}$ corresponding to the parity states $\{ \ket{\chi_y} \}_{y \in Y'}$ such that $\ket{\widehat{\phi}} = \sum_{y \in Y'} \widehat{\beta}_y \ket{\chi_y}$ satisfies
$$
|\langle \widehat{\phi} | \psi \rangle|^2 \geq \opt - 4\sqrt{\varepsilon_1}.
$$
This is ensured with an overall success probability $\geq 1-\delta$. Setting $\varepsilon_1 = \varepsilon^2/16$ gives us the desired result. Using Theorem~\ref{thm:parameter_learning} consumes 
$$
\text{sample complexity: }O\left(\frac{k 2^{14k}}{\varepsilon^{15}} \cdot \Big(k + \log\frac{k}{\delta \cdot \varepsilon}\Big)\right), \enspace 
\text{time complexity: }  O\left(\frac{n^2 k 2^{14k}}{\varepsilon^{15}} \cdot \Big(k + \log\frac{k}{\delta \cdot \varepsilon}\Big)\right),
$$
since $\kappa = O(k2^{2k}/\varepsilon^2)$, $\mu = O(\varepsilon^2/2^{2k})$, and $\varepsilon_p=\varepsilon/2$ of the theorem statement. The overall sample and time complexity of the algorithm is then due to the above.
\end{proof}
\end{document}